\pdfoutput = 1
\documentclass[12pt]{scrartcl}
\makeatletter
\DeclareOldFontCommand{\rm}{\normalfont\rmfamily}{\mathrm}
\DeclareOldFontCommand{\sf}{\normalfont\sffamily}{\mathsf}
\DeclareOldFontCommand{\tt}{\normalfont\ttfamily}{\mathtt}
\DeclareOldFontCommand{\bf}{\normalfont\bfseries}{\mathbf}
\DeclareOldFontCommand{\it}{\normalfont\itshape}{\mathit}
\DeclareOldFontCommand{\sl}{\normalfont\slshape}{\@nomath\sl}
\DeclareOldFontCommand{\sc}{\normalfont\scshape}{\@nomath\sc}
\makeatother

\usepackage{hyphenat,float}
\usepackage[T1]{fontenc}         
\usepackage[utf8]{inputenc}
\usepackage[british]{babel}

\areaset{158mm}{238mm}
\usepackage{setspace}
\usepackage[all]{xypic}
\usepackage{scrlayer-scrpage}
\usepackage{tikz}
\usetikzlibrary{cd}

\usepackage{amsmath}
\usepackage{mathtools}
\usepackage{mathrsfs}
\usepackage{amsthm}
\usepackage{array,multirow,graphicx,booktabs}

\usepackage{amssymb}
\usepackage[mathscr]{eucal}
\usepackage{atbegshi}
\AtBeginShipoutFirst{}
\usepackage[unicode=true,linktoc=all]{hyperref}
\usepackage{cite}
\usepackage{bm}
\usepackage{slashed}

\usepackage{tikz,pgf}
\usetikzlibrary{shapes}
\usetikzlibrary{calc}
\usetikzlibrary{decorations.pathmorphing}
\usetikzlibrary{decorations.pathreplacing,shapes.misc}
\usetikzlibrary{positioning}
\usetikzlibrary{decorations.markings}
\tikzset{->-/.style={decoration={
  markings,
  mark=at position .5 with {\arrow{>}}},postaction={decorate}}}
\tikzset{-<-/.style={decoration={
  markings,
  mark=at position .5 with {\arrow{<}}},postaction={decorate}}}

\def\bR {\mathbb{R}}
\def\bZ {\mathbb{Z}}

\def\bC {\mathbb{C}}

\newcommand{\cA}{{\mathcal A}}

\newcommand{\cN}{{\mathcal N}}

\newcommand{\cP}{{\mathcal P}}
\newcommand{\cQ}{{\mathcal Q}}
\newcommand{\cS}{{\mathcal S}}
\newcommand{\cV}{{\mathcal V}}

\newcommand{\cM}{{\mathcal M}}

\def\bea{\begin{eqnarray}}
\def\eea{\end{eqnarray}}
\def\ie{\begin{equation}\begin{aligned}}
\def\fe{\end{aligned}\end{equation}}

\newcommand{\A}{{\alpha}}

\newcommand{\D}{{\delta}}

\newcommand{\ra}{\rangle}

    \newcommand{\secref}[1]{\S\ref{#1}}
    \newcommand{\figref}[1]{Figure~\ref{#1}}
    \newcommand{\appref}[1]{Appendix~\ref{#1}}
    
    \def\ie{\begin{equation}\begin{aligned}}
    \def\fe{\end{aligned}\end{equation}}
     
    \def\ria{\rightarrow}
    \newcommand{\E}{{\epsilon}}
    \newcommand{\pa}{\partial}

    \newcommand{\ld}{\lambda}
    \newcommand{\cI}{{\mathcal I}}
    
\theoremstyle{plain}
  
  \newtheorem{proposition}{Proposition}
  \newtheorem{lemma}{Lemma}
  \newtheorem{corollary}{Corollary}

\theoremstyle{definition}

\theoremstyle{remark}
  
  \newtheorem{remark}{Remark}

\newcommand{\be}{\begin{equation}}
\newcommand{\ee}{\end{equation}}

\newcommand{\GX}{{\mathcal X}}
\newcommand{\TX}{\mathcal T_\GX}
\newcommand{\TNk}{\mathbf{TN}_k}
\newcommand{\ALEk}{\mathbf{ALE}_k}

\numberwithin{equation}{section}

\title{\vspace{-2cm}
Evidence for an Algebra \\ of $\boldsymbol{G_2}$ Instantons \\ [0.5cm]}

\author{Michele Del Zotto$^{\dagger}$, Jihwan Oh$^{\ddagger}$, and Yehao Zhou$^{\ast}$
\\[1cm]
	\small\slshape$^\dagger$ Mathematics Institute, Uppsala University,  \\[-0.2cm] 
	\small\slshape Box 480, SE-75106 Uppsala, Sweden\\
	\small\slshape$^\dagger$ Department of Physics and Astronomy, Uppsala University,  \\[-0.2cm] 
	\small\slshape Box 516, SE-75120 Uppsala, Sweden\\
	\small\slshape$^\ddagger$ Mathematical Institute, University of Oxford,  \\[-0.2cm] 
	\small\slshape Woodstock Road, Oxford, OX2 6GG, United Kingdom\\
	\small\slshape$^\ast$ Perimeter Institute for Theoretical Physics,\\[-0.2cm] 
	\small\slshape 31 Caroline St. N., Waterloo, ON N2L 2Y5, Canada\\
	}
\date{}

\begin{document}

\maketitle

\paragraph{\hspace{.9cm}\large{Abstract}}
\vspace{-1cm}
\begin{abstract}

\noindent In this short note, we present some evidence towards the existence of an algebra of BPS $G_2$ instantons. These are instantonic configurations that govern the partition functions of 7d SYM theories on local $G_2$ holonomy manifolds $\mathcal X$. To shed light on such structure, we begin investigating the relation with parent 4d $\mathcal N=1$ theories obtained by geometric engineering M-theory on $\mathcal X$. The main point of this paper is to substantiate the following dream: the holomorphic sector of such theories on multi-centered Taub-NUT spaces gives rise to an algebra whose characters organise the $G_2$ instanton partition function. As a first step towards this program we argue by string duality that a multitude of geometries $\mathcal X$ exist that are dual to well-known 4d SCFTs arising from D3 branes probes of CY cones: all these models are amenable to an analysis along the lines suggested by Dijkgraaf, Gukov, Neitzke and Vafa in the context of topological M-theory. Moreover, we discuss an interesting relation to Costello's twisted M-theory, which arises at local patches, and is a key ingredient in identifying the relevant algebras.

\end{abstract}

\vfill{}
--------------------------

September 2021

\thispagestyle{empty}

\newpage

\tableofcontents

\allowhyphens

\section{Introduction}

The study of string compactifications on backgrounds with $G_2$ holonomy has received a recent renewal of interest in the physics literature \cite{Halverson:2014tya,Halverson:2015vta,Braun:2016igl,daCGuio:2017ifs,Braun:2017ryx,Braun:2017uku,Braun:2017csz,Braun:2018fdp,Acharya:2018nbo,Braun:2018vhk,Braun:2019lnn,Braun:2019wnj,Barbosa:2019bgh,Acharya:2020vmg,Acharya:2021rvh,Ashmore:2021pdm,Cvetic:2021maf,Hubner:2020yde,Cvetic:2020piw} due to the discovery of infinitely many novel examples of such varieties. We can group the two kinds of novel $G_2$ backgrounds in two broad classes: on the one hand we have the compact ones, which arise from twisted connected sums of asymptotically cylindrical Calabi Yau 3-folds \cite{MR2024648,Kovalev2010,MR3109862,Corti:2012kd}, on the other hand we have non-compact local models that are obtained either from specific circle fibrations over asymptotic Calabi-Yau cones \cite{Foscolo:2017vzf} or as asymptotically conical $G_2$ spaces in their own right \cite{Foscolo:2018mfs}.

\medskip

Our paper builds on the remark that many of these non-compact cases might admit equivariant-like twists as M-theory backgrounds \cite{Nekrasov:2014nea,Costello:2016mgj,Costello:2016nkh} --- see also \cite{Eager:2021ufo,Saberi:2021weg}. In analogy with what happens for Calabi-Yau 3-folds \cite{Iqbal:2003ds,Nekrasov:2014nea,DelZotto:2021gzy} it is natural to expect these twisted backgrounds are related to suitable generalisations of topological string partition functions, the so called topological M-theories \cite{Dijkgraaf:2006um,Dijkgraaf:2004te}, in terms of equivariant Donaldson-Thomas (DT) theories \cite{Nekrasov:2014nea}.\footnote{\ See also \cite{deBoer:2005pt,Bonelli:2005ti,Bonelli:2005rw,Bonelli:2006ph} for other attempts at a topological string or topological M-theory on $G_2$ manifolds.} From this perspective in particular, it is natural to expect an interpretation of topological M-theory in terms of a generating function of Witten indices of supersymmetric quantum mechanical sigma models with target suitable supersymmetric versions of the $G_2$ instanton moduli spaces of a given geometry. 

\medskip


This short note is motivated by the question of understanding the structures governing the enumerative geometry of local $G_2$ manifolds \cite{Donaldson:1996kp}, as well as their physical applications in the context of (equivariantly) twisted M-theory \cite{Nekrasov:2014nea,Costello:2016nkh,Costello:2017fbo,Gaiotto:2019wcc} and stringy correspondences \cite{DelZotto:2021gzy,CORR}.

\medskip

Geometric engineering techniques in M-theory can be exploited to associate to a given local $G_2$ manifold $\GX$ a four-dimensional theory with $\mathcal N=1$ supersymmetry which we denote $\TX$ in this paper. Whenever $\TX$ has a conserved $U(1)_R$ symmetry one can consider twisting this theory on any Kähler four-manifold, $\mathcal M$, by identifying the $U(1)_R$ symmetry with a $U(1)$ subgroup of its $U(2)$ holonomy \cite{Johansen:1994aw,Witten:1994ev,Johansen:1994ud,Johansen:2003hw} -- see also \cite{Closset:2013vra,Closset:2014uda}. The resulting theory is not topological in this case, rather it has a dependence on the holomorphic structure of $\mathcal M$. Often it is possible to compute the corresponding partition function $Z_{\mathcal T}(\mathcal M)$, for example by means of localization methods. Then, we have a natural correspondence
\be\label{eq:4d}
Z_{\TX}(\mathcal M) = \mathcal Z_M(\mathcal M \times \mathcal X)
\ee
between an M-theory partition function on a background with topology $\mathcal M \times \mathcal X$ and the twisted 4d $\mathcal N=1$ partition function on the background $\mathcal M$ of the field theory $\TX$. This poses the question of constructing $G_2$ manifolds corresponding to theories $\TX$ with an unbroken $U(1)_R$ symmetry. Thanks to string duality, we argue in this paper that an infinite class of such geometries exist, giving rise to local $G_2$ manifolds that geometrically engineer 4d $\mathcal N=1$ SCFTs.\footnote{\ We warmly thank I{\~n}aki Garc{\'i}a Etxebarria for suggesting to exploit this duality to us, as well as explaining several of his results that are crucial to argue for the stability of these backgrounds in M-theory.}

\medskip

In this paper we are interested in the case when either $\mathcal M$ is a multicentered Taub-NUT space with $k$ centers, which we denote $\TNk$, or $\mathcal M$ is an ALE space of type $A_k$.\footnote{\ We refer our readers to the papers \cite{Witten:2009xu} and \cite{DelZotto:2021gzy} for a through review of the properties of these spaces and their relations. Of course our results can be generalised to include other ALE spaces, as well as discrete $C_3$ form fluxes in M-theory to give rise to $G_2$ instantons in all other simple gauge groups.} Recall that the theory obtained from M-theory on $\TNk$ is a 7d $U(k)$ maximally supersymmetric gauge theory. Similarly, the theory  obtained from M-theory on the $\ALEk$ space is a 7d $SU(k)$ maximally supersymmetric gauge theory.\footnote{\ In this paper we are neglecting all the subtleties related to the choices of global structures for these models \cite{GarciaEtxebarria:2019caf,Albertini:2020mdx,Morrison:2020ool}. We plan to revisit our analysis in the near future to take the latter into account.}  Therefore, the backgrounds of the form $\TNk \times \mathcal X$ and backgrounds $\ALEk \times \mathcal X$ are such that they admit a distinct interpretation in geometric engineering, as 7d partition functions for these gauge theories, namely
\be\label{eq:7d}
Z^{7d}_{U(k)}(\mathcal X) = \mathcal Z_M(\TNk \times \mathcal X) \qquad\text{and}\qquad Z^{7d}_{SU(k)}(\mathcal X) = \mathcal Z_M(\ALEk \times \mathcal X)
\ee
Whenever the $G_2$ manifolds have compact associative 3-cycles, we argue that the 7d partition functions on the LHS of these identities receive contributions from $G_2$ instantons for the groups $U(k)$ and $SU(k)$ respectively, which therefore gives an interpretation in terms of enumerative geometry of local $G_2$ manifolds of the partition functions we are interested in. Moreover, if the manifolds $\mathcal X$ admit compact co-associative four-cycles, we expect the parition functions can receive contributions also from 7d monopoles.

\medskip

From the relation with the M-theory partition function, combining  \eqref{eq:4d} with  \eqref{eq:7d} we expect a 4d/7d correspondence of the schematic form \cite{CORR}
\be\label{eq:4d7d}
\begin{gathered}
\xymatrix{
Z^{\,4D}_{\TX}(\mathbf{TN}_k)\ar[d]_{\mathbf{TN} \text{ radius} \atop \to \infty }  \ar@{<->}[r] & Z_{U(k)}^{\,7D}(\mathcal X)\ar[d]^{\text{decouple} \atop U(1)} \\
Z^{\,4D}_{\TX}(\mathbf{ALE}_{A_k})  \ar@{<->}[r]& Z_{SU(k)}^{\,7D}(\mathcal X)
}
\end{gathered}
\ee
which extends the results of \cite{DelZotto:2021gzy} to this example, thus giving rise to a dictionary between holomorphically twisted partition functions of 4d $\mathcal N=1$ SCFTs on multi-centered Taub-NUT spaces and the $U(k)$ $G_2$-instanton partition function. Whenever $\mathcal X$ admits a $U(1)$ isometry,  the $k=1$ case of this correspondence (in particular, the bottom line in equation \eqref{eq:4d7d}) coincides with the definition of topological M-theory as given in \cite{Dijkgraaf:2004te}, in terms of a generating function of $D2$-$D6$ boundstates arising from exploiting the $U(1)$ fibration of $\mathcal X$ as an M-theory circle to reduce to IIA. In this paper we will discuss how this perspective can be exploited to begin unraveling the properties of these partition functions in a simple example. In future papers of this series we will address more complicated geometries and constructions. \textit{The main message of this paper is that building on the 4d/7d correspondence, the algebra of operators which govern the holomorphic sector of 4d $\mathcal N=1$ theory is such that the resulting $G_2$ instanton partition function can be interpreted as a character.}

\medskip

This paper is organised as follows. In section \ref{sec:G2instantonpartitionfun} we argue that the partition function of the 7d $U(K)$ gauge theory, in presence of compact associative cycles, localises on $G_2$ instantons. This gives rise to a definition of the partition function $Z_{U(K)}^{\,7d}(\mathcal X)$ as a generating function for some sort of equivariant volumes of $G_2$-instanton moduli spaces. In section \ref{sec:infinity} we discuss a string duality that allows to construct infinitely many novel examples of $G_2$ manifolds with an interpretation as 4d $\mathcal N=1$ SCFTs that will play a role in the construction of the correspondence and give rise to the details of the corresponding dictionary as suggested by geometric engineering. The resulting local $G_2$ geometries can be understood in terms of fibrations of Taub-NUT spaces over a collection of intersecting associatives. For this class of geometries a strategy to reconstruct the full $G_2$ instanton partition function is presented, in terms of the 4d side of the correspondence in section \ref{sec:example}. Since the $G_2$ instantons are engineered by M2 branes wrapping these loci, this gives a way of determining their contributions in terms of a matrix model which can also be uplifted to an index for an associated supersymmetric quantum mechanics which we derive by a chain of string dualities. 

\medskip

Exploiting such quantum mechanics it is the key towards understanding the algebra of $G_2$ instantons for the example we consider, in particular, zooming to the intersection point of a pair of associatives, one has a local model that can be traced back to a well-known orbifold of the Bryant-Salamon metric for the $G_2$ cone over the three-dimensional complex projective plane $\mathbf{CP}^3$ \cite{BryantSalamon}, that were discussed in \cite{Atiyah:2001qf,Acharya:2001gy}. This is a universal building block for the $G_2$ instanton partition functions that can be understood in terms of the twisted M-theory {\textit{\`a} la} Costello \cite{Costello:2016mgj} (see also  \cite{Eager:2021ufo,Saberi:2021weg}), along the lines of \cite{Oh:2021bwi}. 

Indeed, our search for the algebra of $G_2$ instantons was initiated by the fascinating observation by Costello\footnote{We learned this from Kevin Costello's lecture ``Cohomological hall algebras from string and M theory" in Perimeter Institute, which can be found in https://pirsa.org/19020061.}: if we instead consider twisted M-theory on $\bR\times\bC\times TN_K$ not the $G_2$ cone, one can argue that the Donaldson-Thomas partition function on $\bC\times TN_K$ forms a character of an algebra of operators of twisted M-theory. In this paper, we will consider twisted M-theory on the $G_2$ cone over $\bf{CP}^3$ and compute the $G_2$ instanton partition, which is presumably a part of the $G_2$ Donaldson-Thomas partition function\footnote{It will be nice to study a relation between our work and Joyce's conjecture \cite{Joyce:2016fij}.}. Along the way, we will be able to see that the analogous statement to that of Costello can be made.

As a final comment, let us remark that the main aim of this paper is to pose this problem: here we show that the study of the $G_2$ instanton parition functions have interesting interconnections with various aspects of geometric engineering and M-theory, which already unveils interesting algebraic structures: this is our evidence for the existence of an algebra of $G_2$ instantons.

\medskip

For the benefit of our readers, the vast majority of the technical aspects of the discussion are confined in the Appendix of the manuscript.

\section{M-theory and $G_2$ instantons}\label{sec:G2instantonpartitionfun}

In this section we discuss our dictionary between the 7d $G_2$ instanton partition function and the underlying M-theory geometry. We stress that in this paper, as in \cite{DelZotto:2021gzy}, the explicit expression of the equivariantly twisted M-theory background is not yet known explicitly for the cases of interest. The non-twisted version of these backgrounds are well known to exist in M-theory, and have been discussed recently in \cite{Dedushenko:2017tdw,Cecotti:2019hyk}: these are 11 dimensional spaces with topology $\mathcal M \times \mathcal X$ where $\mathcal M$ is a HyperK\"ahler space and $\mathcal X$ is a $G_2$ holonomy space. The backgrounds are supersymmetric and preserve 2 supercharges. Strong evidence towards the existence of similar equivariantly twisted structures comes from the supergravity limit \cite{Nekrasov:2014nea,Costello:2016nkh,Eager:2021ufo,Saberi:2021weg}, whenever the spaces $\mathcal M$ and $\mathcal X$ have sufficiently large isometry groups.

\subsection{7d gauge theory and geometry}

It is well known that M-theory on $\TNk  \times\mathbf{R}^7 $ gives rise to a $U(k)$ gauge theory along $\mathbf R^7$. The space of normalisable forms on $\TNk$ can be parametrised by $k$ elements,
\be
\omega_{0}, \omega_\alpha \qquad \alpha = 1,...,k-1
\ee
where $\omega_0$ is such that $\int \omega_0 \wedge \omega_{\alpha} = 0$ for all $\alpha$ and $\int \omega_\alpha \wedge \omega_\beta = C_{\alpha,\beta}$, the Cartan matrix of $A_{k-1}$ type. Consequently, the reduction of the $C_3$ field of M-theory on such a basis gives rise to 7d vectors:
\be
C_3 = A^{(0)} \wedge \omega_0 + \sum_\alpha  A^{(\alpha)} \wedge \omega_\alpha
\ee
where $A^{(0)}$ is a $U(1)$ gauge potential, while the remaining $A^{(\alpha)}$ are associated to the Cartan of $SU(k)$. In the limit in which the Taub-NUT radius of $\TNk$ is sent to infinite size, the form $\omega_0$ stops being normalisable and one is left with an $\ALEk$ metric, which captures only the dynamics of a 7d $SU(k)$ gauge theory. In the moduli space of these geometries we have $k-1$ compact $\mathbb P^1_\alpha$ whose volumes correspond to the vevs in the Cartan of $SU(k)$ for the scalars in the corresponding 7d vectormultiplets
\be
\langle \phi_\alpha \rangle = \text{vol } \mathbb P^1_\alpha \qquad\qquad \int_{\mathbb P^1_\alpha}\omega_\beta = \delta_{\alpha\beta} \,.
\ee
M2 branes wrapping such 2-cycles give rise to the corresponding W-bosons. Sending these volumes to zero we have an enhancement of symmetry and the gauge group becomes non-Abelian.

\medskip

Besides these degrees of freedom, the resulting 7d gauge theory has the following non-perturbative solitonic degrees of freedom:
\begin{itemize} 
\item Instantons with worldvolumes of codimension four that are parametrized by M2 branes strechced along $\mathbf{R}^7$;
\item Monopoles with worldvolumes of codimension three that are parametrized by M5 branes wrapped around $\mathbb P^1_\alpha$;
\item Domain walls with worldvolumes of codimension six that are parametrized by M5 branes strechced along $\mathbf{R}^7$.
\end{itemize}
In this paper, we are interested in understanding some features of the partition function of this theory coupled to a rigid supersymmetric background. In general, one expects to receive contributions to these partition functions from BPS solitons that depend on the properties of the localising BPS background.

\subsection{Localising on $G_2$ instantons}

We are interested in the parition function for 7d SYM on a non-compact space with $G_2$ holonomy $\mathcal X$, suitably regularized to avoid the potential IR divergence with a choice of boundary condition at infinity. Several results about this theory on curved spaces have been obtained recently, see e.g. \cite{Minahan:2015jta,Prins:2018hjc,Iakovidis:2020znp}, working \textit{\`a la} Festuccia-Seiberg \cite{Festuccia:2011ws}, however a detailed study of the case of pure $G_2$ holonomy is complicated by the fact that in this latter case one has only a single parallel spinor, and therefore the localising locus is not reduced to a combination of fixed points arising from a Reeb-like structure. Nevertheless, twisted versions of this theory are well-known to exist and lead to an interesting cohomological field theory depending on the $G_2$ structure of $\mathcal X$ --- see \cite{Acharya:1997gp,Baulieu:1997em,Baulieu:1997jx}. In particular, it is expected that these theories receive contributions from analogues of $G_2$ instanton configurations \cite{Fubini:1985jm,Popov:1992sy,Ivanova:1992nj,Harland:2010ojo}.

\medskip

The main aim of this section is indeed to argue heuristically that the 7d instanton partition function is governed by a suitable supersymmetrization of the topological term 
\be\label{eq:toptop}
\varkappa(F) =  \int_{\mathcal X} \Big(\varphi_3 \wedge ch_2(F)\Big)
\ee
where $\varphi_3$ is the parallel $G_2$ form, which is closed and co-closed, i.e. it satisfies
\be
d \varphi_3 = 0 \qquad\qquad d*\varphi_3 = 0\,.
\ee 
This ensures that $\varkappa(F)$ is a characteristic class. The action for 7d SYM is obtained by dimensional reduction from the 10d SYM Lagrangian (see e.g. \cite{Green:1987mn})
\be
\mathcal L = - {1\over 4} F_{MN}F^{MN} - {i \over 2} \overline{\Psi} \Gamma^M D_M \Psi
\ee
where $\Psi$ is a Majorana-Weyl 10d spinor in the Adjoint of the gauge group. The supersymmetry transformations are
\be
\delta A_M = {i \over 2} \overline{\epsilon} \Gamma_M \Psi \qquad\qquad \delta \Psi = - {1 \over 4} F_{MN}\Gamma^{MN}\epsilon
\ee
The resulting theory has an Adjoint scalar triplet of $SU(2)_R$ symmetry, $L_a$, that arises from the components of $A_M$ in the directions 789. The corresponding supersymmetric action can be found in \cite{Townsend:1983kk}. Here we consider the partition function of the 7d $U(k)$ gauge theory on $\mathcal X$, which can be written schematically as
\be\label{partition1}
Z^{7d}_{U(K)}(\mathcal X)=\int \mathcal D\Phi e^{- \int_{\mathcal X} S[\Phi]},
\ee
where we collectively denote the fields in the 7d $\mathcal N=1$ SYM as $\Phi$, $S[\Phi]$ is the action of the 7d SYM on the $G_2$ manifold $\mathcal X$, and of course
\be\label{eq:ym}
S[\Phi] \supseteq S_{YM}(F) = \, {1\over g^2} \int_{\mathcal X} \text{tr}_{\mathfrak g}  F_A\wedge \ast F_A
\ee
where $F_A$ is the 7d $U(k)$ curvature with potential $A$. The trace is taken in the the normalisation for which a single instanton configuration of $SU(2)$ has unit instanton number, and $g$ is the gauge coupling constant of the 7d theory. We proceed by arguing that the corresponding action is dominated by $G_2$ instantons in the saddle point approximation, following \cite{Tian:2000fu}. We present our argument on the Yang-Mills side  \eqref{eq:ym}.

\medskip

First of all it is necessary to remind our readers that two forms on a $G_2$ manifold are naturally organised according the representation theory of $G_2$ as follows \cite{joyce2000compact} (see also the nice review \cite{Acharya:2004qe})
\be
\wedge^2 T^* \mathcal X =\wedge^2_{\mathbf{7}} \oplus \wedge^2_{\mathbf{14}}, \qquad\qquad H^2(\mathcal X) = H^2_{\mathbf{7}}(\mathcal X) \oplus H^2_{\mathbf{14}} (\mathcal X),
\ee
where the splitting is orthogonal with respect to the canonical pairing
\be\label{eq:hodge}
\langle\eta,\xi\rangle_{\mathcal X} = \int_{\mathcal X} \eta \wedge \ast \xi
\ee
induced by the Hodge star of the $G_2$ manifold. We call $\pi_{\mathbf 7}$ and $\pi_{\mathbf{14}}$ the corresponding projectors.\footnote{\ Sometimes we use the short hand $\eta_{\mathbf{7}} = \pi_\mathbf{7} \eta$ and $\eta_{\mathbf{14}} = \pi_\mathbf{14} \eta$.} Then, of course
\be
S_{YM}(F) = S_{YM}(F_{\mathbf{7}}) + S_{YM}(F_{\mathbf{14}}) 
\ee
There is a natural $G_2$--equivariant linear map, which is compatible with such a splitting, meaning that \cite{BryantMetrics}
\be\label{eq:Tmap}
T: \wedge^2 T^* \mathcal X \to \wedge^2 T^* \mathcal X \qquad T(\eta) \equiv \ast (\eta \wedge \varphi_3) = 2 \pi_{\mathbf 7}(\eta) - \pi_{\mathbf{14}}(\eta).
\ee
It is easy to see that 
\be
\int_{\mathcal X} \eta \wedge \eta \wedge \varphi_3 = 2 \langle \pi_{\mathbf 7}(\eta),\pi_{\mathbf 7}(\eta)\rangle_{\mathcal X} - \langle\pi_{\mathbf{14}}(\eta),\pi_{\mathbf{14}}(\eta)\rangle_{\mathcal X}
\ee
follows from the definition \eqref{eq:hodge} and the property \eqref{eq:Tmap} by orthogonality. 

\medskip

Therefore, the topological term in \eqref{eq:toptop} is such that
\be
\begin{aligned}
\alpha \varkappa(F) &= \text{ tr}_{\mathfrak g} \int_{\mathcal X} F_A \wedge F_A \wedge \varphi_3 = 2 \text{ tr}_{\mathfrak g} \langle \pi_{\mathbf 7}(F_A),\pi_{\mathbf 7}(F_A)\rangle_{\mathcal X} - \text{tr}_{\mathfrak g} \langle\pi_{\mathbf{14}}(F_A),\pi_{\mathbf{14}}(F_A)\rangle_{\mathcal X}\\
&= g^2 \Big( \, 2  S_{YM}(F_{\mathbf{7}}) -  S_{YM}(F_{\mathbf{14}})\Big)
\end{aligned}
\ee
Here $\alpha$ is a constant to keep track of the correct normalisation for the second Chern class.  This gives rise to the following identities \cite{earpthesis,Earp:2011dh}:
\be
\begin{aligned}
S_{YM}(F) =  {1 \over 2} \Big(3 S_{YM}(F_{\mathbf{14}}) + \frac{\alpha}{g^2} \varkappa(F) \Big) = 3 S_{YM}(F_{\mathbf{7}}) - \frac{\alpha}{g^2} \varkappa(F)
\end{aligned}
\ee
Since $S_{YM}(F)$ has to be positive, then only two possibilities are allowed, either $F_{\mathbf{7}} = 0$ and $\varkappa(F) < 0$, or $F_{\mathbf{14}} = 0$ and $\varkappa(F) > 0$, which are the $G_2$ analogues of the self-dual and anti-self dual equations. For compact $G_2$ manifolds, $F_{\mathbf{7}}$ is always zero and for this reason here we will also focus on those instanton configurations.  The topological term in equation \eqref{eq:toptop} is then saturated by $G_2$ instanton configurations that are supported on an associative 3-cycle $\sigma_3$ and have a non-trivial Chern class in the corresponding Poincaré dual four-form. By definition, these $G_2$ instantons are solutions to the equation
\be
\ast F_A = - F_A \wedge \varphi_3\,.
\ee
It is interesting to remark that $G_2$ instanton equations are rigid for compact $G_2$ spaces, meaning that the corresponding moduli spaces have virtual dimension zero, and the corresponding moduli space is a collection of points \cite{Doanthesis,Walpuskithesis}, however for non-compact manifolds the moduli problem becomes, of course richer, and more interesting moduli spaces are allowed \cite{driscoll2021deformations}. Of course, for 7d SYM field configurations the relevant moduli spaces will be enhanced by the contribution from the other fields in the 7d SYM Lagrangian.

\medskip

\textbf{Remark.} In general,  $\mathcal X$ can have calibrated associative 3-cycles as well as calibrated co-associative 4-cycles. For the examples that we are considering in this paper, however, we always have a phase of the geometry where there are no compact co-associative 4-cycles, so we will assume it to be the case for the time being. We will comment on the phases with compact coassociative 4-cycles below.

\medskip

 Let us denote by $\sigma_{3,i}$ a basis of $H_3(\mathcal X,\mathbb Z)$ and by abuse of notation also their calibrated representatives. The complexified volumes will give rise to complex scalar moduli for M-theory, which we denote
\be
q_i =  \int_{\sigma_{3,i}} {\boldsymbol \varphi}_3 \qquad\qquad {\boldsymbol \varphi}_3 = \varphi_3 + {i \over 2\pi} \, C_3.
\ee
Here, of course, $C_3$ is the M-theory 3-form. A $G_2$ instanton configuration for $U(k)$ with instanton number $\kappa_i$ along $\sigma_{3,i}$ is engineered in this setting by a stack of Euclidean M2 branes wrapping $\sigma_{3,i}$ for $\kappa_i$ times. The contribution from such a $G_2$ instanton to the parition function is proportional to
\be\label{eq:M2inst}
Z^k_{\kappa_i}(\sigma_{3,i}) Q_{i}^{\kappa_i},
\ee
where $Q_{i}^{\kappa_i}  = e^{-q_i \kappa_i}$ and $Z^k_{\kappa_i}(\sigma_{3,i})$ is the partition function of $\kappa_i$ M2 branes wrapping the associative 3-cycle $\sigma_{3,i} \subset \mathcal X$ in the background of $\TNk \times \mathcal X$. Whenever we have multi-wrapping, this gives rise to further contributions proportional to the same power of $Q_{i}$. In general, the coefficient multiplying $Q_i^N$ is of the from
\be
\sum_{\kappa_i, m_i \text{ such that}\atop N = \kappa_i m_i} (Z^k_{\kappa_i}(\sigma_{3,i}))^{m_i}\,.
\ee
Whenever the instanton is supported on $\sigma_{3} = \sum_{i=1}^{b_3(\mathcal X)} m_i \sigma_{3,i}$ with instanton numbers $\kappa_i$ on each $\sigma_{3,i}$, its contribution to the partition function receives a weight
\be
Q_{\sigma_3}^{\kappa} = \prod_{i=1}^{b_3(\mathcal X)} Q_{i}^{m_i \kappa_i}.
\ee
Therefore, schematically, for examples without co-associatives we can write
\be
\mathcal Z_M = \sum_{\sigma \in H_3(\mathcal X)} \sum_\kappa Z^k_{\kappa}(\sigma) Q_{\sigma}^{\kappa}.
\ee

\medskip

\noindent \textbf{Remark.} The M-theory partition function could in principle receive contributions also from M5 branes wrapping $\mathbb P^1_\alpha \times \mathcal S_{4,a}$ where $\mathcal S_{4,a}$ is a compact co-associative 4-cycle. These contributions will be relevant for slight generalisations of our backgrounds, which are associated to a different twist of the 4d $\mathcal N=1$ theory, in particular, the vevs of such co-associatives will break the $U(1)_R$ symmetry and so the corresponding backgrounds would no longer be twisted holomorphically. A construction of such partition functions on the 4d side of the 4d/7d correspondence can be achieved  by exploiting rigid curved supersymmetry and non-minimal $\mathcal N=1$ supergravity coupled to the corresponding Ferrara-Zumino $\mathcal F$ multiplet. Since the latter reduces to the $\mathcal R$ multiplet whenever we are switching off the offending vevs, we believe that such partition functions are the most natural extension to the partition functions we consider in this note.


\bigskip

\noindent Based on the above analysis then, naively one expects that
\be\label{eq:7dpartf}
\widehat{Z}^{\,7D}_{U(k)}(\mathcal X)= {Z^{7D}_{U(k)}(\mathcal X) \over Z_{U(k)}^{(0)}(\mathcal X) } = \sum_{\sigma_3 \in H_3(X)} \sum_{\kappa \geq 0} Q_{\sigma_3}^{\kappa}  \int_{\mathcal M^{U(k)}_{\sigma_3,\kappa}} \mathbf{1},
\ee
where:
\begin{itemize}
\item $Z_{U(k)}^{(0)}(\mathcal X)$ is a perturbative contribution;
\item $\mathcal M^{U(k)}_{\sigma_3,\kappa}$ is the moduli space of a BPS version of the $U(k)$ $G_2$ instanton, with instanton numbers $\kappa_i$ over the collection of normal bundles in the associative 3-cycles in the support of the vector $m_i$ associated to $\sigma_3 = \sum_i m_i \sigma_{3,i}$. Each $G_2$ instanton configuration contributes to the partition function with a suitably regularized volume of such moduli space.\footnote{\ A possible regularization arise in the equivariant setting, where we expect that parameters for each instanton contribution are associated to a local $\Omega$ deformation \cite{Lossev:1997bz,Nekrasov:2002qd} of the $SO(4)$ rotation of $\mathcal N_{\sigma_{3,i}}$. The explicit form of the glueing rules among such parameters into global ones can be read off by the existence of the global $G_2$ structure.}
\end{itemize}
Comparison between \eqref{eq:7dpartf} and \eqref{eq:M2inst} leads to identify
\be
Z_{\kappa_i}^k(\sigma_{3,i}) = \int_{\mathcal M^{U(k)}_{\kappa_i,\sigma_{3,i}}} \mathbf{1}\,.
\ee
Building on the idea of topological M-theory \cite{Dijkgraaf:2004te}, for the class of examples that we will consider in this paper we claim that the matrix model computing  $Z_{\kappa_i}^k(\sigma_{3,i}) $ can be obtained by considering bound states of D2 and D6 branes in the type IIA frame. For the special example we consider in section \ref{sec:example} below, we also comment on an interesting duality to IIB where the same configuration is related to a 3d $\mathcal N=2$ theory with impurities (a close cousin of the theories studied in \cite{Cherkis:2009hpw,Cherkis:2011ee}), which reminds also of the discussions in\cite{Dedushenko:2017tdw,Cecotti:2019hyk}.\footnote{\ This last connection is very interesting: the theory with impurities considered in these references is associated to an M5 wrapping a co-associative, precisely of the type that would contribute to our parition function in presence of compact 4-cycles. We believe this perspective could be the key to model certain associative transitions, which are the $G_2$ manifold uplift of the conifold transitions in IIA \cite{Kachru:2001je,Atiyah:2001qf}.}
 
\section{$\mathcal N=1$ SCFTs, $G_2$ manifolds, and dimers}\label{sec:infinity}

In this section we give a simple argument for the existence of infinitely many examples of $\TX$ that have a $U(1)_R$ isometry.

\subsection{Lightning review of dimers and mirrors}
It is well known that probing singularities at the tip of Calabi-Yau cones $\mathbf{X}$ with a stack of $N$ D3 branes in type IIB one obtains 4d $\mathcal N=1$ SCFTs, which we label $\mathcal T_{\mathbf{X},N}$.\footnote{\ There is a multitude of references about this subject and therefore it is hard to properly cite them, some references there were key for the preparation of this section are \cite{Ooguri:1997ih,Lawrence:1998ja,Klebanov:1998hh,Beasley:1999uz, Hanany:2001py,Feng:2005gw}.} In particular, when $\mathbf{X}$ is a local toric CY threefold, the corresponding theories on the stack of $N$ D3 branes have a 4d $\mathcal N=1$ Lagrangian description that is straightforward to compute, being encoded in a dimer model. In general, many of these theories can be captured by a quiver diagram (also in the non-toric cases) that can be read off directly from the structure of the singular orbifold geometry exploiting standard techniques that we are not going to review here. Some examples of such quivers can be found in Figure \ref{fig:quiversfromIIB}. 

\begin{figure}

\begin{center}
$$
\begin{gathered}
\xymatrix{&\circ \ar[dr]\ar@<-0.4ex>[dr]\ar@<+0.4ex>[dr] & \\ \circ \ar[ur]\ar@<-0.4ex>[ur]\ar@<+0.4ex>[ur] &&\circ \ar[ll]\ar@<-0.4ex>[ll]\ar@<+0.4ex>[ll] }
\end{gathered} \qquad\qquad \begin{gathered}
\xymatrix{\circ \ar@<-0.4ex>[rr]\ar@<+0.4ex>[rr] && \circ \ar@<-0.4ex>[dd]\ar@<+0.4ex>[dd] \\ \\ \circ \ar@<-0.4ex>[uu]\ar@<+0.4ex>[uu] &&\circ \ar@<-0.4ex>[ll]\ar@<+0.4ex>[ll] }
\end{gathered} \qquad\qquad \begin{gathered}\xymatrix{}\end{gathered}
$$
\end{center}
\caption{Some well-known examples of quivers from dimers.}\label{fig:quiversfromIIB}
\end{figure}
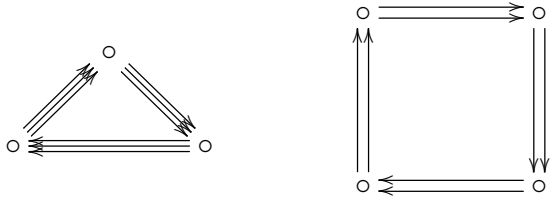

\medskip

A well-known feature of these geometries is that the resulting theories can be represented in Type IIA by exploiting mirror symmetry. By mirror symmetry, the IIB theory on the toric threefold $\mathbf{X}$ is equivalent to IIA on the mirror 3-CY $\mathbf{X}^\vee$. The D3 brane stack in IIB is mapped to a configuration of wrapped $D6$ branes that have support on special Lagrangian 3-cycles in the mirror geometry $\mathbf X^\vee$, which can be understood, e.g. in terms of the SYZ fibration \cite{Strominger:1996it}. We denote a basis of such cycles $\mathcal C_i \in H_3(\mathbf{X}^\vee,\mathbb Z)$. For the mirrors of the toric cases we are interested in, the topology of these 3-cycles is either an $S^3$ or an $S^1 \times S^2$.

\medskip

On the IIB side, the D3 brane fractions at the singularity $\mathbf{X}$ in a number $b_3$ of fractional D3 branes \cite{Douglas:1996sw}. These are in one-to-one correspondence with the cycles $\mathcal C_i$ in the mirror dual. The original D3 brane is a bound state of fractional 3-branes, consisting of $N_i$ copies of the $i$-th 3-brane. This is mapped in the mirror to $N_i$ D6 branes wrapping the 3-cycle $\mathcal C_i$. These D6 branes give rise to gauge groups $SU(N_i)$ on the 4d theories:\footnote{\ The corresponding $U(1)$ are expected to decouple in the IR.} if $\mathcal C_i$ is an $S^3$ we obtain an $\mathcal N=1$ vectormultiplet, if instead $\mathcal C_i \simeq S^1 \times S^2$ the resulting degrees of freedom are those of an $\mathcal N=2$ vectormultiplet. The intersection pairing 
\be\label{eq:adjacia}
B_{ij} = \mathcal C_i \cdot \mathcal C_j
\ee
in the mirror geometry corresponds to an adjacency matrix of the quiver theory describing the D3 brane worldvolume. In particular, if a pair of wrapped cycles $\mathcal C_i, \mathcal C_j$ are intersecting, then the corresponding wrapped D6 branes are intersecting, in which case we obtain extra massless matter from open strings ending on them. Generically, we will have configurations in which the 3-cycles are $S^3$'s that intersect over a collection of $|B_{ij}|$ points. In correspondence to each point with $B_{ij} > 0$ we obtain an $\mathcal N=1$ chiral multiplet in a bifundamental $(\overline{\mathbf{N}}_i,\mathbf{N}_j)$. In cases when the 3-cycles overlap along an $S^1$ the resulting contribution is an $\mathcal N=2$ hypermultiplet, which however gives $B_{ij} = 0$. For most of the toric models, the topology of the $\mathcal C_i$ is always that of $S^3$ and one has an $\mathcal N=1$ theory. In particular, in this case, the $N_i > 0$ must satisfy 
\be\label{eq:noano}
B_{ij}N_j = 0,
\ee
which guarantees anomaly cancellation on the 4d $\mathcal N=1$ side. Whenever the quiver has loops, the resulting 4d $\mathcal N=1$ theory can also have an interesting tree-level superpotential, which, in favourable circumstances, is also determined using geometry. For the toric cases $N_i = N$ for all $i$.

\subsection{From IIA to M-theory}

The $G_2$ geometries that we will consider are captured by the IIA geometry of $\mathbf{X}^\vee$ by duality.\footnote{\ It would be interesting to investigate generalising slightly the construction of Foscolo, Haskins, and N{\"o}rdstrom \cite{Foscolo:2017vzf} to include orbifold singularities. We believe many of the geometries in this section are orbifolds of the FHN metrics.}  One has to uplift the configuration of intersecting D6 branes over the collection of special Lagrangian 3-cycles $\mathcal C_i$ to a hyperK\"ahler fibration on a collection of associative cycles $\sigma_{3,i} \simeq \mathcal C_i$, which gives rise to a $G_2$ manifold that is a $U(1)$ fibration over $\mathbf{X}^\vee$. The fibration is such that the $U(1)$ bundle degenerates along the cycles $\sigma_{3,i}$ that correspond to the special lagrangian 3-spheres $\mathcal C_i$ that are wrapped by $N_i$ $D6$ branes, giving rise to a local $\mathbb C^2/\mathbb Z_{N_i}$ singularities with locus $\sigma_{3,i}$. The singularities enhance to $\mathbb C^2/\mathbb Z_{N_i+N_j + 1}$ at the intersection where $N_i$ D6s meet with $N_j$ D6s, similar to what happens in other geometric engineering settings \cite{Katz:1996xe}. In this context, it is natural to conjecture that the condition for the absence of anomalies \eqref{eq:noano} is equivalent to the condition that the corresponding $\varphi_3$ form is not only closed but also co-closed. For toric examples, we have that $N_i = N$, where $N$ is the number of D3-branes probing $\mathbf{X}$. Therefore, we label this class of $G_2$ manifolds $\mathcal X_{\mathbf{X}^\vee, N}$. From the perspective of geometric engineering, we have that
\be
\mathcal T_{\mathbf{X},N} = \mathcal T_{\mathcal X_{\mathbf{X}^\vee,N}}.
\ee

\medskip

In the toric examples, the mirror geometries are known explicitly. One starts from the toric diagram of $\mathbf{X}$ and obtains the mirror geometry
\be
\mathbf{X}^\vee = \{uv = W(X,Y)\} \subset \mathbf C^2 \times (\mathbf C^*)^2,
\ee
where $W(X,Y)$ is the Newton polynomial obtained from the toric diagram by assigning to the internal dots with coordinates $(n,m)$ the monomial $X^nY^m$, and $X$ and $Y$ are $\mathbb C^*$ variables \cite{Hori:2000kt}. We refer our readers to the beautiful detailed discussion in \cite{Feng:2005gw} for more aspects of these geometries. 

\medskip

An application of the fact that these models have an explicit description in terms of a toric quiver, is that it is possible to use these gauge theories to argue for the lack of contributions to the superpotential from euclidean M2 branes \cite{Harvey:1999as,Pantev:2009de,Hubner:2020yde} wrapping the associative 3-cycles $\sigma_{3,i}$. This follows because for all toric dimer models the quiver associated to the corresponding geometry is necessarily such that $N_f = \ell N_c$ where $\ell \geq 2$. Since each gauge node can be treated independently from the perspective of the generation of a non-perturbative superpotential in gauge theory, this is precisely the condition that there are as many incoming as there are outgoing arrows from each node, which is a condition that every toric quivers satisfy. Then there is no non-perturbative superpotential generated \cite{Intriligator:1994jr}: in all these cases the instanton generates no superpotential, but rather higher F-terms \cite{Beasley:2004ys}.\footnote{\ We thank I\~naki Garc{\'i}a Etxebarria for an important clarifying discussion about this point.}   The higher F-terms involve lots of fermion insertions and therefore most likely will be irrelevant deformations from the perspective of the SCFT. This guarantees that these geometries are not receiving non-perturbative superpotential contributions from the associative 3-cycles, and therefore evade the main issue in geometric engineering of 4d $\mathcal N=1$ SCFTs, namely that non-perturbative contributions to the superpotential could spoil geometric intuition giving rise to a quantum moduli space that differs from the classical geometrical one --- see e.g. \cite{Apruzzi:2018oge} for a recent discussion of this effect in the context of F-theory geometric engineering of 4d $\mathcal N=1$ SCFTs.

\medskip

This concludes our brief overview of these geometries. We stress that there are many applications of this geometric setting, which uplift to $G_2$, for instance, the geometrization of dualities by Ooguri and Vafa \cite{Ooguri:1997ih}. We will return to the study of more details of these geometries in future papers in this series.

\section{Towards the algebra of $G_2$ instantons}\label{sec:example}

\subsection{Evidence from 4d/7d correspondence}
In the previous section we have given an argument to provide infinitely many examples of local $G_2$ geometries that give rise to systems with an unbroken $U(1)_R$ symmetry, and therefore can be exploited to construct 4d holomorphically twisted $\TNk$ partition functions.\footnote{\ It would be interesting to explicitly construct the holomorphic partition function of these models on multicentered Taub-NUTs.} We can use this fact to our advantage while proceeding in our investigation of the structure governing the BPS $G_2$ instantons on the stringy side for the corresponding $G_2$ manifolds $\mathcal X_{\mathbf{X}^\vee,N}$. 

\medskip

The main feature of all the local models of $G_2$ geometries we introduced in the previous section is that those are dual to collections of $N$ D6 branes wrapping compact special Lagrangians in the mirrors of toric geometries $\mathcal X_{\mathbf{X}^\vee}$. We are interested in a configurations of Euclidean D2 branes, wrapping the same special Lagrangians. From geometry it is clear that the most interesting contributions will arise from intersecting loci (other loci would be closer to ordinary instantons). Our task therefore is the following: (1) exploiting branes, identify the contribution of intersecting loci to our putative BPS $G_2$ instanton partition function; (2) establish gluing rules corresponding to the quiver diagram associated to the corresponding geometry. This is in a sense the analogue of the topological vertex \cite{Aganagic:2003db,Iqbal:2007ii} for this special class of $G_2$ manifolds.

\medskip

This task is seemingly impossible, however some recent progresses in the context of twisted M-theory can be exploited to our advantage. The main point is that the twisted M-theory background by Costello \cite{Costello:2016nkh} provides a local model of the form\footnote{\ We refer our readers to appendix \ref{app:costellobackground} for a quick review.}
\be\label{eq:costellone}
\TNk \times \mathbf{R}^3 \times \mathbf{TN}_{N_i} 
\ee
where the twist is holomorphic on $\TNk$, and topological on $\mathbf{R}^3 \times \mathbf{TN}_{N_i}$. Such a geometry can be identified with a local patch in the structure of any of the geometries discussed in section \ref{sec:infinity}: it is enough to identify the $\mathbf{R}^3$ in equation \eqref{eq:costellone} with a local neighborhood of one of the associatives 3-cycles $\sigma_{3,i}$. The fact that the twist is holomorphic on $\TNk$ guarantees that this twisted M-theory construction contributes a subalgebra of the cohomological field theory associated to the holomorphic twist of one of the 4d $\mathcal N=1$ models $\mathcal T_{\mathcal X_{\mathbf{X}^\vee,N}}$. More precisely, the contribution from a subsector like \eqref{eq:costellone} is encoded in an algebra that we schematically denote $\mathcal A^k_{N_i}$. It is natural to conjecture that 
\be\label{eq:ninja}
\mathcal A^k_{N_i} \simeq U_{\epsilon_1}(\mathcal O_{\epsilon_2} (\TNk) \otimes \mathfrak{gl}_{N_i}) 
\ee
which is the quantum deformation with parameter $\epsilon_1$ of the universal enveloping algebra of the algebra of deformed holomorphic functions on $\TNk$ tensored by $\mathfrak{gl}_{N_i}$.

\medskip

For the class of examples in this note, whenever we have a pair of intersecting D6 branes in the background of $\TNk$, we expect to obtain a more general algebra that has two such factors
\be\label{eq:bifalg}
\mathcal A_{N_i}^{k} \otimes_{B_{ij}} \mathcal A_{N_j}^{k},
\ee
where the notation $\otimes_{B_{ij}}$ indicates that there are some further relations which arise from the number of points of intersection among the associatives, and from gauging, and hence the corresponding system is not just a tensor product (although being close to one).

\medskip



Based on the structure of the geometries $\mathcal X_{\mathbf{X}^\vee,N}$ it is natural to expect that the factors $\mathcal A_{N_i}^{k} \otimes_{B_{ij}} \mathcal A_{N_j}^{k} $
fully determine the algebra of the associated quiver theory. Therefore, by 4d/7d correspondence we claim that for all these examples, there is an algebraic structure underlying the $G_2$ instanton partition function, which we schematically denote $\mathcal A_k(\mathcal X_{\mathbf{X^\vee},N})$. This is the algebra of $G_2$ instantons, meaning that the corresponding 7d $G_2$ instanton partition function is organized according to characters of $\mathcal A_k(\mathcal X_{\mathbf{X^\vee},N})$. The task of finding a precise definition of such an algebra is extremely hard, and in this section we build some evidence towards its existence exploiting string dualities, which allows to capture some of the features of its building blocks. From the perspective of the 4d $\mathcal N=1$ theory, we expect the $\mathcal A_k(\mathcal X_{\mathbf{X^\vee},N})$ to be computable in terms of the structure of the complexes associated to the differential arising from the holomorphic twist on $\TNk$. We will devote to this computation another paper in this series.

\medskip

Exhibiting the details of the gluing rule goes beyond the scope of this short note. In what follows we proceed with a consistency check of \eqref{eq:bifalg}: we consider the $k=1$ case, and we prove that indeed equation \eqref{eq:ninja} is verified, by considering a local model for the D6 brane intersection.

\subsection{Building blocks from intersecting $D6$ branes}\label{sec:G2instantonsfromIIA}

As a local model for the intersection between two D6 branes stacks, we can exploit a famous class of asymptotcally conical $G_2$ spaces, which are metric cones on complex weighted projective spaces $\mathbf{WCP}^3_{N,N,K,K}$
\be
\mathcal X_{N,K} \equiv \mathcal C(\mathbf{WCP}^3_{N,N,K,K}),
\ee
here the case $N=K=1$ coincides with the Bryant-Salamon $G_2$ cone over the three-dimensional complex projective plane $\mathbf{CP}^3$ \cite{BryantSalamon}, while the other cases are natural generalization of the Bryant-Salamon construction \cite{Atiyah:2001qf,Acharya:2001gy}. The physics of these models is understood in terms of a single chiral $U(N)\times U(K)$ bifundamental multiplet \cite{Atiyah:2001qf,Acharya:2001gy}:
\be\label{eq:4dthy}
\mathcal T_{\mathcal X_{N,K}} = \left\{\textit{4d } \mathcal N=1 \text{ chiral } \atop U(N)\times U(K) \text{ bifundamental}\right\}\,.
\ee

It is well-known that the $G_2$ manifolds $\mathcal X_{N,K}$ are dual to a IIA configuration of $N$ D6 branes intersecting with $K$ D6 branes at angles. We choose a frame of reference such that the stack of $N$ D6 branes is located along the directions $0,1,2,3,4,5,6$ and the stack of $K$ D6 branes is along the directions $0,1,2,3$ and it wraps a plane $\mathbf{L}_{\boldsymbol\theta} \simeq \mathbf{R}^3$ which extends in the directions $4,5,6,7,8,9$, such that each of its components form an angle $\theta_i$ with the directions $4,5,6$ respectively for $i=1,2,3$.\footnote{In other words,
\be
\mathbf{L}_{\boldsymbol \theta} = \begin{cases}x_4 = \tan \theta_1 x_5 \\ x_6 = \tan \theta_2 x_7 \\  x_8 = \tan \theta_3 x_9\end{cases} \qquad \boldsymbol\theta = (\theta_1,\theta_2,\theta_3)
\ee}
The M-theory background $\TNk \times \mathcal X_{N,K}$ is therefore dual to the following type $II_A$ background
\be\label{eq:IIAbraneconf}
\begin{tabular}{c|cccccccccc}
&0&1&2&3&4&5&6&7&8&9\\
\hline
$N$ D6$_1$ & $\bullet$ & $\bullet$ & $\bullet$ & $\bullet$ & $\bullet$ & $\bullet$ & $\bullet$  &  &  & \\
$M_1$ D2$_1$ &  &  &  &  & $\bullet$ & $\bullet$ & $\bullet$ &  &  & \\
$K$ D6$_2$ & $\bullet$ & $\bullet$ & $\bullet$ & $\bullet$  &  &  & $\mathbf{L}_{\boldsymbol\theta}$ & & & \\
$M_2$ D2$_2$ &  &  &  &  &  &  & $\mathbf{L}_{\boldsymbol\theta}$ & & &\\
$\mathbf{TN}_L$ & & & & & $\bullet$ & $\bullet$ & $\bullet$ & $\bullet$ & $\bullet$ & $\bullet$ \\
\end{tabular}
\ee
 where we are marking the directions that are filled  by the corresponding object. In the Type IIA frame the M2 branes wrapping the corresponding associative cycles are dual in this context to two stacks of D2 branes, one is parallel to D6$_1$ and is wrapping the 456 directions, the other is parallel to D6$_2$ and it is wrapping $\mathbf{L}_{\boldsymbol\theta}$.
 
 \medskip

From the 4d/7d correspondence we expect that the algebra $\mathcal A_{N}^L \otimes_{B_{12} = 1} \mathcal A_{K}^L$ organises the generating function of some suitably regularised volumes of moduli spaces encoded in a quiver matrix model supported at the common intersection of these two stacks of D2 branes. These moduli spaces control the contributions of the M2 branes to the 7d partition function on the $G_2$ manifold.

\medskip

Here, we are interested in probing the conjecture with an explicit computation. therefore, we consider the case $L=1$ and we send the Taub-NUT radius to infinity, thus focusing on the topological M-theory limit of our parition function. The resulting algebra is the $L=0$ version of the $G_2$ instanton algebra. In the absence of the Taub-NUT background it is relatively straightforward to derive the quiver which captures the D2 brane boundstates from the building blocks we described above, namely we will have two distinct ADHM-like quivers, each corresponding to one of the two factors of the algebra $\mathcal A_{N}^0$ and $\mathcal A_{K}^0$, arising from the D2-D6 systems. On top of that from the intersection of the two D2 branes, we expect to obtain a single chiral bifundamental contributing to the quiver. As a result we obtain the quiver in figure \ref{fig:d2d6quiver}.\footnote{\ We review the relevant computation using string quantization in appendix \ref{app:stringquant}, and here we only report on the main result.}
\begin{figure}[H]
\begin{center}
\includegraphics[scale=0.3]{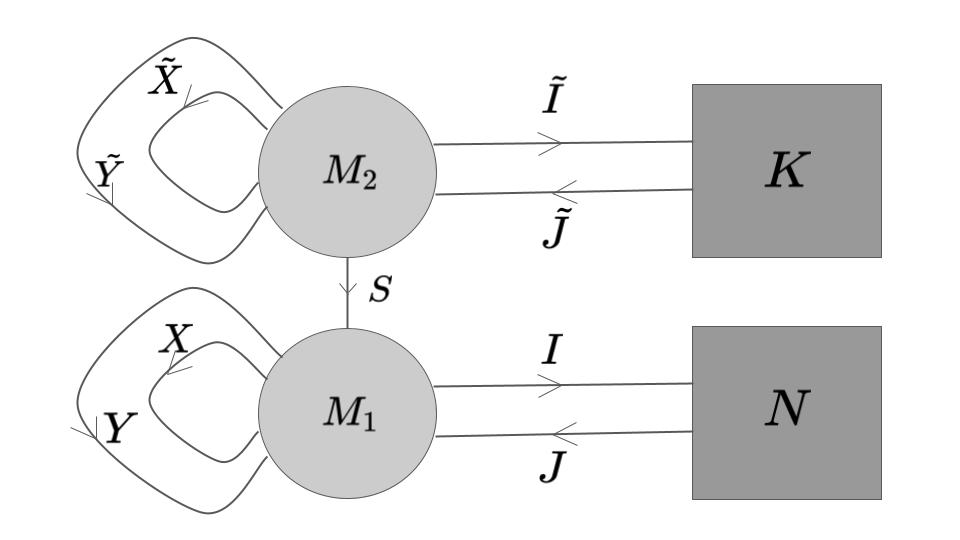}
\caption{The quiver for the matrix model related to the $\mathcal A_{N}^0 \otimes_{B_{12} = 1} \mathcal A_{K}^0$ algebra.}\label{fig:d2d6quiver}
\end{center}
\end{figure}
 Increasing $B_{12}$, the quiver is essentially the same, only we obtain a collection $S^{\ell}$ of chiral bifundamental fields, at the intersections, which would carry an index $\ell = 1,..., B_{12}$. Here, we will focus on the case $\ell=1$, which corresponds to the orbifolds of the Bryant-Salamon $G_2$ cones that we are using as local models.

\medskip

In appendix \ref{app:stringquant}, we discuss a proposal for the corresponding action for the matrix model supported at the intersection of the two D2 branes. From the D-term of the associated quiver description, we see that the two ADHM quivers satisfy their own complex moment map equations:
\ie\label{ADHM eqn}
\left[X,Y\right]+IJ=0,\\
\left[\tilde X,\tilde Y\right]+\tilde I\tilde J=0,
\fe
and we also have a modified real moment map equation, which includes the contribution from the chiral bifundamental field $S$
\ie
\left[X,X^\dagger\right]+\left[Y,Y^\dagger\right]+II^\dagger-JJ^\dagger-S^\dagger S=\xi\cdot\text{I}_{M_1\times M_1},\\
\left[\tilde X,\tilde X^\dagger\right]+\left[\tilde Y,\tilde Y^\dagger\right]+\tilde I\tilde I^\dagger-\tilde J\tilde J^\dagger+SS^\dagger=\xi\cdot\text{I}_{M_2\times M_2},
\fe
where $\xi$ is an FI term. This is the starting point for a more formal analysis for the structures of the corresponding moduli spaces, which we describe in appendix \ref{app:proofs}. 

In Appendix \ref{app:proofs}, we compute the equivariant K-theory index of the corresponding moduli space. Here, we briefly summarize the result, relegating the details to the appendix. The first step is to utilize a well-known fact, namely that the virtual tangent bundle of the moduli space encoded in the quiver diagram is captured by the cohomology of a specific complex:
\ie
\mathrm{End}(\mathcal V_{1})\oplus \mathrm{End}(\mathcal V_{2})\overset{\sigma}{\longrightarrow} \mathrm{End}(\mathcal V_{1})^{\oplus 2}\oplus \mathrm{End}(\mathcal V_{2})^{\oplus 2} \oplus \mathcal V_{1}^{\oplus N} \oplus \mathcal V_{1}^{*\oplus N}\oplus \mathcal V_{2}^{\oplus K} \oplus \mathcal V_{2}^{*\oplus N}\\
\oplus \mathrm{Hom}(\mathcal V_{1},\mathcal V_{2})\oplus \mathrm{Hom}(\mathcal V_{2},\mathcal V_{1})
\overset{d\mu_{\bC}}{\longrightarrow} \mathrm{End}(\mathcal V_{1})\oplus \mathrm{End}(\mathcal V_{2}).
\fe  
Here $\cV_1$, $\cV_2$ are universal bundles of rank $M_1$, $M_2$, $\sigma$ is related to the stability condition or real D-term equation, and $\mu$ is complex moment map equation. Our goal is to compute the Euler characteristics of the virtual tangent bundle of this moduli space\footnote{This Euler character is interpreted physically in terms of a Witten index.}. To do so, it is necessary to find a suitable compactification of the moduli space, which typically is achieved by means of equivariance. We can achieve this by turning on five-torus $\mathbf{T}'_1\times\mathbf{T}'_2\times \mathbf{T}_3$ acting on $\cM^{N,K}(M_1,M_2)$, defined by 
\ie 
(r_1,r_2)\in \mathbf{T}'_{1}&:(X,Y,I,J)\mapsto (r_1X,r_2Y,I,r_1r_2J),\\
(s_1,s_2)\in\mathbf{T}'_{2}&:(\tilde X,\tilde Y,\tilde I,\tilde J)\mapsto (s_1\tilde X,s_2\tilde Y,\tilde I,s_1s_2\tilde J),\\
c\in\mathbf{T}_{3}&:S\mapsto cS.
\fe
We show in the appendix that the torus fixed points of the moduli space are in one-to-one correspondence with $M_1+2$ copies of Young diagrams, see Proposition \ref{counting fixed pts}. Using this fact, we are able to compute the equivariant K-theoretic index of the moduli space of the $G_2$ instantons in \appref{sec:5.5}. Especially, the equivariant K-theory index nicely factorizes into two pieces in the limit $r_1r_2\ria1$, $s_1s_2\ria1$ for the equivariant parameters of the torus action. We show this in Proposition \ref{prop: factorization}. We conjecture that each piece is captured by the Witten index that we will discuss in \secref{sec:4.4}. Finally, in \appref{app:c.6} we explicitly show that in a certain limit of the equivairant parameters the equivariant K-theory index admits a structure of Verma modules of $\mathcal A_{1}^{1} \otimes_{1} \mathcal A_{1}^{1}$. This series of propositions is a strong evidence toward the existence of an algebra of $G_2$ instantons.

\medskip

In sum, we claim that the contributions from the equivariant volumes of the $G_2$ instantons to the 7d partition function are captured by the equivariant K-theory index described above, which we denote $\mathcal I_{M_1,M_2}^1(N,K)$. This gives rise to the following explicit contribution to the partition function
\be
\widehat{Z}^{\,7D}_{U(1)}(\mathcal X_{1,1}) = \sum_{M_1,M_2} Q_1^{M_1 1}Q_2^{M_2 1} \, \mathcal I_{M_1,M_2}^1(1,1)\,.
\ee
The main outcome of such analysis is that indeed, the resulting partition function is organised as a character of $\mathcal A^1_{1}\otimes_1 \mathcal A^1_1$, which is a first consistency check for the ideas discussed here. 
\medskip

\subsection{$G_2$ instantons from an SQM Witten index}\label{sec:4.4}
In this section, we present an alternative way to compute the $G_2$ instanton partition function as a Witten index \cite{Witten:1982df}.  We can achieve this by performing T-duality along two directions on the $D2_1/D6_1/D2_2/D6_2$ brane system. The following table summarizes the resulting D-brane configuration $D0_1/D4_1/D4_2/D8_2$. 
\begin{table}[H]\centering
\begin{tabular}{|c|cccc|cccccc|}
\hline
   &  0 &  1 &  2 &  3 &  4 &  5 &  6 &  7 &  8 &  9 \\
\hline\hline
Background & & & & &A &B &B &A & B & B\\
\hline
$(D2_1\ria)D0$ & & & & &$+$& & && &\\
\hline
$(D6_1\ria)D4_1$ &$+$ &$+$&$+$&$+$&$+$&&& &&\\
\hline
$(D2_2\ria)D4_2$ & & & & & $-$ &$+$ &$+$ &$|$&$+$ &$+$\\
\hline
$(D6_2\ria)D8$ &$+$ &$+$&$+$&$+$&$-$& $+$& $+$& $|$&$+$&$+$\\
\hline
\end{tabular}
\caption{Brane configuration after the T-duality. A and B denote angles and B-field. For instance, D4${}_2$ and D8 branes extend over an oblique line in 47 plane.}
\label{table:D2, D6-brane}
\end{table}
\noindent Note that we T-dualize two of the directions previously occupied by a stack of $D6_1(D2_1)$ branes, which are forming angles with a stack of $D6_2(D2_2)$ branes. As we explained in appendix \ref{app:stringquant}, this T-duality converts D-branes at angle configuration to B-field backgrounds along 5689 directions. It does not modify the result of the string quantization, since the open string boundary condition do not change under the T-duality. In other words, we may simply promote the matrix model field content we discussed in the previous section to a 1d $\cN=4$ quiver supersymmetric quantum mechanics.

\subsubsection{Supersymmetric quantum mechanics}\label{sec:6.1}

Let us first collect the quantum mechanical fields that are related to D0 and D4${}_1$ branes in table \ref{table:D0-D4-D8}.  
\begin{table}[H]
\centering
\begin{tabular}{|c|c|c|c|c|}
\hline
strings &${\cal N}=4$ multiplets & fields & {\footnotesize SU(2)$_-\times$SU(2)$_+\times$SU(2)$^R_-\times$SU(2)$^R_+$} & {\footnotesize{$U(M_1)\times U(N)$}}
\\ \hline
\multirow{8}{*}{D0-D0}& \multirow{3}{*}{vector} & gauge field &$({\bf 1},{\bf 1},{\bf 1},{\bf 1})$ & \multirow{8}{*}{({\bf{Adj,1}})}
\\ \cline{3-4}
& & scalar & $({\bf 1},{\bf 1},{\bf 1},{\bf 1})$ & 
\\ \cline{3-4}
&  & fermions &$({\bf 1},{\bf 2},{\bf 1},{\bf 2})$ & 
\\ \cline{2-4}
&Fermi & fermions & $({\bf 2},{\bf 1},{\bf 2},{\bf 1})$ &
\\ \cline{2-4}
&\multirow{2}{*}{twisted hyper}  & scalars & $({\bf 1},{\bf 1},{\bf 2},{\bf 2})$ &
\\ \cline{3-4}
&& fermions & $({\bf 1},{\bf 2},{\bf 2},{\bf 1})$ &
\\ \cline{2-4}
&\multirow{2}{*}{hyper} & scalars & $({\bf 2},{\bf 2},{\bf 1},{\bf 1})$ &
\\ \cline{3-4}
&& fermions & $({\bf 2},{\bf 1},{\bf 1},{\bf 2})$ &
\\ \hline
\multirow{3}{*}{D0-D4${}_1$}& \multirow{2}{*}{hyper} & scalars & $({\bf 1},{\bf 2},{\bf 1},{\bf 1})$ & \multirow{3}{*}{({\bf{M${}_1$,N}})}
\\ \cline{3-4}
& & fermions & $({\bf 1},{\bf 1},{\bf 1},{\bf 2})$ &
\\ \cline{2-4}
& Fermi & fermions & $({\bf 1},{\bf 1},{\bf 2},{\bf 1})$ &
\\ \hline
\end{tabular}
\caption{The field content of the D0-D4 quantum mechanics. Note that 1d twisted multiplets does not have a matrix model counterpart. Indeed, the field decouples from the computation of the index.}
\label{table:D0-D4-D8}
\end{table}
\noindent They are the standard ADHM mechanical fields \cite{Witten:1994tz,Douglas:1995bn,Douglas:1996uz,Douglas:1996sw} (see also \cite{Hwang:2014uwa,Hori:2014tda}). In the table, we organize the fields by their representation under $SO(4)=SU(2)_-\times SU(2)_+$ and $SO(4)_R=SU(2)^R_-\times SU(2)^R_+$, where $SO(4)$ acts on 0123 directions and $SO(4)_R$ acts on 5689 directions.

Next, the fields that are not related to the usual ADHM quiver are from D0-D4${}_2$ strings and D0-D8 strings, which are T-dual to D2${}_1$-D2${}_2$ strings and D2${}_1$-D6${}_2$ strings. However, we have seen in appendix \ref{app:stringquant} that the D2${}_1$-D6${}_2$(D0-D8) are massive and decouple in the IR. Hence, we only need to discuss the D0-D4${}_2$ strings. Their T-dual is summarized in Table \ref{table:D0-D42}.
\begin{table}[H]
\centering
\begin{tabular}{|c|c|c|c|c|}\hline
strings &${\cal N}=4$ multiplets & fields & {\footnotesize SU(2)$_-\times$SU(2)$_+\times$SU(2)$^R_-\times$SU(2)$^R_+$}  & {\footnotesize U(M${}_1)\times$U(M${}_2$)}
\\ \hline
\multirow{2}{*}{D0-D4${}_2$}& \multirow{2}{*}{twisted hyper} & scalar & $({\bf 1},{\bf 1},{\bf 1},{\bf 2})$ &\multirow{2}{*}{(\bf M${}_1$,M${}_2$)}
\\ \cline{3-4}
& & fermions & $({\bf 1},{\bf 2},{\bf 1},{\bf 1})$&
\\ \cline{2-4}
& Fermi & fermions & $({\bf 1},{\bf 1},{\bf 2},{\bf 1})$ &
\\ 
\hline
\end{tabular}
\caption{The quantum mechanical fields from the D0-D4${}_2$ strings. These will introduce local operators that we can integrate out to form 1-loop determinants-- see \eqref{d0d42}.}
\label{table:D0-D42}
\end{table}

\begin{figure}[H]
\centering
\includegraphics[scale = 0.3]{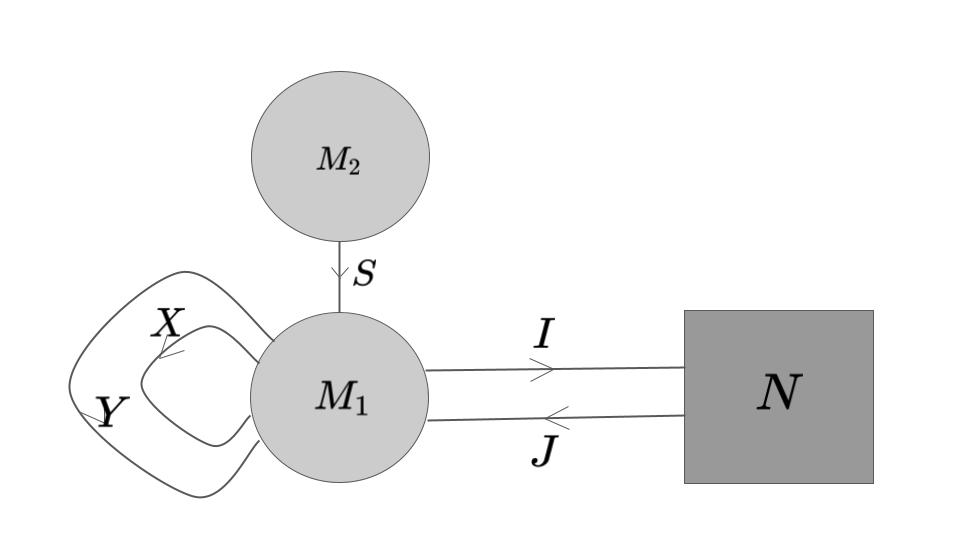}
  \vspace{-5pt}
\caption{A subquiver of the whole quiver that is captured by the Witten index.}
\label{WittenIndexQuiver}
 \centering
\end{figure}

The quantum mechanical fields collected in Table \ref{table:D0-D4-D8} and Table \ref{table:D0-D42} are fields that appear in a subquiver of the entire quiver for the system: compare figure \ref{WittenIndexQuiver} with figure \ref{fig:d2d6quiver}. This is the same quiver that was needed for our analysis in appendix \ref{app:proofs} --- see figure \figref{quiver for fibers} there. We showed in Lemma 3 that $G_2$ instanton moduli space $\cM^{N,K}(M_1,M_2)$ is an entire space of a trivial fibration over a single ADHM moduli space $\cM^N(M_1)$ with the fiber being the moduli space described by \figref{WittenIndexQuiver}. Here we will mimic that construction from the perspective of a Witten index computation.

\subsubsection{Witten index of SQM}\label{sec:6.2}
Here we are interested in computing the following Witten index for our quiver in figure \ref{WittenIndexQuiver}
\ie
\cI^{M_1,N}=Tr_{\mathcal{H}_{\text{QM}}}\left[(-1)^Fe^{\beta\{Q,Q^\dagger\}}t^{2(J_++J_R)}u^{2J_-}v^{2J'_R}w_i^{2\Pi_i}x_l^{\Pi'_l}\right]. 
\fe
Here $F$, $J_\pm$, $J_R$, $J'_R$, $\Pi_i$, and $\Pi'_l$ are the Fermion number and charges under the Cartan generators of $SU(2)_\pm$, $SU(2)^R_+$, $SU(2)^R_-$, $U(N)$, $U(M_2)$. $\mathcal{H}_{\text{QM}}$ is the Hilbert space of the quantum mechanics. 

The Witten index can be computed by supersymmetric localization of the quantum mechanical path integral. As it has now become a standard technique, we will simply take formulas from the original references and direct the reader to \cite{Hwang:2014uwa}. The difference between the original set-up  and our set-up is the presence of $D0_1-D4_2$ strings, which are T-dual of $D2_1-D2_2$ strings. Since the other extra fields are all massive, the computation reduces to
\ie\label{theintegral}
\cI^{M_1,N}=\int[d\phi_I]Z^{M_1}_{D0-D0}(\phi_I,\E_+,\E_-,\tilde\E)Z^{M_1}_{D0-D4_1}(\phi_i,\E_+,\E_-,\tilde\E,a_l)Z^{M_1}_{D0-D4_2}(\phi_I,\E_+,\E_-,m_l),
\fe
where $\phi_i$ are Cartan elements of $U(M_1)$, $\E_\pm$, $\tilde\E$ are Cartan elements of $SU(2)_\pm$, $SU(2)^R_-$, and $a_i$, $m_l$ are Cartan elements of $U(N)$ and $U(M_2)$. 

$Z^{M_1}_{D0-D0}Z^{M_1}_{D0-D4_1}$ is given in \cite{Hwang:2014uwa,Kim:2016qqs}:
\ie\label{d0d0d0d4}
Z^{M_1}_{D0-D0}Z^{M_1}_{D0-D4_1}=\prod^{M_1}_{I,J=1}\frac{\sinh'\frac{\phi_I-\phi_J}{2}\sinh\frac{\phi_I-\phi_J+2\E_+}{2}}{\sinh\frac{\phi_I-\phi_J+2\E_1}{2}\sinh\frac{\phi_I-\phi_J+2\E_2}{2}}\prod^{M_1}_{I,J=1}\frac{\sinh\frac{\phi_{IJ}\pm \tilde\E-\E_-}{2}}{\sinh\frac{\phi_{IJ}\pm \tilde\E-\E_+}{2}}\prod^{M_1}_{I=1}\prod^N_{i=1}\frac{\sinh\frac{\phi_I-a_i\pm \tilde\E}{2}}{2\sinh\frac{\phi_I-a_i\pm\E_+}{2}},
\fe
where the prime on sinh indicates that sinh(x) is omitted when x = 0. 

We can read $Z^{M_1}_{D0-D4_2}$ from \cite[(3.6)]{Kim:2016qqs}. The difference here is that we do not have massless $D4_1-D4_2$ strings. The chiral multiplet $S$ maps to the twisted hypermultiplet in SQM.\footnote{\ Here by twisted hypermultiplet we mean a multiplet whose scalar transforms non trivially under $SU(2)^R_\pm$ but trivially under $SU(2)_\pm$.} The corresponding 1-loop determinant is 
\ie\label{d0d42}
Z^{M_1}_{D0-D4_2}=\prod^{M_1}_{I=1}\prod^{M_2}_{l=1}\frac{\sinh\frac{\phi_I-m_l\pm\E_-}{2}}{\sinh\frac{\phi_I-m_l\pm\E_+}{2}}.
\fe
After substituting the ingredients \eqref{d0d0d0d4}, \eqref{d0d42} in \eqref{theintegral}, we can use the JK prescription \cite{Jeffrey:aa,Benini:2013xpa}, which guides what poles to choose when one evaluates multidimensional complex integrals. Let us explain the contour of the integral \eqref{theintegral} following \cite{Kim:2016qqs}. We can classify the poles selected by the JK prescription in two types:
\ie\label{poles}
\cP^{1}:~&\phi_I-a_i+\E_+=0,\quad\phi_I-\phi_J\pm\E_-+\E_+=0,\\
\cP^{2}:~&\phi_I-m_l-\E_+=0,\quad\phi_I-\phi_J\pm\tilde\E-\E_+=0.
\fe
When evaluating the rank $M_1$ $\phi_I$ integral, one needs to choose $M_1$ poles out of \eqref{poles} and evaluate the residue integral. In other words, since we have two sets, if we can pick $M$ poles in $\cP^1$, we need to pick the rest $M_1-M$ poles in $\cP^2$. It was shown in \cite{Kim:2016qqs} that each of the first and second set of poles is classified by the N-colored Young diagrams $Y=\{Y_1,\cdots,Y_N\}$ with total size $M$ and $M_2$-colored\footnote{The case for $M_2=1$ was discussed in \cite{Kim:2016qqs}, but the statement can be extended for general $M_2$. For instance, quiver gauge theory with higher rank gauge nodes was discussed in \cite{Haouzi:2019jzk}.} Young diagrams $\tilde Y$ with total size $M_1-M$, respectively. Each of the contributions can be naturally understood as the D0 bound states on the $D4_1$ branes and on the $D4_2$ branes. With this information, we can organize the result of the integral as the following double summation over two sets of Young diagrams:
\ie\label{doublsum}
\cI^{M_1,N}=\sum_{Y}\sum_{\tilde Y}f_{Y,\tilde Y}(\E_+,\E_-,\tilde\E,a_i,m_l)
\fe
It is assumed in the above formula that we take $e^{\tilde\E}\ria0$ limit to decouple the twisted multiplet contribution, as we commented in the caption of Table \ref{table:D0-D4-D8}. This is a valid limit, as discussed in \cite{Kim:2016qqs}.

Compared to \cite{Kim:2016qqs}, we do not have a Fermi multiplet, which would come from $D4_1-D4_2$ strings. The corresponding 1-loop determinant is
\ie\label{fermi1-loop}
Z_{D4_1-D4_2}=\prod_{i=1}^N\prod_{j=1}^{M_2}2\sinh\frac{a_i-m_j}{2}.
\fe
Since $D4_1-D4_2$ strings are not related to D0 branes(i.e. \eqref{fermi1-loop} does not depend on $\phi_I$), this 1-loop determinant does not affect the $\phi_I$ integral. Hence, the structure of our integral \eqref{doublsum} is the same as that of \cite{Kim:2016qqs}.

Computing $\cI^{M_1,N}$ for each $M_1$, we can form a generating series by weighting each $\cI^{M_1,N}$ with the instanton counting parameter $Q_{1}$ with power $M_1$:
\ie\label{genseries1}
1+\sum_{M_1=1}Q^{M_1}_{1}\cI^{M_1,N}.
\fe
This gives a part of the $G_2$ instanton partition function that is associated to the fiber of the bundle 
\ie
\pi:\cM^{N,K}(M_1,M_2)\ria\cM^N(M_1).
\fe
We naturally conjecture that \eqref{genseries1} is equivalent to the half \eqref{halffactor} of the factorization formula in Proposition \ref{prop: factorization}, since both of them are labeled by the same information: two sets of Young diagrams $Y$ and $\tilde Y$.

Similarly, one can do the same computation for the second saddle of the $G_2$ instanton partition function that is represented by the following subquiver:
\begin{figure}[H]
\centering
\includegraphics[width=10cm]{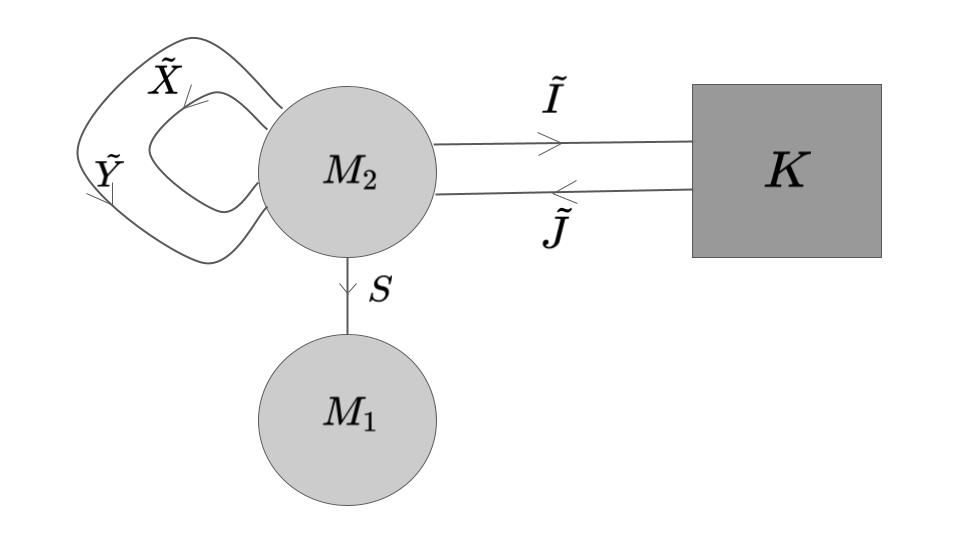}
  \vspace{-5pt}
\caption{}
\label{}
 \centering
\end{figure}
This is done by first choosing a coordinate frame where we find $D2_2$, $D6_2$ branes aligning in the 4,5,6 coordinate axis and on the other hand $D2_1$, $D6_2$ branes forming angles with 4,5,6 axis. Performing the same T-duality, $D2_2$ brane stack into $D0'$ brane stack. The new configuration is as follows
\begin{table}[H]
\centering
\begin{tabular}{|c|cccc|cccccc|}
\hline
   &  0 &  1 &  2 &  3 &  4 &  5 &  6 &  7 &  8 &  9 \\
\hline\hline
Background & & & & &A &B &B &A & B & B\\
\hline
$(D2_1\ria)D4'_2$& & & & & $-$ &$+$ &$+$ &$|$&$+$ &$+$\\ 
\hline
$(D6_1\ria)D8'$&$+$ &$+$&$+$&$+$&$-$& $+$& $+$& $|$&$+$&$+$\\ 
\hline
$(D2_2\ria)D0'$ & & & & &$+$& & && &\\
\hline
$(D6_2\ria)D4'_1$ &$+$ &$+$&$+$&$+$&$+$&&& &&\\
\hline
\end{tabular}
\caption{Brane configuration after the T-duality. A and B denote angles and B-field. For instance, D4'${}_2$ and D8' branes extend over a line in 47 plane.}
\label{table:D2, D6-brane}
\end{table}
Doing the exactly same computation as above, we can form a generating series by weighting each $\cI^{M_2,K}$ with the instanton counting parameter $Q_{2}$.
\ie\label{genseries2}
1+\sum_{M_2=1}Q^{M_2}_{2}\cI^{M_2,K}.
\fe
This gives a part of the $G_2$ instanton partition function that is associated to the fiber of the bundle 
\ie
\pi':\cM^{N,K}(M_1,M_2)\ria\cM^K(M_2).
\fe
Similarly to the other half, we expect that \eqref{genseries2} is equivalent to the other half of the factorization formula in Proposition \ref{prop: factorization}, since both of them are labeled by the same information: two sets of Young diagrams $Y$ and $\tilde Y$.

\subsection{A relation with 3d $\mathcal N=2$ theories with impurities}
We conclude this section with a remark about a relation between our setup and 3d $\mathcal N=4$ systems with supersymmetry breaking impurities. Exploiting the Taub-NUT radius to perform a T-duality with IIB, starting from equation \eqref{eq:IIAbraneconf}, we obtain the following setup:
\be
\begin{tabular}{c|cccccccccc}
&0&1&2&3&4&5&6&7&8&9\\
\hline
$N$ D5$_1$ & & $\bullet$ & $\bullet$ & $\bullet$ & $\bullet$ & $\bullet$ & $\bullet$ &  &  & \\
$M_1$ D3$_1$ & $\bullet$  &  &  &  & $\bullet$ & $\bullet$ & $\bullet$ &  &  & \\
$K$ D5$_2$ &  & $\bullet$ & $\bullet$ & $\bullet$  &  &  & $\mathbf{L}_{\boldsymbol\theta}$   \\
$M_2$ D3$_2$ & $\bullet$  &  &  &  &  &  & $\mathbf{L}_{\boldsymbol\theta}$   \\
$L$ NS5 & & & & & $\bullet$ & $\bullet$ & $\bullet$ & $\bullet$ & $\bullet$ & $\bullet$ \\
\end{tabular}
\ee
From this T-dual frame, it is interesting to notice a connection between our computation and 3d $\mathcal N=4$ theories with supersymmetry breaking impurities, which is somehow reminiscent of the analysis in \cite{Cherkis:2009hpw,Cherkis:2011ee,Dedushenko:2017tdw}. Indeed, the brane setup involving D3$_1$, D5$_1$, and NS5 is the famous Hanany-Witten brane engineering for 3d $\mathcal N=4$ systems \cite{Hanany:1996ie}. On top of that, we have the injection of the D3$_2$ and D5$_2$ stacks, which can be described as impurities localised at points from the perspective of the corresponding 3d $\mathcal N=4$ theories. The partition functions of M2 branes contributing to the 7d instanton partition function should have an interpretation also in terms of this kind of systems. In a sense, these would be captured by 3d $\mathcal N=4$ $S^3$ or $S^1 \times S^2$ parition function with the insertion of impurity operators.



\section*{Acknowledgements}

We thank Christopher Beem, Kevin Costello, I\~naki Garc{\'i}a Etxebarria, Dominic Joyce, Raeez Lorgat, Maxim Zabzine for useful discussion and email correspondences. We are especially grateful to Kevin Costello for asking a fascinating question that led us to initiate this project. The work of MDZ and RM has received funding from the European Research Council (ERC) under the European Union’s Horizon 2020 research and innovation programme (grant agreement No. 851931).  Research of JO was supported by ERC Grant numbers 864828 and 682608. Research at the Perimeter Institute is supported by the Government of Canada through Industry Canada and by the Province of Ontario through the Ministry of Economic Development $\&$ Innovation.

\medskip

\appendix

\section{A brief review of the twisted M-theory background}\label{app:costellobackground}
Twisted M-theory is eleven-dimensional supergravity background that is specified by three data $(g,C,\Psi)$: the metric, M-theory 3-form, and bosonic ghost for a certain component of local supersymmetry. The metric is in general a product metric of a hyper-Kahler 4-manifold $\cM_4^{HK}$ and a $G_2$ 7-manifold $\cM_7^{G_2}$. In the original work \cite{Costello:2016nkh}, the geometry was given by
\ie
\underbrace{\bC_z\times\bC_w}_{\cM_4^{HK}}\times\underbrace{\bR_t\times\bC_1\times\text{TN}_K}_{\cM_7^{G_2}}.
\fe

The generalized Killing spinor $\Psi$ imposes a topological and holomorphic twist on $\cM_4^{HK}$ and $\cM_7^{G_2}$ respectively. More precisely, we turn a nonzero bosonic ghost $\Psi$ \cite{Costello:2016mgj} for the local supersymmetry of the supergravity. We determine $\Psi$ in the same way we get the scalar supercharge in a twisted QFT \cite{Eager:2021ufo}. Start with an 11d Killing spinor $\Psi_{11}$ and decompose it into 7d and 4d components. Each of the component spinors decomposes into representations of $G_2$ and $SU(2)_-\times U(1)_+$, where $U(1)_+$ is the Cartan of $SU(2)_+$.
\ie
\Psi_{11}=\Psi_7\otimes\Psi_4=({\mathbf{1}_{G_2}\oplus7_{G_2})\otimes(1_{-1}\oplus1_{+1}\oplus2_0)}.
\fe
We keep a non zero bosonic ghost for the scalar $\Psi=(\mathbf{1_{G_2}\otimes1_{-1}})$. There are $\Psi_m$ and $\Psi_{\dot\A}$ such that 
\ie
\{\Psi,\Psi_m\}=P_m,\quad\{\Psi,\Psi_{\dot\A}\}=P_{-\dot\A}.
\fe
This implies that the dynamics in  $\cM^{G_2}_7$ direction becomes topological and that in $\cM^{HK}_4$ becomes holomorphic. 

We can deform $\Psi$ into $\Psi_\E$ \cite{Costello:2016nkh} such that
\ie
\Psi_\E^2=\E \vec{\cV},
\fe
where $\vec{\cV}$ is a Killing vector field on $\cM_7^{G_2}$.
With this deformation, we can incorporate an Omega deformation $\Omega_{\E_1}$  on $\bC_1$ and $\Omega_{\E_2}\times\Omega_{\E_3}$ on $TN_K\sim(\bC_2\times\bC_3)/\bZ_K$. $\E_i$ are deformation parameters and from now on we will denote $\Omega_{\E_i}$ deformed $\bC_i$ plane as $\bC_{\E_i}$.

Lastly, the M-theory 3-form is 
\ie
C=\E_2V^d d\bar zd\bar w,
\fe
where $V^d$ is a linear dual 1-form of the vector field $V$, which generates a rotation along the Taub-NUT circle $S^1_{TN}$.

As a result of the Omega deformation, we can expect that a localization happens in the twisted M-theory. We can explicitly see this by going to the type IIA frame, by reducing along $S^1_{TN}$. The Taub-NUT geometry converts into $K$ D6 branes and the closed string states decouple from the Q-closed spectrum as discussed in \cite{Costello:2016nkh}. Moreover, due to $\Omega_{\E_1}$, the 7d SYM living on $K$ D6 branes localizes to the 5d $U(K)$ topological holomorphic Chern-Simons theory \cite{Costello:2018txb}:
\ie
\int_{\bC_z\times\bC_w\times\bR_t}dz\wedge dw\wedge\left(AdA+\frac{2}{3}A\star_{\E_2} A\star_{\E_2} A\right).
\fe
Here, $\star_{\E_2}$ is a wedge product deformed by the non-commutative B-field $B=\E_2 d\bar zd\bar w$, which descends from the M-theory 3-form. This is a typical Moyal product, which is defined between two holomorphic functions as
\ie
f\star_{\E_2}g=fg+\E_2\frac{1}{2}\E_{ij}\frac{\pa}{\pa_{z_i}}f\frac{\pa}{\pa_{z_j}}g+\cdots.
\fe
Another important difference with the usual 5d Chern-Simons action is that the 5d gauge field has only three components:
\ie
A=A_tdt+A_{\bar z}d\bar z+A_{\bar w}d\bar w,
\fe
due to the presence of the holomorphic top-form $dzdw$ in the action.

Due to the non-commutativity, the naive gauge symmetry algebra $\mathfrak{gl}_{K}$ is modified to $\mathfrak{g}=\text{Diff}_{\E_2}\bC_z\otimes\mathfrak{gl}_{K}$. Without the quantum deformation, i.e. $\E_1 = 0$, the symmetry algebra $A_{\text{classical}}=A_{\E_1=0,\E_2}$ of the 5d CS theory is generated
by the Fourier modes of the ghost $c[m, n]$ for the gauge symmetry, together with the BRST differential $\D$.
As a graded associative algebra, $A_{\text{classical}}$ is isomorphic to $C^*(\mathfrak{g})$, which is the Chevalley-Eilenberg algebra of cochains on the Lie algebra $\mathfrak{g}$.

One of the major achievements of \cite{Costello:2017fbo} was to use Kozsul duality $C^*(\mathfrak{g})=U(\mathfrak{g})$ and identify the 5d CS algebra of operators with the protected operator algebra $\cA^{(N)}_K$ of $N$ M2 branes(on $\bR_t\times\bC_{\E_1}$) probing the twisted M-theory background in the large $N$ limit. Moreover, there is a surjective map $s$ \cite{Costello:2017fbo}
\ie
s:\mathrm{U}_{\epsilon_1}(\mathfrak{gl}(K)\otimes \mathrm{Diff}_{\epsilon_2}(\mathbb{C}))\rightarrow \mathcal A^{(N)}_K,
\fe
which is compatible with a sequence of surjective maps $s_i:\mathcal A^{(i+1)}_K\to \mathcal A^{(i)}_K$. Moreover, the intersection of the kernels of $\mathrm{U}_{\epsilon_1}(\mathfrak{gl}(K)\otimes \mathrm{Diff}_{\epsilon_2}(\mathbb{C}))\to \mathcal A^{(N)}_K$ for all $N$ is zero. In this sense, one can define $\mathrm{U}_{\epsilon_1}(\mathfrak{gl}(K)\otimes \mathrm{Diff}_{\epsilon_2}(\mathbb{C}))$ as a {\textit{large N limit}} of $\mathcal A^{(N)}_K$, and denote it by $\mathcal{A}_K$ by dropping the superscript $(\infty)$.

Given the abstract description of $\cA^{(N)}_K$ and its relation with the 5d CS algebra of operators, let us provide the explicit description of $\cA^{(N)}_K$. One convenient UV descriptions of the M2 branes worldvolume theory is 3d $\cN=4$ $U(N)$ gauge theory with 1 adjoint hypermultiplet(with scalars $X,Y$) and $K$ fundamental hypermultiplets(with scalars $I,J$) \cite{Bashkirov:2010kz}. Due to the topological twist applied to $\bR_t\times\bC_{\E_1}\subset\cM^{G_2}_7$, we pass to a particular Q-cohomology that only captures the Higgs branch chiral ring. This consists of gauge-invariant operators made of $X,Y,I,J$, divided by the ideal generated by the F-term relation,
\ie\label{fterm}
X^a_cY^c_b-X^c_bY^a_c+I_bJ^a=\E_2\D^a_b,
\fe
where we only presented gauge indices $a,b,c,d$, suppressing flavor indices.

$\Omega_{\E_1}$ deforms the chiral ring into an algebra $\cA_K$ \cite{Yagi:2014toa,Bullimore:2016nji,Bullimore:2016hdc} by making Poisson brackets into commutators
\ie\label{basiccom}
\left[X^a_b,Y^c_d\right]=\E_1\D^a_d\D^c_b,\quad[J^b,I_a]=\E_1\D^b_a,
\fe 
and the theory localizes on one-dimensional topological quantum mechanics \cite{Mezei:2017kmw,Dedushenko:2016jxl}
\ie
\frac{1}{\E_1}\int_{\bR_t}\text{Tr}[\E_2A_t+XD_tY+JD_tI]dt.
\fe
Note that $\epsilon_1$ is a deformation parameter, and $\epsilon_2$ is an FI parameter of the 3d $\cN=4$ gauge theory. 

The algebra $\cA^{(N)}_K$ is generated by
\ie\label{algebraelements}
t_{m,n}=\frac{1}{\E_1}\text{STr}X^mY^n,
\fe
where $STr[\bullet]$ means to take a trace of a symmetrization of a polynomial $\bullet$. The generating relations of the algebra was proposed in \cite{Gaiotto:2020vqj}, proved in \cite{Oh:2020hph}, and used in \cite{Gaiotto:2020vqj,Hatsuda:2021oxa} to compute a large set of correlation functions in 3d $\cN=4$ ADHM gauge theories.

\section{The D2 matrix model from string quantization}\label{app:stringquant}
In this appendix we review the computation that leads to the equation describing the moduli space of vacua for the matrix model supported at the intersection of the two D2 brane stacks in our local model arising from the $\mathcal X_{N,K}$ geometry. We will use the language of 3d supersymmetry because the multiplets we obtain organise naturally accordingly. However, these are all fields supported at a point (the point of intersection of the two D2 brane stacks).

\subsection{Supersymmetry}\label{sec:supersymmetry}
We can determine the supersymmetry preserved by the intersecting D6 branes \cite{Berkooz:1996km} by solving the following equations of 10d spinor $(\E_L,\E_R)$ equation($\E_L$ and $\E_R$ are 10d Majorana-Weyl spinors with different chiralities).
\ie\label{susyconditiond6}
\E_L&=\Gamma^{0123456}\E_R,\\
\E_L&=\Gamma^{01234'5'6'}\E_R,
\fe
where 456, 4'5'6' denote the directions along two stacks of D6 branes in the 456789 direction and $\Gamma^{i_1\ldots i_n}$ denote the fully anti-symmetrized products of Gamma matrices $\Gamma^{i_1},\ldots,\Gamma^{i_n}$. 

We need to remember that there is a B-field turned on 0123 directions. Let us parametrize it as
\ie
B_{01}=-B_{10}=b_1,\quad B_{23}=-B_{32}=b_2.
\fe
and set $b_a=\tan\pi v_a$. Hence,
\ie\label{bfieldrange}
-\frac{1}{2}<v_a<\frac{1}{2}.
\fe
By 
\cite{Witten:2000mf}, the Gamma matrix equation above changes to
\ie\label{susyconditiond62}
\E_L&=\mathcal{B}_{0123}\Gamma^{0123456}\E_R,\\
\E_L&=\mathcal{B}_{0123}\Gamma^{01234'5'6'}\E_R,
\fe
where
\ie
\mathcal{B}_{0123}=\exp\pi\left((v_1+1/2)\Gamma^0\Gamma^1+(v_2+1/2)\Gamma^2\Gamma^3\right)
\fe
For now, this change will not affect the calculation, since both of \eqref{susyconditiond6} and \eqref{susyconditiond62} reduce into
\ie\label{spinoreqn}
\Gamma^{456}\E_R=\Gamma^{4'5'6'}\E_R,
\fe
where we used $(\Gamma^{ij})^2=1$. Therefore, it only constrains the 6d part $\E_{6d}$ of the 10d spinor $\E_R=\E_{10d}=\E_{4d}\otimes\E_{6d}$ whose eigenvalues under the action of
\ie
\Gamma^{01},\quad\Gamma^{23},\quad\Gamma^{47},\quad\Gamma^{58},\quad\Gamma^{69}.
\fe
is $(\pm,\pm,\pm,\pm,\pm)$. Those are raised and lowered by the action of
\ie
&\Gamma^{0\pm}=\frac{1}{2}(\pm\Gamma^0+\Gamma^1),~\Gamma^{1\pm}=\frac{1}{2}(\Gamma^{2}\pm i\Gamma^{3}),\\
&\Gamma^{2\pm}=\frac{1}{2}(\Gamma^{4}\pm i\Gamma^{7}),~\Gamma^{3\pm}=\frac{1}{2}(\Gamma^{5}\pm i\Gamma^{8}),~\Gamma^{4\pm}=\frac{1}{2}(\Gamma^{6}\pm i\Gamma^{9}).
\fe
Let $\mathbb{R}$ be the $SO(6)$ rotation that takes the first stack of D6 branes to the second stack of D6 branes. 
\ie
\mathbb{R}=\text{diag}(e^{i\theta_1},e^{-i\theta_1},e^{i\theta_2},e^{-i\theta_3},e^{i\theta_3},e^{-i\theta_3}).
\fe
Then we have
\ie
\Gamma^{0123456}=\mathbb{R}\left(\Gamma^{01234'5'6'}\right)\mathbb{R}^{-1}.
\fe
The solution for \eqref{spinoreqn} only exists when $\mathbb{R}$ belongs to $SU(3)$ subgroup of the $SO(6)$ that rotates 456789. The $SU(3)$ condition is
\ie\label{anglesusy}
\theta_1\pm\theta_2\pm\theta_3=0~\text{mod}~2\pi.
\fe
There are 4 spinors $\E_{10d}=\E_{4d}\otimes\E_{6d}$ such that $\E_{4d}=(\pm,\pm)$, $\E_{6d}$ with one choice of signs out of 8 possible choices $(\pm,\pm,\pm)$. In other words, a 4d theory along 0123 directions preserves a chiral $\cN=1$ supersymmetry.

Now, introduce $D2_1$ branes along 456 directions. This gives another spinor equation:
\ie
\E_L=\Gamma^{456}\E_R.
\fe
Coupling with the first line of \eqref{susyconditiond62}, we get
\ie\label{4dn1chiral1}
\mathcal{B}_{0123}\Gamma^{0123456}\E_R=\Gamma^{456}\E_R
\fe
Decomposing $E_R=\E_{4d}\otimes\E_{6d}$, we get
\ie
(\mathcal{B}_{0123}\Gamma^{0123}\E_{4d})\otimes(\Gamma^{456}\E_{6d})=\E_{4d}\otimes(\Gamma^{456}\E_{6d})
\fe
Hence, it reduces to
\ie\label{0123}
\mathcal{B}_{0123}\Gamma^{0123}\E_{4d}=-\E_{4d}
\fe
It has a solution only if 
\ie
\exp\pi(\pm i(v_1+1/2)\pm i(v_2+1/2))=-1.
\fe
In other words, if $v_1$, $v_2$ satisfy
\ie\label{v1v2}
v_1=\pm v_2,
\fe
we have 
\ie
\Gamma^{0123}\E_{4d}=\E_{4d}.
\fe
Note, however, by \cite{Witten:2000mf} only $v_1=-v_2$ is allowed.\footnote{
Our B-field is anti-self-dual, which is consistent with $v_1+v_2=0$, as shown in \cite{Witten:2000mf}.
\ie\nonumber
\star B=\E_2\star(d\bar z_1\wedge d\bar z_2)=\E_2\star\left[(dx_1-idx_3)\wedge(dx_2-idx_4)\right]=-\E_2d\bar z_1\wedge d\bar z_2=-B
\fe
}
In any case, out of 4 spinors $\E_{4d}=(\pm,\pm)$, only 2 spinors survive: $\E_{4d}=(+,+),(-,-)$; the number of the preserved supersymmetry is independent of $B$ as shown in \cite{Seiberg:1999vs}.

Similarly, consider $D2_2$ branes along 4'5'6' directions. This gives
\ie
\E_L=\Gamma^{4'5'6'}\E_R.
\fe
Coupling with the second line of \eqref{susyconditiond6}, we get \eqref{0123}. Therefore, out of 4 spinors $\E_{4d}=(\pm,\pm)$, only 2 spinors survive: $\E_{4d}=(+,+),(-,-)$. We are allowed have both $D2_1$ and $D2_2$ branes and preserve the same amount of supersymmetry, since both lead to the same equation \eqref{0123}.

Therefore, we conclude that D2/D6/D6' configuration preserves 2 supercharges. In D2 brane worldvolume, this is 3d $\cN=1$ supersymmetry. 

Before doing the sting quantization, let us review supermultiplets of 3d $\cN=1$ SUSY \cite{Gates:1983nr}. First, $\cN=1$ vector multiplet $\cV$ consists of a gauge field $A_\mu$ and a Majorana fermion $\ld_\A$
\ie
\cV=(A_\mu,\ld_\A),\quad\text{where }\mu=1,2,3,~\A=+,-.
\fe
Second, $\cN=1$ scalar multiplet $\Phi$ consists of a real scalar $\phi$, a Majorana fermion $\psi_\A$, and an auxiliary field $D$
\ie
\Phi=(\phi,\psi_\A,D).
\fe

\subsection{String quantization}\label{sec:quantization}
We will follow \cite{Nekrasov:2016gud} to account for the B-field background-- recall we have turned B-field on $TN_1$\footnote{Note that the D6-branes at angles configuration is equivalent to D0-D6 brane system with B-field on ND directions \cite{Witten:2000mf}.}. To quantize open strings, we need to mode-expand worldsheet bosons $Z$ and fermions $\Psi$ that respect two types of boundary conditions at two boundaries. In the presence of the B-field, the equation of motion is given by \cite{Abouelsaood:1986gd}
\ie
\pa_{++}Z=e^{-2\pi i\nu}\pa_{--}Z&,\quad\Psi^+=e^{2\pi i\nu}\Psi^-\text{ at }\sigma=0,\\
\pa_{++}Z=e^{-2\pi i\mu}\pa_{--}Z&,\quad\Psi^+=e^{2\pi i\mu}\Psi^-\text{ at }\sigma=\pi,\\
\fe
Solving the equations above, we see the worldsheet boson $Z$ and R-sector fermion $\Psi$ have modes $\bZ+\theta$ with $\theta=\mu-\nu$. The NS-sector fermions have modes $\bZ+\E$ with $\E=\theta+\frac{1}{2}$. We will denote the Neumann, Dirichlet, Twisted(by the B-field effect) boundary conditions as {\textbf{{N}}, {\textbf{{D}}, {\textbf{{T}}.

To arrange the spectra by their energy, we need to find a zero-point energy. The NS sector zero-point energy is
\ie
E_0^{NS}=\frac{1}{8}-\frac{1}{2}||\theta|-\frac{1}{2}|.
\fe
The first excited state has energy $E_1^{NS}=E_0^{NS}+|\E|$ if $-\frac{1}{2}\leq\E\leq\frac{1}{2}$ and $E_1^{NS}=E_0^{NS}+|1-\E|$ if $\frac{1}{2}<\E<\frac{3}{2}$. The R sector zero-point energy $E_0^R=0$. The first excited state energy is $E_1^R=|\theta|$ for $0\leq|\theta|\leq\frac{1}{2}$ and $1-|\theta|$ for $\frac{1}{2}\leq|\theta|\leq1$.

Let us recall the brane configuration from equation \eqref{eq:IIAbraneconf}. 
We can quantize strings that stretch between two stacks of D-branes. Start with $D2_1-D2_1$ strings. Although we have determined the minimal amount of supersymmetry that the brane configuration preserves is 2, we will later see that the genuine 3d $\cN=1$ multiplet is massive and decoupled in the low energy. Therefore, we will supplement the discussion with the 3d $\cN=2$ representation. Moreover, we will sometimes find helpful to reorganize the $\cN=2$ supermultiplets into the 3d $\cN=4$ supermultiplets, even though we do not have such an enhanced supersymmetry. There will be just one genuine 3d $\cN=2$ massless multiplet after all.\\
\newline
{\textbf{D2-D2 string NS sector}}: We know $E_0^{NS}=-\frac{1}{2}$, so to get the massless states, we need to act with NS fermion oscillators $d^{\mu}_1{}^\dagger$. For $\mu=4,5,6$, which are {\textbf{NN}} directions, $d^{\mu}_1{}^\dagger$ gives rise to three states, which consist of the 3d $G$ gauge field $A_\mu$. For other d's, which are {\textbf{TT}} and {\textbf{DD}}, they give rise to 7 scalars $\phi_a$ valued in adj$(\mathfrak{g})$ of G. \\
\newline
{\textbf{D2-D2 string R sector}}: R-sector vacuum energy is always zero. Since we have 10 NN+DD directions, we can use all 10 zero modes to act on the R-sector vacuum state and form 32 dimensional ground state. The GSO projection projects out half of 32 and the remaining 16 fermionic states pair up with the bosons determined from the NS sector to complete the 3d $\cN=1$ supermultiplet.

Therefore, we see
\ie
D2_1-D2_1\text{ strings }\equiv~\cV\text{ and }7~\Phi_a
\fe
One can split 7 $\Phi_a$ into $\{\Phi_7,\Phi_8,\Phi_9\}$ and $\{\Phi_0,\Phi_1,\Phi_2,\Phi_3\}$. In the 3d $\cN=2$ notation, one of the scalars in the first set can be identified with the real scalar $\sigma$ and the other two scalars in the first set can be thought of as the complex scalar $\varphi$.  On the other hand, one can split the 4 scalars in the second set into 2 complex scalars. In the presence of $\cN=2$ SUSY, each of them forms $\cN=2$ chiral multiplets $X$, $Y$.\\
\newline
{\textbf{D2-D6 string NS sector}}: Let us denote two complex directions in 4 ND directions as $a$ and $b$. We first need to derive NS zero point energy. Recall B-field is only turned on the 4 ND directions. Hence, the zero point energy in these directions is
\ie\label{zero1}
\frac{1}{8}-\frac{1}{2}|v_a|+\frac{1}{8}-\frac{1}{2}|v_b|=\frac{1}{4}-\frac{1}{2}(|v_a|+|v_b|).
\fe
Other directions do not have twisted boundary conditions by the B-field. Taking two directions out of three D2 world volume directions as the lightcone, we can compute the zero point energy for the {\textbf{NN}} and {\textbf{DD}} directions:
\ie\label{zero2}
\frac{1}{8}-\frac{1}{4}+\frac{1}{8}-\frac{1}{4}=-\frac{1}{4}.
\fe
Summing \eqref{zero1} and \eqref{zero2}, we get
\ie
E^{NS}_0=-\frac{1}{2}(|v_a|+|v_b|).\fe 
By applying a suitable oscillator that raises the energy by $|v_a|$ and $|v_b|$, we get 4 states with energies $\frac{1}{2}(\pm v_a\pm v_b)$. 2 out of 4 states are projected out by the GSO projection and the remaining ones have energy $\pm\frac{1}{2}(v_a-v_b)$. Those are two complex scalars transforming in $(N,M_1)$ of $U(N)\times U(M_1)$ of the D6 and D2 gauge groups with masses
\ie
m^2=\pm\frac{1}{2}(v_a-v_b).
\fe
One of the two is tachyonic, which needs further treatment later.
Since for 3d $\cN=1$ the elementary multiplet is a real scalar multiplet, it indicates that there are 4 scalars $\Phi^I$.\\
\newline
{\textbf{D2-D6 string R sector}}: As before, the zero point energy of R-sector is zero. For this case, we have 6 zero modes from 6 NN+DD directions. Therefore, we have $8=2^3$ dimensional ground states. We can label those as $|\A,A,\dot A\ra$, where $\A$ is 3d spinor index, $A,\dot A$ are the spinor index of $SU(2)\times SU(2)=SO(4)$ that rotates the transverse direction of $D2$ inside D6 worldvolume. Then, these fermionic massless fields make a supersymmetric completion of the bosons obtained from the NS sector analysis.

Therefore, we get 4 scalar multiplets:
\ie
D2_1-D6_1\text{ strings }\equiv~4~\Phi^I.
\fe
In the presence of 3d $\cN=2$ supersymmetry, one can recombine the 4 real scalars into 2 complex scalars $I$, $J$. Each of them is a part of the 3d $\cN=2$ fundamental chiral multiplet.\\
\newline
{\textbf{Equivalence between B-field and Angle configuration}}\\
\newline
It is helpful to recognize the equivalence between B-field background and D-branes at angles \cite{Balasubramanian:1996uc} to unify the background in a single frame before we do $D2_1-D6_2$ quantization, whose ND directions involve both B-field background and the angle background.

B-field background along $2n$ directions is T-dual to the D-branes at angles along $2n$ direction. For instance, our $D6_1/D6_2$ at angles configuration is T-dual to $D0/D6$ brane system with B-field on 6 ND directions.

To see this, let us look at the open string boundary condition on 456789 directions in the original $D6_1/D6_2$ system. 
\ie
&\pa_nX^{4,5,6}=\pa_\tau X^{7,8,9}=0,\\
&\pa_n\left[\text{cos}\varphi_{ij}X^i+\text{sin}\varphi_{ij}X^j\right]=0,\\
&\pa_\tau\left[-\text{sin}\varphi_{ij}X^i+\text{cos}\varphi_{ij}X^j\right]=0,
\fe
where $i\in\{4,5,6\}$ and $j\in\{7,8,9\}$. 

T-dualizing along $X^4$, $X^5$, $X^6$ modifies the equations as
\ie
\pa_\tau X^{4,5,6,7,8,9}&=0,\\
\text{tan}\varphi_{ij}\pa_nX^{j}+\pa_\tau X^{i}&=0,\\
\text{tan}\varphi_{ij}\pa_nX^i-\pa_\tau X^j&=0.
\fe
These boundary conditions are exactly that of D0 branes bounded on D6 branes with B-field on 6 ND directions. The explicit dictionary is the following:
\ie
B_{ij}=-B_{ji}=\text{tan}\varphi_{ij}
\fe
Recall that we had a useful parametrization for the B-field to utilize in the string quantization
\ie
B_{i,i+3}=b_{i-2}=\text{tan}\pi v_{i-2}.
\fe
Hence, we can identify 
\ie
\varphi_{i,i+3}/\pi=v_{i-2}.
\fe
\newline
{\textbf{D2${}_1-$D6${}_2$, D2${}_2-$D6${}_1$ strings NS sector:}} First of all, this configuration preserves 2 supersymmetries, which are the minimal amount among all configuration. We can see this from the following Gamma matrix exercise:
\ie
\Gamma^{4'5'6'}\E_l=\E_r,\quad\Gamma^{0123456}\E_l=\E_r~\Rightarrow~\Gamma^{0123456}\E_l=\Gamma^{4'5'6'}\E_l.
\fe
Decompose $\E_l=(\E_l)_4\otimes(\E_l)_6$ then the above becomes
\ie
\left(\Gamma^{0123}(\E_l)_4\right)\otimes\left(\Gamma^{456}(\E_l)_6\right)=(\E_l)_4\otimes\left(\Gamma^{4'5'6'}(\E_l)_6\right)~\Rightarrow~\begin{cases}\Gamma^{0123}(\E_l)_4=(\E_l)_4\\
\Gamma^{456}(\E_l)_6=\Gamma^{4'5'6'}(\E_l)_6.
\end{cases}
\fe
We know from \cite{Berkooz:1996km} the second equation gives one solution and the first equation gives two solutions for $(\E_l)_4$: $(+,+)$ or $(-,-)$.

Let us resume the $D2_1-D6_2$ quantization. We T-dualize 4, 5, 6 directions to convert the D branes intersecting with angles into transversely intersecting D branes with B-field background. We will parametrize the B-field background by $v_2$, $v_3$, $v_4$. Now, the $D2_1-D6_2$ configuration becomes $D(-1)_1-D9_2$: D-instantons probing D9 brane worldvolume with 4 ND directions with the original B-field and the remaining 6 ND directions with the new B-field. In this unified background, we compute the NS zero point energy. First, start with the 0123 direction where the B-field is applied:
\ie
\frac{1}{8}-\frac{1}{2}|v_0|+\frac{1}{8}-\frac{1}{8}|v_1|=\frac{1}{4}-\frac{1}{2}(|v_0|+|v_1|).
\fe
Second, for 456789 directions with the new B-field, we have
\ie
\frac{1}{8}-\frac{1}{2}|v_2|+\frac{1}{8}-\frac{1}{8}|v_3|+\frac{1}{8}-\frac{1}{8}|v_4|=\frac{3}{8}-\frac{1}{2}(|v_2|+|v_3|+|v_4|).
\fe
Summing them up, we have
\ie
E^{NS}_0=\frac{5}{8}-\frac{1}{2}(|v_0|+|v_1|+|v_2|+|v_3|+|v_4|).
\fe
We produce the excited states by acting NS fermion oscillators that increase energy by $|v_i|$ on the NS vacuum state. First 32 states have energies
\ie\label{NSzeroD2D6}
\frac{5}{8}+\frac{1}{2}(\pm v_0\pm v_1\pm v_2\pm v_3\pm v_4).
\fe
Recall that there is a list of constraints:
\ie\label{listofconstraints}
v_0+v_1&=0,\\
\pi(v_2\pm v_3\pm v_4)&=0\text{ mod }2\pi,\\
-\frac{1}{2}<&v_i<\frac{1}{2},
\fe
where each of the conditions comes from \eqref{v1v2}, \eqref{anglesusy}, \eqref{bfieldrange}. Moreover, by the construction of the twisted M-theory, we have a small non-commutativity parameter
\ie
v_0=-v_1=\E_2\ll1.
\fe
As for the second constraint in \eqref{listofconstraints}, we have a freedom to fix 
$v_2=v_3=v_4=\frac{2}{3}$. To be consistent with the third line, we shift $v_2$, $v_3$, $v_4$ by the period of the tangent function, 1:
\ie
v_2=v_3=v_4=\frac{2}{3}-1=-\frac{1}{3}.
\fe
Hence, \eqref{NSzeroD2D6} becomes
\ie\label{32states}
\frac{5}{8}+\frac{1}{2}(\pm \E_2\pm \E_2\pm \frac{1}{3}\pm \frac{1}{3}\pm \frac{1}{3}).
\fe
Notice that only the minimal energy configuration can potentially be negative or zero energy depending on the value of $\E_2$ being greater or equal to $\frac{1}{8}$, since
\ie
E_0=\frac{5}{8}+\frac{1}{2}(- \E_2- \E_2- \frac{1}{3}- \frac{1}{3}- \frac{1}{3})=\frac{1}{8}-\E_2.
\fe
From now, we will denote 32 states by the signs in \eqref{32states}, e.g. $(\pm,\pm,\pm,\pm,\pm)$.

Now, let us consider GSO projection
\ie
\Gamma_{GSO}=\frac{1}{2}(1+(-1)^{F_{NS}}).
\fe
We assign $F_{NS}=1$ to the NS vacuum state(the minimal energy configuration). In the above convention, it corresponds to the state $(-,-,-,-,-)$. Then, all states with even number of `+' are projected out, since their $(-1)^{F_{NS}}=-1$. Therefore, $|E_0\ra$ is projected out and we do not have any massless or tachyonic state in the spectrum. In other words, in the low energy limit that we are interested in, there is no interesting state coming from the NS sector of $D2_1-D6_2$ or $D2_2-D6_1$ strings. 

Still, there are massive scalars. The lightest modes are from 0123 directions, due to the relative strength of the B-field. After the GSO projection, we get 2 bosons. \\
\newline
{\textbf{D2${}_1-$D6${}_2$ and D2${}_2-$D6${}_1$ strings R sector:}} Since two D-branes are fully transverse to each other, there are no NN or DD directions in this configuration. Hence, there are no fermionic zero modes. Since R-sector zero point energy is zero, there is still one zero energy state, the R-sector vacuum. However, the GSO projection projects this state out. Hence, we do not obtain any massless fermionic state. 

Still, there are massive fermions that pair up with the massive scalars determined from the NS sector to complete the two lightest massive supermultiplet.\\
\newline
{\textbf{D6${}_1$-D6${}_2$ strings:}} We already know that they give rise to the 4d $\cN=1$ bi-fundamental chiral multiplets at the 4d intersection\cite{Berkooz:1996km}. Since we are only interested in the 3d worldvolume theory of either $D2_1$ or $D2_2$ branes, we treat the D6 branes as a heavy background. Hence, $D6_1-D6_2$ strings will not play any important role in our story. \\
\newline
{\textbf{D2${}_1$-D2${}_2$ strings:}} We can T-dualize the $D6_1-D6_2$ configuration in the 4 directions of intersection and arrive at $D2_1-D2_2$ configuration. Since T-duality does not change the amount of preserved supersymmetry, the $D2_1-D2_2$ string quantization will yield the dimensional reduction of the 4d $\cN=1$ bifundamental chiral multiplet determined above and live at the intersection of $D2_1$ and $D2_2$. It smears out the 3d worldvolume of both $D2_1$ and $D2_2$ branes. Hence, we will treat it as 3d $\cN=2$ bifundamental chiral multiplet $S$, which consists of two $\cN=1$ scalar multiplets $S_1$ and $S_2$. As we will see later, $S$ will provide an edge between two conventional ADHM quivers and make our moduli space more interesting.  \\
\newline
{\textbf{Summary}}\\
Let us summarize the result of the string quantization of our brane system in the following
\ie\label{fieldcontent}
D2_1-D2_1\text{ strings }&\equiv~\text{3d }\cN=1\text{ vector multiplets }\cV_2\text{ and }7~\text{scalar multiplets }\Phi_a\\
D2_2-D2_2\text{ strings }&\equiv~\text{3d }\cN=1\text{ vector multiplets }\cV_2\text{ and }7~\text{scalar multiplets }\Phi'_a\\
D2_1-D2_2\text{ strings }&\equiv~\text{3d }\cN=1\text{ bifundamental scalar multiplets } S_1,~S_2\\
D2_1-D6_1\text{ strings }&\equiv~\text{3d }\cN=1~4~\text{scalars multiplets }\Phi^I\\
D2_2-D6_2\text{ strings }&\equiv~\text{3d }\cN=1~4~\text{scalars multiplets }\Phi'{}^I\\
D2_1-D6_2\text{ strings }&\equiv~\text{3d }~\cN=1~\text{massive scalar multiplet }\\
D2_2-D6_1\text{ strings }&\equiv~\text{3d }~\cN=1~\text{massive scalar multiplet }\\
D6_1-D6_1\text{ strings }&\equiv~\text{7d }U(N)\text{ vector multiplet}\\
D6_2-D6_2\text{ strings }&\equiv~\text{7d }U(K)\text{ vector multiplet}\\
D6_1-D6_2\text{ strings }&\equiv~\text{4d }\cN=1\text{ bifundamental chiral multiplet}
\fe
For our purpose that will be described in the next subsection, we can focus on the strings that are associated with the 3d massless fields: $D2_i-D2_i$, $D2_i-D6_i$, $D2_1-D2_2$ strings. Considering only the massless states, we get 3d $\cN=2$ theory on the $D2$ branes.
\subsection{Deriving the equations for the moduli space}\label{sec:Dterm}
Given the field content of the 3d $\cN=2$ theory on each of the D2 brane stacks, we can proceed to study the 
matrix model supported at the point of intersection. It is useful first to recall our situation. We have converted the $G_2$ geometry and the $G_2$ instantons into the D6 brane and D2 brane stacks. Hence, the relevant moduli space $\cM$ is the moduli space of the intersecting D2 brane stacks fluctuating inside the intersecting D6 brane stacks. The directions of the fluctuation are 0123. Therefore, we need to focus on the quantum fields that arise from the quanta of the strings that correspond to 0123 directions. We have denoted those fields as $\Phi_{0,1,2,3}$, $\Phi^{0,1,2,3}$, $\tilde\Phi_{0,1,2,3}$, $\tilde\Phi^{0,1,2,3}$. We also need to include the chiral bifundamental $S$ that is supported at the intersection of the D2 stacks.

In terms of 3d $\cN=2$ superfields, $\cM$ is parametrized by the scalar components of  
\ie
X=\Phi_1+i\Phi_2,~Y=\Phi_3+i\Phi_4,&~I=\Phi^1+i\Phi^2,~J=\Phi^3+i\Phi^4,\\
\tilde X=\tilde\Phi_1+i\tilde\Phi_2,~\tilde Y=\tilde\Phi_3+i\tilde\Phi_4,&~\tilde I=\tilde\Phi^1+i\tilde\Phi^2,~\tilde J=\tilde\Phi^3+i\tilde\Phi^4,\\
S=S_1&+iS_2.
\fe
The moduli space is given by the zeros of the scalar potential, which is supposed to be
\ie
V=\frac{g^2}{2}\text{Tr}D^2+\ldots,
\fe
where $D$ is the auxiliary field in the off-shell multiplets. The zeros of $V$ is obtained by setting $D=0$. Our goal is to derive a set of D-term equations in the 3d $\cN=2$ system that fully specifies the moduli space. To do that, the first step is to figure out all auxiliary fields in the supermultiplets. 

Let us start from $D2_1-D2_1$ strings:
\ie
\cV=&(A_\mu,\ld_\A),\\
\Phi_{7}=&(\phi_{7},\psi_{7,\A},D_7),\quad\Phi_{8}=(\phi_{8},\psi_{8,\A},D_8),\quad\Phi_{9}=(\phi_{9},\psi_{9,\A},D_9),\\
\Phi_0=&(\phi_0,\psi_{0,\A},D_0),\quad\Phi_1=(\phi_1,\psi_{1,\A},D_1),\quad\Phi_2=(\phi_2,\psi_{2,\A},D_2),\quad\Phi_3=(\phi_3,\psi_{3,\A},D_3).
\fe
We will think of $\phi_7$ as the real scalar $\sigma$ and $(\phi_8,\phi_9)$ as the complex scalar $\varphi$ in 3d $\cN=2$ notation. We find a triplet $(D_7,D_8,D_9)$ of the auxiliary D's. Combine them into $D_{ab}$, where $a,b$ are SU(2) indices. In the last line, we find 4 scalars $\phi_0$, $\phi_1$, $\phi_2$, $\phi_3$. In $\cN=4$ notation, we can combine $X=\phi_0+i\phi_1$, $Y=\phi_2+i\phi_3$ and collect them into a fictitious\footnote{Even though there is no such $SU(2)$ symmetry in the 3d $\cN=2$ system, this formal bookkeeping tool enables us to present the material in a compact way.} $SU(2)$ doublet $\mathcal{Q}=(X,Y)$. There is a coupling between $\mathcal{Q}_a$ and $D_{ab}$:
\ie
\text{Tr }\mathcal{Q}^aD_{a}^b\mathcal{Q}_b.
\fe
Also, there is a quartic term of $D$:
\ie
\text{Tr }|D_a|^2.
\fe

From $D2_1-D6_1$ strings, we got
\ie
\Phi^0=(\phi^0,\psi^0_\A,G^0),\quad\Phi^1=(\phi^1,\psi^1_\A,G^1),\quad\Phi^2=(\phi^2,\psi^2_\A,G^2),\quad\Phi^3=(\phi^3,\psi^3_\A,G^3),
\fe
where we distinguish the auxiliary fields from those of $D2_1-D2_1$ strings by denoting them as $G$, not $D$. In $\cN=4$ notation, we can combine $I=\phi^0+i\phi^1$, $J=\phi^2+i\phi^3$ and further form a fictitious $SU(2)$ doublet $\cI=(I,J)$. It also couples to $D_{ab}$:
\ie
\text{Tr }\cI^aD^b_a\cI_b.
\fe

We can do the similar analysis for D2${}_1-$D2${}_2$ strings as they also form 3d $\cN=2$ chiral multiplets. Let us combine two real scalar superfields into a complex chiral multiplet $S=S_1+iS_2$ and treat it as a component of the fictitious $SU(2)$ doublet $\cS=(S,T)$. Since there is no $T$ to pair with, one should treat this expression formally. Then, one can write down, at least formally, the relevant Lagrangian: 
\ie\label{firstd}
(\cS^a)^I_i(D^b_a)^i_j(\cS_b)^j_J,
\fe
where $i,j$ are $U(M_1)$ gauge index for $M_1$ D2 branes on 456 direction and $I,J$ are $U(M_2)$ gauge index for $M_2$ D2 branes on 4'5'6' direction. 

Finally, as our gauge group contains  $U(1)$ as a subgroup, there is an FI term:
\ie
\text{Tr }\xi^b_aD^a_b.
\fe


Collecting the D-related Lagrangian terms, we have the following matrix model action
\ie
\text{Tr}\left[|D_a|^2+\cQ^bD^a_b\cQ_a+\cI^bD^a_b\cI_a+\xi^b_aD^a_b+\text{Tr }\cS^aD^b_a\cS_b\right],
\fe
which gives the equation of motion of $D_7=D$
\ie\label{Dtermsp1}
D=\left[X,X^\dagger\right]+\left[Y,Y^\dagger\right]+II^\dagger-JJ^\dagger-S^\dagger S-\xi\cdot\text{I}_{M_1\times M_1}=0.
\fe
The equation of motion of $D_8+iD_9=D_\bC$ is
\ie\label{CDtermsp1}
D_\bC=\left[X,Y\right]+IJ=0.
\fe
We notice that there is no $S$ in this equation, due to the absence of the $\cN=4$ completion of $S$. In other words, $T$ is zero in the formal expression \eqref{firstd}.

We can do the same for the second quiver. Collecting all the relevant D-terms and solving the EOM for $\tilde D_7=\tilde D$, we get
\ie\label{Dtermsp2}
\tilde D=\left[\tilde X,\tilde X^\dagger\right]+\left[\tilde Y,\tilde Y^\dagger\right]+\tilde I\tilde I^\dagger-\tilde J\tilde J^\dagger+SS^\dagger-\xi\cdot\text{I}_{M_2\times M_2}=0.
\fe
The equation of motion of $\tilde D_8+i\tilde D_9=\tilde D_\bC$ is
\ie\label{CDtermsp2}
\tilde D_\bC=\left[\tilde X,\tilde Y\right]+\tilde I\tilde J=0.
\fe
Hence, the moduli space $\cM^{N,K}_{M_1,M_2}$ is the $U(M_1)\otimes U(M_2)$ quotient of the space of solutions of \eqref{Dtermsp1}, \eqref{CDtermsp1}, \eqref{Dtermsp2}, \eqref{CDtermsp2}, where the groups $U(M_1)\otimes U(M_2)$ act as
\ie
\left(X,Y,I,J,S\right)\mapsto\left(g^{-1}Xg,g^{-1}Yg,g^{-1}I,Jg,g^{-1}S\right),\quad\text{where }g\in U(M_1),\\
\left(\tilde X,\tilde Y,\tilde I,\tilde J,S\right)\mapsto\left(\tilde g^{-1}\tilde X\tilde g,\tilde g^{-1}\tilde Y\tilde g,\tilde g^{-1}\tilde I,\tilde J\tilde g,\tilde g^{-1}S\right),\quad\text{where }\tilde g\in U(M_2)
\fe
\newline
\noindent{\textbf{The FI parameter $\xi$ and the non-commutativity parameter $\E_2$}}\\
\newline
There is a relation between the FI parameter $\xi$ and the non-commutativity parameter $\E_2$. We have observed in the $D2_1-D6_1$ string quantization a tachyon with its mass squared given by
\ie\label{tachyonp}
m^2=-(v_1-v_2).
\fe
We can resolve this potential instability and supersymmetry breaking issue by looking at the induced mass term in the 3d Lagrangian \cite{Nekrasov:2016gud}:
\ie\label{massLag}
-\frac{g^2}{2}\text{Tr}\left(-\xi II^\dagger+\xi J^\dagger J\right).
\fe
Mass squared of $I$ is $\xi\frac{g^2}{2}$ and that of $J$ is $-\xi\frac{g^2}{2}$. Hence, identifying this information with \eqref{tachyonp}, we get:
\ie
v_1-v_2=-\xi,
\fe
Because $v_1=-v_2=\E_2$ in our background, the FI parameter is related to the non-commutativity parameter $\E_2$:
\ie
\xi\sim\E_2.
\fe
This condition makes sure the low energy action to be supersymmetric.
For the detailed discussion on the tachyon condensation and restoration of the supersymmetry, see \cite{Nekrasov:2016gud,Pomoni:2021hkn}.

\section{Formal description of the moduli spaces}\label{app:proofs}

\subsection{Moduli space as a subspace of a Nakajima quiver variety}\label{sec:5.1}

If we add an arrow $T$ from $\bC^{M_2}$ to $\bC^{M_1}$, see Figure \ref{double Nakajima quivers}, then it is the double \footnote{For a quiver $Q$, we define its double $\overline{Q}$ to be the quiver with the same nodes as $Q$ and for every arrow of $Q$ we add an arrow of reverse direction.} of the quiver with same nodes but arrows are $X,I,\tilde X,\tilde I,S$. 
\begin{figure}[H]
\centering
\includegraphics[width=10cm]{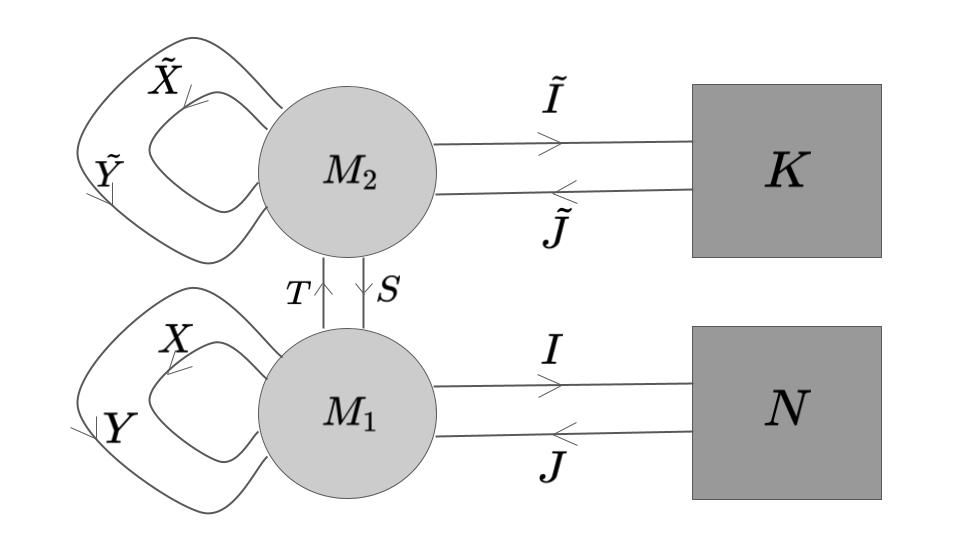}
  \vspace{-10pt}
\caption{}
\vspace{-5pt}
\label{double Nakajima quivers}
 \centering
\end{figure}

Let us denote the corresponding Nakajima quiver variety of Figure \ref{double Nakajima quivers} by $\widetilde{\cM}^{N,K}(M_1,M_2)$. Apparently, $\widetilde{\cM}^{N,K}(M_1,M_2)$ is non-empty and is smooth of dimension $2M_1M_2+2KM_2+2NM_1$. Recall that $\widetilde{\cM}^{N,K}(M_1,M_2)$ is defined as the space of solutions to the equations
\ie\label{eqn for auxiliary NQV}
\left[X,Y\right]&+IJ+TS=0,\\
\left[\tilde X,\tilde Y\right]&+\tilde I\tilde J-ST=0,\\
\left[X,X^\dagger\right]+\left[Y,Y^\dagger\right]&+II^\dagger-JJ^\dagger-S^\dagger S+TT^{\dagger}=\xi\cdot\text{I}_{M_1\times M_1},\\
\left[\tilde X,\tilde X^\dagger\right]+\left[\tilde Y,\tilde Y^\dagger\right]&+\tilde I\tilde I^\dagger-\tilde J\tilde J^\dagger+SS^\dagger-T^{\dagger}T=\xi\cdot\text{I}_{M_2\times M_2},
\fe
modulo the action of $U(M_1)\times U(M_2)$. The first two equations of \eqref{eqn for auxiliary NQV} says that the complex moment map $\mu_{\bC}(X,Y,I,J,\tilde X,\tilde Y,\tilde I,\tilde J,S,T)=0$ and the last two equations of \eqref{eqn for auxiliary NQV} says that the real moment map $\mu_{\bR}(X,Y,I,J,\tilde X,\tilde Y,\tilde I,\tilde J,S,T)=(\xi\cdot\text{I}_{M_1\times M_1},\xi\cdot\text{I}_{M_2\times M_2})$.

On the moduli space $\widetilde{\cM}^{N,K}(M_1,M_2)$, there is a universal rank $M_1$ bundle denoted by $\mathcal V_1$, and a universal rank $M_2$ bundle denoted by $\mathcal V_2$, and a universal map $\mathcal T:\mathcal V_2\to \mathcal V_1$.

Assume that $\xi\neq 0$, and denote by $\widetilde{M}^{N,K}(M_1,M_2)$ the space of solutions to \eqref{eqn for auxiliary NQV}, then it is well-known that $U(M_1)\times U(M_2)$ acts on $\widetilde{M}^{N,K}(M_1,M_2)$ freely, and moreover the pullback of the vanishing locus of $\mathcal T$ is exactly the locus in $\widetilde{M}^{N,K}(M_1,M_2)$ where $T$ in the equations \eqref{eqn for auxiliary NQV} is zero. So we have
\begin{lemma}
Assume that $\xi\neq 0$, then $\cM^{N,K}(M_1,M_2)$ embeds into $\widetilde{M}^{N,K}(M_1,M_2)$ as the vanishing locus of $\mathcal T$.
\end{lemma}

Recall that $\widetilde{\cM}^{N,K}(M_1,M_2)$ has algebro-geometric description \cite{Nakajima:1994}. Let's assume that $\xi>0$, then $\widetilde{\cM}^{N,K}(M_1,M_2)$ is the space of \textbf{stable} solutions to the equations
\ie\label{eqn for auxiliary NQV 2}
\left[X,Y\right]+IJ+TS=0,\\
\left[\tilde X,\tilde Y\right]+\tilde I\tilde J-ST=0,
\fe
modulo the action of $\mathrm{GL}_{M_1}\times \mathrm{GL}_{M_2}$. Here ``stable" means that the only subspace of $\bC^{M_1}\oplus \bC^{M_2}$, which contains the image of $I$ and $\tilde I$ and is invariant under the actions of $X,Y,\tilde X,\tilde Y,S,T$ is $\bC^{M_1}\oplus \bC^{M_2}$ itself.
This has the following obvious consequences:
\begin{lemma}\label{Lemma: algebro-goemetric description}
Assume that $\xi>0$, then we have an isomorphism:
\ie 
\cM^{N,K}(M_1,M_2)=\mu_{\bC}^{-1}(0)^{\mathrm{stable}}/\mathrm{GL}_{M_1}\times \mathrm{GL}_{M_2},
\fe
where ``stable" means the locus inside $\mu_{\bC}^{-1}(0)$ such that $\{X,Y,I,J,\tilde X,\tilde Y,\tilde I,\tilde J,S\}$ satisfies that $$\bC\langle X,Y\rangle \mathrm{Im}(I)=\bC^{M_1},\;\bC\langle \tilde X,\tilde Y\rangle (\mathrm{Im}(\tilde I)+\mathrm{Im}(S))=\bC^{M_2}.$$
\end{lemma}

\begin{remark}
If we assume $\xi<0$ instead, then the stability condition in the above lemma becomes
\begin{itemize}
    \item[*] The only subspace of $\bC^{M_1}\oplus \bC^{M_2}$, which is annihilated by $J\oplus \tilde J$ and invariant under the actions of $X,Y,\tilde X,\tilde Y,S,T$ is $0$.
\end{itemize}
\end{remark}

\begin{remark}
Lemma \ref{Lemma: algebro-goemetric description} says that $$(\mu_{\bC}\times \mu_{\bR})^{-1}(0,\xi\cdot\text{I}_{M_1\times M_1},\xi\cdot\text{I}_{M_2\times M_2})/U(M_1)\times U(M_2)=\mu_{\bC}^{-1}(0)^{\mathrm{stable}}/\mathrm{GL}_{M_1}\times \mathrm{GL}_{M_2},$$and the following technical result \cite[Proposition 3.5]{Nakajima:1994} provide a geometric interpretation of this equality:
\begin{itemize}
    \item Let $x$ be a point in $\mu_{\bC}^{-1}(0)$, then $x$ is stable if and only if the $\mathrm{GL}_{M_1}\times \mathrm{GL}_{M_2}$ orbit through $x$ intersects with $\mu_{\bR}^{-1}(\xi\cdot\text{I}_{M_1\times M_1},\xi\cdot\text{I}_{M_2\times M_2})$.
\end{itemize}
Namely, if we consider the union of $\mathrm{GL}_{M_1}\times \mathrm{GL}_{M_2}$ orbits of $(\mu_{\bC}\times \mu_{\bR})^{-1}(0,\xi\cdot\text{I}_{M_1\times M_1},\xi\cdot\text{I}_{M_2\times M_2})$, then this set is exactly $\mu_{\bC}^{-1}(0)^{\mathrm{stable}}$, and it is known that $\mathrm{GL}_{M_1}\times \mathrm{GL}_{M_2}$ acts on $\mu_{\bC}^{-1}(0)^{\mathrm{stable}}$ freely, so we have
\ie 
&\mu_{\bC}^{-1}(0)^{\mathrm{stable}}/\mathrm{GL}_{M_1}\times \mathrm{GL}_{M_2}\\
=&(\mathrm{GL}_{M_1}\times \mathrm{GL}_{M_2})\cdot (\mu_{\bC}\times \mu_{\bR})^{-1}(0,\xi\cdot\text{I}_{M_1\times M_1},\xi\cdot\text{I}_{M_2\times M_2})/\mathrm{GL}_{M_1}\times \mathrm{GL}_{M_2}\\
=&(\mu_{\bC}\times \mu_{\bR})^{-1}(0,\xi\cdot\text{I}_{M_1\times M_1},\xi\cdot\text{I}_{M_2\times M_2})/U(M_1)\times U(M_2).
\fe 
\end{remark}

Notice that when $\xi>0$, the stability condition implies that the subset of quiver data $(X,Y,I,J)$ is stable as an individual ADHM quiver. This shows that:
\begin{lemma}\label{Fibration over Hilb}
Assume that $\xi> 0$, then there is a map $\pi: \cM^{N,K}(M_1,M_2)\to \cM^N(M_1)$, where $\cM^N(M_1)$ is the moduli space of torsion-free sheaves on $\mathbb{P}^2$ framed at $\mathbb{P}^1_{\infty}$ with rank $N$ and second Chern number $M_1$. Moreover, $\pi$ is a locally trivial fibration with fibers isomorphic to the quiver variety associated to the quiver diagram.
\begin{figure}[H]
\centering
\includegraphics[width=10cm]{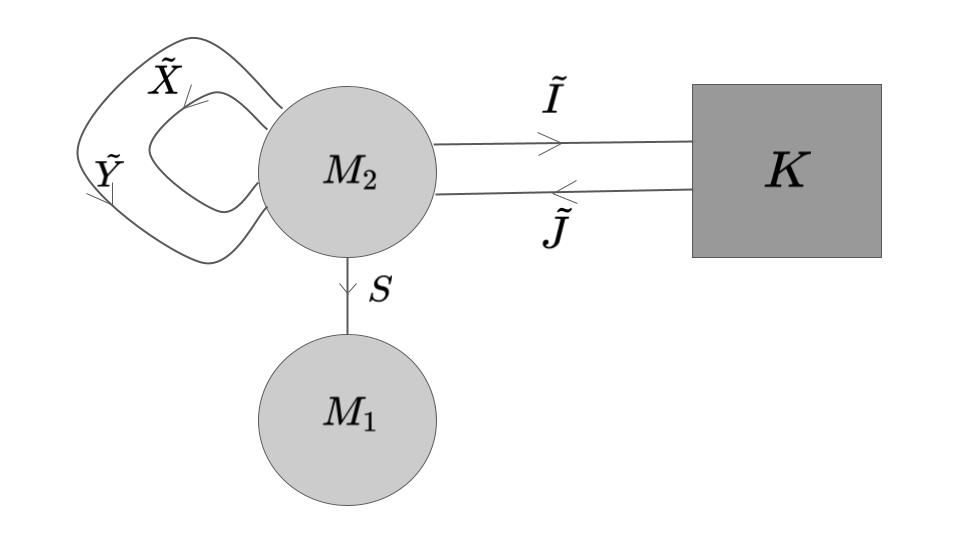}
  \vspace{-10pt}
\caption{}
\label{quiver for fibers}
 \centering
\end{figure}
\noindent The stability condition is that $\bC\langle \tilde X,\tilde Y\rangle (\mathrm{Im}(\tilde I)+\mathrm{Im}(S))=\bC^{M_2}$.
\end{lemma}

\begin{proof}
$\pi$ is defined by $(X,Y,I,J,\tilde X,\tilde Y,\tilde I,\tilde J,S)\mapsto (X,Y,I,J)$. On the other hand, if we use the universal bundle $\mathcal V_1$ on $\cM^N(M_1)$ to form a fibration of quiver variety associated to \figref{quiver for fibers} with stability condition that $\bC\langle \tilde X,\tilde Y\rangle (\mathrm{Im}(\tilde I)+\mathrm{Im}(S))=\bC^{M_2}$, where $\bC^{M_1}$ is replaced by the bundle $\mathcal V_1$, then the total space is exactly the moduli space of stable solutions to equations \eqref{ADHM eqn}.
\end{proof}

\begin{proposition}\label{Fibration over Hilb, K=0}
Assume that $\xi> 0$ and $K=0$, then $\pi: \cM^{N,0}(M_1,M_2)\to \cM^N(M_1)$ is a locally trivial fibration with fibers isomorphic to $\mathrm{Quot}^{M_2}(\mathcal O_{\bC^2}^{M_1})$ \footnote{ $\mathrm{Quot}^{M_2}(\mathcal O_{\bC^2}^{M_1})$ is the moduli space of quotient sheaf of $\mathcal O_{\bC^2}^{M_1}$ such that the $\bC$-dimension of the sheaf is $M_2$.}.
\end{proposition}

\begin{proof}
In the quiver \ref{quiver for fibers}, if we set $K=0$, then the F-term equation is $[\tilde X,\tilde Y]=0$, and the stability condition reads that $\bC [\tilde X,\tilde Y]\mathrm{Im}(S)=\bC^{M_2}$, this is exactly the definition of moduli space of quotient sheaf of $\mathcal O_{\bC^2}^{M_1}$ such that the $\bC$-dimension of the sheaf is $M_2$. Note that $\bC^{M_2}$ is the linear space of quotient sheaf and $\tilde X,\tilde Y$ is the coordinate functions of $\bC^2$.
\end{proof}

In general $\cM^{N,K}(M_1,M_2)$ is highly singular, but the following proposition shows that the singularity of $\cM^{N,K}(M_1,M_2)$ is not too bad under some mild assumptions.

\begin{proposition}\label{prop: regular embedding}
Assume that $\xi\neq 0$ and $NK\neq 0$, then the embedding $\cM^{N,K}(M_1,M_2)\hookrightarrow\widetilde{M}^{N,K}(M_1,M_2)$ is regular of codimension $M_1M_2$. Moreover, locally the set of matrix elements of $\mathcal T$ is a regular sequence of length $M_1M_2$ \footnote{A regular sequence of length $n$ in a commutative ring $\cA$ is a set of elements $a_1,\cdots,a_n\in \cA$ such that $a_{i+1}$ does not have zero-divisor in the quotient ring $\cA/(a_1,\cdots,a_i)\cA$, for all $i$ ranging from $0$ to $n-1$ (set $a_0=0$). The key point is that, if $\cA$ is smooth, then a sequence $a_1,\cdots,a_n\in \cA$ is regular if and only if $\dim \mathrm{Spec}\cA-\dim \mathrm{Spec}\cA/(a_1,\cdots,a_i)\cA=n$.}.
\end{proposition}

\begin{proof}
It suffices to prove that $\cM^{N,K}(M_1,M_2)$ is local complete intersection \footnote{Local complete intersection means locally embeds into a smooth ambient space and is cut out by a regular sequence of equations.} (l.c.i) of pure dimension \footnote{Pure dimension means all irreducible components have the same dimension.} $M_1M_2+2M_1N+2M_2K$, then it automatically follows that locally the set of matrix elements of $\mathcal T$ is a regular sequence of length $M_1M_2$. To this end, we need to show that $\mu_{\bC}^{-1}(0)^{\mathrm{stable}}$ is l.c.i of pure dimension $M_1M_2+2M_1N+2M_2K+M_1^2+M_2^2$, since the quotient map $\mu_{\bC}^{-1}(0)^{\mathrm{stable}}\to \mu_{\bC}^{-1}(0)^{\mathrm{stable}}/\mathrm{GL}_{M_2}\times \mathrm{GL}_{M_1}=\cM^{N,K}(M_1,M_2)$ is a principal $\mathrm{GL}_{M_2}\times \mathrm{GL}_{M_1}$ bundle and l.c.i property descends to a smooth morphism. 

We claim that the map 
$$\mu_{\bC}: \mathbf{M_1}\times \mathbf{M_1}^*\times \mathrm{End}(\mathbf{M_1})^{2}\times \mathbf{M_2}\times \mathbf{M_2}^*\times \mathrm{End}(\mathbf{M_2})^{2}\times \mathrm{Hom}(\mathbf{M_1},\mathbf{M_2})\to \mathfrak{gl}_{M_1}\times \mathfrak{gl}_{M_2}$$
is flat \footnote{Flatness means pullback of short exact sequence of sheaves is a short exact sequence.}. Note that the component $\mathrm{Hom}(\mathbf{M_1},\mathbf{M_2})$ where $S$ takes value in, plays no role in the map $\mu_{\bC}$, so $\mu_{\bC}$ factors through the product of moment maps $$\mu_{\bC}^N\times \mu_{\bC}^K:\left( \mathbf{M_1}\times \mathbf{M_1}^*\times \mathrm{End}(\mathbf{M_1})^{2}\right)\times \left( \mathbf{M_2}\times \mathbf{M_2}^*\times \mathrm{End}(\mathbf{M_2})^{2}\right)\to \mathfrak{gl}_{M_1}\times \mathfrak{gl}_{M_2}.$$
And it is known that $\mu_{\bC}^N:\mathbf{M_1}\times \mathbf{M_1}^*\times \mathrm{End}(\mathbf{M_1})^{2}\to \mathfrak{gl}_{M_1}$ is flat, using Crawley-Boevey's criterion on the flatness of moment map \cite[Theorem 1.1]{Crawley-Boevey:2001} \footnote{We use the equivalence between $(1)$ and $(4)$ in \cite[Theorem 1.1]{Crawley-Boevey:2001}, and easy computation shows that $(4)$ holds.}. Hence both $\mu_{\bC}^N$ and $\mu_{\bC}^K$ are flat and we deduce that $\mu_{\bC}$ is flat.

Since $\{0\}\hookrightarrow \mathfrak{gl}_{M_1}\times \mathfrak{gl}_{M_2}$ is a regular embedding, the flatness of $\mu_{\bC}$ implies that $\mu_{\bC}^{-1}(0)$ embeds into the ambient space regularly, and it has dimension
\ie
\dim \mathbf{M_1}\times \mathbf{M_1}^*\times \mathrm{End}(\mathbf{M_1})^{2}\times \mathbf{M_2}\times \mathbf{M_2}^*\times \mathrm{End}(\mathbf{M_2})^{2}\times \mathrm{Hom}(\mathbf{M_1},\mathbf{M_2})-\dim \mathfrak{gl}_{M_1}\times \mathfrak{gl}_{M_2}\\
=M_1M_2+2M_1N+2M_2K+M_1^2+M_2^2.
\fe
This shows that $\mu_{\bC}^{-1}(0)$ is l.c.i of of pure dimension $M_1M_2+2M_1N+2M_2K+M_1^2+M_2^2$, and obviously its open subset $\mu_{\bC}^{-1}(0)^{\mathrm{stable}}$ has the same property.
\end{proof}

In the proof of the proposition, we show that 
\begin{itemize}
    \item Assume that $\xi\neq 0$ and $NK\neq 0$, then every irreducible component of $\cM^{N,K}(M_1,M_2)$ has dimension $M_1M_2+2M_1N+2M_2K$. In particular, the virtual dimension is the actual dimension.
\end{itemize}

\subsection{Virtual tangent bundle}\label{sec:5.2}

Assume that $\xi\neq 0$ in this subsection. It is well known that the tangent bundle of $\widetilde{\cM}^{N,K}(M_1,M_2)$ is the cohomology of the complex
\ie\label{tangent of auxiliary moduli}
\mathrm{End}(\mathcal V_{1})\oplus \mathrm{End}(\mathcal V_{2})\overset{\sigma}{\longrightarrow} \mathrm{End}(\mathcal V_{1})^{\oplus 2}\oplus \mathrm{End}(\mathcal V_{2})^{\oplus 2} \oplus \mathcal V_{1}^{\oplus N} \oplus \mathcal V_{1}^{*\oplus N}\oplus \mathcal V_{2}^{\oplus K} \oplus \mathcal V_{2}^{*\oplus N}\\
\oplus \mathrm{Hom}(\mathcal V_{1},\mathcal V_{2})\oplus \mathrm{Hom}(\mathcal V_{2},\mathcal V_{1})
\overset{d\mu_{\bC}}{\longrightarrow} \mathrm{End}(\mathcal V_{1})\oplus \mathrm{End}(\mathcal V_{2})
\fe
where the middle term has cohomology degree zero. $d\mu_{\bC}$ is the differential of the moment map, and explicitly it is written as
\ie 
(x,y,\tilde x,\tilde y,i,j,\tilde i,\tilde j,s,t)\mapsto ([X,y]+[x,Y]+iJ+Ij+sT,[\tilde X,\tilde y]+[\tilde x,\tilde Y]+\tilde i\tilde J+\tilde I\tilde j-St).
\fe
$\sigma$ descends from the action of $\mathfrak{gl}_{M_1}\times \mathfrak{gl}_{M_2}$, and explicitly it is written as 
\ie 
(u,v)\mapsto ([u,X],[Y,u],[v,\tilde X],[\tilde Y,v],uI,-Ju,v\tilde I,-\tilde J v,vS-Su,uT-Tv).
\fe
It is easy to verify that $d\mu_{\bC}\circ \sigma=0$. Note that $\sigma$ is injective and $d\mu_{\bC}$ is surjective by stability. Our moduli space $\cM^{N,K}(M_1,M_2)$ embeds into $\widetilde{\cM}^{N,K}(M_1,M_2)$ regularly, and is locally defined by $M_1M_2$ equations coming from matrix elements of $\mathcal T\in \mathrm{Hom}(\mathcal V_{2},\mathcal V_{1})$, so $\cM^{N,K}(M_1,M_2)$ has a \textit{perfect obstruction theory} $T\widetilde{\cM}^{N,K}(M_1,M_2)|_{\cM^{N,K}(M_1,M_2)}\to \mathrm{Hom}(\mathcal V_{2},\mathcal V_{1})$, and this complex is quasi-isomorphic to the tangent complex $\mathbb{T}\cM^{N,K}(M_1,M_2)$. This obstruction theory can be written as a complex of tautological bundles:
\ie\label{tangent of our moduli}
\mathrm{End}(\mathcal V_{1})\oplus \mathrm{End}(\mathcal V_{2})\overset{\sigma}{\longrightarrow} \mathrm{End}(\mathcal V_{1})^{\oplus 2}\oplus \mathrm{End}(\mathcal V_{2})^{\oplus 2} \oplus \mathcal V_{1}^{\oplus N} \oplus \mathcal V_{1}^{*\oplus N}\oplus \mathcal V_{2}^{\oplus K} \oplus \mathcal V_{2}^{*\oplus N}\\
\oplus \mathrm{Hom}(\mathcal V_{1},\mathcal V_{2})
\overset{d\mu_{\bC}}{\longrightarrow} \mathrm{End}(\mathcal V_{1})\oplus \mathrm{End}(\mathcal V_{2})
\fe
The difference between \eqref{tangent of our moduli} and \eqref{tangent of auxiliary moduli} is that $\mathrm{Hom}(\mathcal V_{2},\mathcal V_{1})$ drops out and chain maps are restricted to the rest of components. In \eqref{tangent of our moduli}, $\sigma$ is  injective by stability condition, but $d\mu_{\bC}$ is not necessarily surjective. We shall show that $d\mu_{\bC}$ fails to be surjective at some points on $\cM^{N,K}(M_1,M_2)$.

\subsection{Some subvarieties of moduli space}\label{sec:5.3}

We assume that $\xi>0$ in this subsection. Let us investigate some subvarieties of $\cM^{N,K}(M_1,M_2)$. We start with an open subvariety denoted by $\mathring{\cM}^{N,K}(M_1,M_2)$, and it is defined by the open condition \footnote{This means that the set of point which satisfy this condition is open.} that
$$\bC\langle \tilde X,\tilde Y\rangle \mathrm{Im}(\tilde I)=\bC^{M_2}.$$
Note that under this condition, two gauge nodes satisfy their own stability condition when considered as individual ADHM quivers. So we have a projection
\ie 
p: \mathring{\cM}^{N,K}(M_1,M_2)\to \cM^N(M_1)\times \cM(M_2,K)
\fe
and in fact the projection map is a vector bundle where fibers are maps between $\mathcal V_{1}$ and $\mathcal V_{2}$, where $\mathcal V_{1}$ and $\mathcal V_{2}$ are universal bundles on $\cM^N(M_1)$ and $\cM(M_2,K)$ respectively. 
In other words, we have an isomorphism
\ie\label{Open locus is a Hom bundle}
\mathring{\cM}^{N,K}(M_1,M_2)=\mathrm{Hom}(\mathcal V_{1},\mathcal V_{2})
\fe
Let us also consider the closed subvariety denoted by ${\cM}^{N,K}_{\tilde J=0}(M_1,M_2)$, and it is defined by the closed condition \footnote{This means that the set of point which satisfy this condition is closed.} that
$$\tilde J=0.$$
The same argument of Lemma \ref{Fibration over Hilb} shows that ${\cM}^{N,K}_{\tilde J=0}(M_1,M_2)$ is a locally trivial fibration over $\cM^N(M_1)$ with fibers isomorphic to $\mathrm{Quot}^{M_2}(\mathcal O_{\bC^2}^{K+M_1})$.
\begin{proposition}
Assume that $\xi>0$ and $K=1$, then ${\cM}^{N,1}_{\tilde J=0}(M_1,M_2)$ is the closure of $\mathring{\cM}^{N,1}(M_1,M_2)$.
\end{proposition}
\begin{proof}
If $K=1$, then the stability condition $\bC\langle \tilde X,\tilde Y\rangle \mathrm{Im}(\tilde I)=\bC^{M_2}$ together with the equation $[\tilde X,\tilde Y]+\tilde I\tilde J=0$ implies that $\tilde J=0$, so $\mathring{\cM}^{N,1}(M_1,M_2)$ is a subvariety of ${\cM}^{N,1}_{\tilde J=0}(M_1,M_2)$. Moreover, we prove in the appendix that ${\cM}^{N,1}_{\tilde J=0}(M_1,M_2)$ is irreducible (see Proposition \ref{Irreducible J=0 variety}), thus ${\cM}^{N,1}_{\tilde J=0}(M_1,M_2)$ is the closure of $\mathring{\cM}^{N,1}$.
\end{proof}

\subsection{Torus fixed points}\label{sec:5.4}

We first introduce some notations. Define the action of tori $\mathbf{T}_{i},i\in \{1,2,3\}$ as follows:
\ie 
a\in \mathbf{T}_{1}&:(X,Y)\mapsto (aX,a^{-1}Y),\\
b\in\mathbf{T}_{2}&:(\tilde X,\tilde Y)\mapsto (b\tilde X,b^{-1}\tilde Y),\\
c\in\mathbf{T}_{3}&:S\mapsto cS,
\fe
and denote by $\mathbf{T}_I,I\subset\{1,2,3\}$ the torus $\prod_{i\in I}\mathbf{T}_i$. 

The next proposition follows from the same argument as the proof of \cite[Proposition 2.3.1]{Maulik:2012}
\begin{proposition}
Assume that $\xi>0$, then the $\mathbf{T}_3$ fixed points have a disjoint union decomposition:
\ie 
\cM^{N,K}(M_1,M_2)^{\mathbf{T}_3}=\bigsqcup_{M_2^{(1)}+M_2^{(2)}=M_2}\cM^{N,0}(M_1,M_2^{(1)})\times \cM^{0,K}(0,M_2^{(2)}).
\fe
\end{proposition}

\begin{remark}
As we have mentioned in the Proposition \ref{Fibration over Hilb, K=0}, $\cM^{N,0}(M_1,M_2^{(1)})$ is a locally trivial fibration over the ADHM moduli space $\cM^N(M_1)$ with fibers isomorphic to $\mathrm{Quot}^{M_2^{(1)}}(\mathcal O_{\bC^2}^{M_1})$. And by definition, $\cM^{0,K}_i(0,M_2^{(2)})$ is nothing but the ADHM moduli space $\cM(M_2^{(2)},K)$.
\end{remark}

\begin{corollary}\label{compactness 1}
Assume that $\xi\neq 0$ and $N=K=1$, then $\cM^{1,1}(M_1,M_2)^{\mathbf{T}_{\{1,2,3\}}}$ is proper (compact).
\end{corollary}

\begin{proof}
We only prove for the case $\xi>0$, the case $\xi<0$ is similar. It is enough to show that:
\begin{itemize}
    \item[(1)] $\cM(M_1,1)^{\mathbf{T}_1}$ is proper;
    \item[(2)] $\mathrm{Quot}^{M_2^{(1)}}(\mathcal O_{\bC^2}^{M_1})^{\mathbf{T}_2}$ is proper.
\end{itemize}
Note that if we can show $(1)$ then the same argument shows that $\cM(M_2^{(2)},1)^{\mathbf{T}_2}$ is proper as well.

Notice that $(1)$ is a special case of $(2)$, since $\cM(M_1,1)=\mathrm{Quot}^{M_1}(\mathcal O_{\bC^2})$ and both of the torus act on $\bC^2$ by $(x,y)\mapsto (rx,r^{-1}y)$, where $(x,y)$ is a coordinate system on $\bC^2$. Let us omit the subscripts and show that $\mathrm{Quot}^{M}(\mathcal O_{\bC^2}^{L})^{\mathbf{T}}$ is proper. Recall that the Hilbert-Chow map \footnote{For the definition of Hilbert-Chow map, see for example \cite[Section 7.1]{Fantechi:2005}.} $h:\mathrm{Quot}^{M}(\mathcal O_{\bC^2}^{L})^{\mathbf{T}}\to \mathrm{Sym}^{M}\bC^2$ is proper and $\mathbf{T}$-equivariant, so $\mathrm{Quot}^{M}(\mathcal O_{\bC^2}^{L})^{\mathbf{T}}\subset h^{-1}((\mathrm{Sym}^{M}\bC^2)^{\mathbf{T}})$. Since $(\mathrm{Sym}^{M}\bC^2)^{\mathbf{T}}$ consists of a single point $0$, $\mathrm{Quot}^{M}(\mathcal O_{\bC^2}^{L})^{\mathbf{T}}$ is a closed subset of $h^{-1}(0)$, which is proper. This concludes the proof.
\end{proof}

\begin{remark}
For general $N,K$, the fixed point set $\cM^{N,K}(M_1,M_2)^{\mathbf{T}_{\{1,2,3\}}}$ can be non-compact. Note that $\cM^N(M_1)^{\mathbf{T}_1}$ is the moduli space of $\mathbf{T}_1$-weighted representations of ADHM quiver of gauge rank $M_1$ and flavour rank $N$, so the gauge node $\bC^{M_1}$ decomposes into $\mathbf{T}_1$-weight space, and flavour node $\bC^N$ has $\mathbf{T}_1$-weight zero, and $X,Y,I,J$ are $\mathbf{T}_1$-equivariant. A connected component of $\cM^N(M_1)^{\mathbf{T}_1}$ corresponds to a weight decomposition
$$\bC^{M_1}=\bigoplus_{i=0}^{M_1-1}\bC\cdot e_i,$$
where $e_i$ has weight $i$, $X$ maps $e_i$ to $\lambda_{i}e_{i-1}$ and $Y$ maps $e_i$ to $\mu_i e_{i+1}$ (we set $\lambda_i=0$ if $i\notin \{1,\cdots,M_1-1\}$ and set $\mu_i=0$ if $i\notin \{0,\cdots,M_1-2\}$), and $I,J$ map between $\bC\cdot e_0$ and the framing node $\bC^N$. It can be easily deduced from the stability condition that 
$$\mu_i\neq 0,\;\forall i\in\{0,\cdots, M_1-2\}.$$
It follows from the equation $[X,Y]+IJ=0$ that
$$\lambda_i\mu_{i-1}=\lambda_{i+1}\mu_i,\;\forall i\ge 1.$$
Combine the above two facts with the equation $\lambda_{M_1}=0$, and we deduce that $\lambda_i=0,\forall i$. In the moduli space of quiver representations, we can scale each node, so we can assume that $\mu_i=1$ for $i=0,\cdots,M_1-2$. The remaining moduli parameters are $I:\bC^N\to \bC\cdot e_0$ and $J:\bC\cdot e_0\to \bC^N$ satisfying $IJ=0$ and the stability condition that $I\neq 0$, and it is well-known that this moduli space is $T^*\mathbb{P}^{N-1}$, the cotangent bundle of $\mathbb{P}^{N-1}$. Hence we see that $\cM^N(M_1)^{\mathbf{T}_1}$ contains a connected component isomorphic to $T^*\mathbb{P}^{N-1}$, which is non-compact. 
\end{remark}

A peculiar fact about $K=1$ is that $\cM^{0,1}(0,M_2^{(2)})=\cM^{0,1}_{\tilde J=1}(0,M_2^{(2)})$ (this is a consequence of stability condition together with the F-term equation), so we can actually turn on a larger torus action, namely the two dimensional torus acting on both $\tilde X$ and $\tilde Y$. This torus will be denoted by $\mathbf{T}'_2$. Similarly, if $N=1$ then we have a two dimensional torus acting on $X$ and $Y$.

More generally, we have a five dimensional torus $\mathbf{T}'_1\times\mathbf{T}'_2\times \mathbf{T}_3$ acting on $\cM^{N,K}(M_1,M_2)$, which extends the action of $\mathbf{T}_{\{1,2,3\}}$. The action is defined by 
\ie 
(a_1,a_2)\in \mathbf{T}'_{1}&:(X,Y,I,J)\mapsto (a_1X,a_2Y,I,a_1a_2J),\\
(b_1,b_2)\in\mathbf{T}'_{2}&:(\tilde X,\tilde Y,\tilde I,\tilde J)\mapsto (b_1\tilde X,b_2\tilde Y,\tilde I,b_1b_2\tilde J).
\fe
Denote this torus by $\mathbf{T}_{\mathrm{edge}}$, then we have the following:
\begin{proposition}\label{counting fixed pts}
Under the above assumptions, $\cM^{1,1}(M_1,M_2)^{\mathbf{T}_{\mathrm{edge}}}$ is disjoint union of points \footnote{In the scheme-theoretical sense, i.e. there is no infinitesimal deformation of those points inside the $\mathbf{T}_{\mathrm{edge}}$-fixed points loci.}, points one-to-one correspond to $M_1+2$ copies of Young diagrams $Y^{(1)},Y^{(2)}_0,Y^{(2)}_1,\cdots, Y^{(2)}_{M_1}$, where $|Y^{(1)}|=M_1$ and $\sum_{i=0}^{M_1}|Y^{(2)}_i|=M_2$. Moreover, the generating function for the number of $\mathbf{T}_{\mathrm{edge}}$-fixed points can be written as
\ie\label{fixed pts generating function}
\sum_{M_1,M_2}\#\left(\cM^{1,1}(M_1,M_2)^{\mathbf{T}_{\mathrm{edge}}}\right)x^{M_1}y^{M_2}=\frac{1}{\left(\frac{x}{(y;y)_{\infty}};\frac{x}{(y;y)_{\infty}}\right)_{\infty}(y;y)_{\infty}},
\fe
where $\left(a;y\right)_\infty=\prod^\infty_{k=0}(1-ay^k)$ is the Pochhammer symbol.
\end{proposition}

\begin{proof}
Recall that the edge torus fixed points of the ADHM quiver variety $\cM^{1,0}(M,0)$ is disjoint union of points which are one-to-one correspond to Young diagrams $Y$ with $|Y|=M$. More precisely, a fixed point $p\in \cM^{1,0}(M,0)^{\mathbf{T}'_1}$ corresponds to a unique $\mathbf{T}'_1$-equivariant quiver representation of which $\mathbb C^M$ decomposes into $\mathbf{T}'_1$ weight spaces
$$\mathbb C^M=\bigoplus_{(i,j)\in Y} \mathbb C\cdot e_{i,j},$$
where $(i,j)\in Y$ is the $(i,j)$ box in the Young diagram $Y$ and $e_{i,j}$ is a vector of weight $(-i,-j)$ under the $\mathbf{T}'_1$ action. We set $e_{i,j}=0$ if $(i,j)\notin Y$, then $X$ maps $e_{i,j}$ to $e_{i+1,j}$ and $Y$ maps $e_{i,j}$ to $e_{i,j+1}$ and $I$ maps the flavour vector to $e_{0,0}$ and all other arrows are zero.

Back to our situation, we have 
$$\cM^{1,1}(M_1,M_2)^{\mathbf{T}_{\mathrm{edge}}}=\bigsqcup_{M_2^{(1)}+M_2^{(2)}=M_2}\cM^{1,0}(M_1,M_2^{(1)})^{\mathbf{T}'_1\times\mathbf{T}'_2}\times \cM^{0,1}(0,M_2^{(2)})^{\mathbf{T}'_2}.$$From the discussion above, we see that $\cM^{0,1}(0,M_2^{(2)})^{\mathbf{T}'_2}$ is disjoint union of points labelled by Young diagrams with $M_2^{(2)}$ boxes, and there is a map $\cM^{1,0}(M_1,M_2^{(1)})^{\mathbf{T}'_1\times\mathbf{T}'_2}\to \cM^{1,0}(M_1,0)^{\mathbf{T}'_1}$ and the latter is also disjoint union of points labelled by Young diagrams with $M_1$ boxes. Let us fix Young diagrams $Y^{(1)}$ and $Y_0^{(2)}$ such that $|Y^{(1)}|=M_1$ and $|Y_0^{(2)}|=M_2^{(2)}$. Then $Y^{(1)}$ corresponds to a point $p$ in $\cM^{1,0}(M_1,0)^{\mathbf{T}'_1}$, and a point in the preimage of $p$ in $\cM^{1,0}(M_1,M_2^{(1)})^{\mathbf{T}'_1\times\mathbf{T}'_2}$ corresponds to a $\mathbf{T}'_1\times\mathbf{T}'_2$-equivariant quiver representation of which $\mathbb C^M$ decomposes into $\mathbf{T}'_1\times \mathbf{T}'_2$ weight spaces
$$\mathbb C^M=\bigoplus_{(i,j)\in Y^{(1)}} \mathbb C\cdot e_{i,j},$$
where $(i,j)\in Y$ is the $(i,j)$ box in the Young diagram $Y^{(1)}$ and $e_{i,j}$ is a vector of weight $(-i,-j)$ under the $\mathbf{T}'_1$ action and of weight zero under the $\mathbf{T}'_2$ action. Next we decompose $\mathbb C^{M_2^{(1)}}$ into $\mathbf{T}'_1$ weight spaces:
$$\mathbb C^{M_2^{(1)}}=\bigoplus_{(i,j)\in Y^{(1)}} V_{i,j},$$where $V_{i,j}$ has $\mathbf{T}'_1$ weight $(-i,-j)$. By equivariance, $S$ maps $e_{i,j}$ to $V_{i,j}$ and $\tilde X,\tilde Y$ maps $V_{i,j}$ to itself, and by stability, we have $\mathbb C\langle \tilde X,\tilde Y\rangle S(e_{i,j})=V_{i,j}$, this shows that the restriction of $\tilde X,\tilde Y$ to $V_{i,j}$ is a stable ADHM quiver representation, where the flavour vector is $e_{i,j}$. In other word, we show that the preimage of $p$ in $\cM^{1,0}(M_1,M_2^{(1)})^{\mathbf{T}'_1\times\mathbf{T}'_2}$ is isomorphic to
$$\bigsqcup_{\substack{N_1,\cdots,N_{M_1}\ge 0\\N_1+\cdots +N_{M_1}=M_2^{(1)}}}\cM^{0,1}(0,N_{1})^{\mathbf{T}'_2}\times \cdots\times \cM^{0,1}(0,N_{M_1})^{\mathbf{T}'_2}.$$ Each product factor $\cM^{0,1}(0,N_i)^{\mathbf{T}'_2}$ is disjoint union of points which are one-to-one correspond to Young diagrams $Y^{(2)}_i$ with $N_i$ boxes. Hence we prove that $\cM^{1,1}(M_1,M_2)^{\mathbf{T}_{\mathrm{edge}}}$ is disjoint union of points which are one-to-one correspond to $M_1+2$ copies of Young diagrams $Y^{(1)},Y^{(2)}_0,Y^{(2)}_1,\cdots, Y^{(2)}_{M_1}$, where $|Y^{(1)}|=M_1$ and $\sum_{i=0}^{M_1}|Y^{(2)}_i|=M_2$.

To enumerate fixed points, notice that 
\ie
    \sum_{M_1} \left(\cM^{1,0}(M_1,0)^{\mathbf{T}_{\mathrm{edge}}}\right)x^{M_1}=\sum_{M_1}p(M_1)x^{M_1}=\prod_{i\ge 1}^{\infty}\frac{1}{1-x^i}=\frac{1}{(x;x)},
\fe
where $p(n)$ is Euler's partition function counting the number of partitions of $n$ in to positive integers. Thus we have 
\begin{align*}
        &\sum_{M_1,M_2}\left(\cM^{1,1}(M_1,M_2)^{\mathbf{T}'_{1}}\right)x^{M_1}y^{M_2}\\
        =&\left(\sum_{M_1}\left(\cM^{1,0}(M_1,0)^{\mathbf{T}'_{1}}\right)\sum_{N_1,\cdots,N_{M_1}\ge 0}\prod_{i=1}^{M_1}\left(\cM^{0,1}(0,N_{i})^{\mathbf{T}'_2}\right)x^{M_1}y^{N_1+\cdots+N_{M_1}}\right)\\
        &\times \left(\sum_{M_2^{(2)}} \left(\cM^{0,1}(0,M_2^{(2)})^{\mathbf{T}'_{2}}\right)y^{M_2^{(2)}}\right)\\
        =&\left(\sum_{M_1}p(M_1)\left(\frac{x}{(y;y)}\right)^{M_1}\right)\times \left(\sum_{M_2^{(2)}}p(M_2^{(2)})y^{M_2^{(2)}}\right)\\
        =&\frac{1}{\left(\frac{x}{(y;y)_{\infty}};\frac{x}{(y;y)_{\infty}}\right)_{\infty}(y;y)_{\infty}}.
\end{align*}
\end{proof}

In the next subsection, we compute the equivariant K-theory index of $\cM^{1,1}(M_1,M_2)$ using $\mathbf{T}_{\mathrm{edge}}$ localization. This is a weighted version of \eqref{fixed pts generating function}, where the $\bC$ coefficients will be replaced by $K_{\mathbf{T}_{\mathrm{edge}}}(\mathrm{pt})_{\mathrm{loc}}$ coefficients, the localized ring of characters of $\mathbf{T}_{\mathrm{edge}}$.

\begin{proposition}\label{compactness 2}
Assume that $\xi\neq 0$, then $\cM^{N,K}(M_1,M_2)^{\mathbf{T}_{\mathrm{edge}}}$ is proper (compact).
\end{proposition}

\begin{proof}
We only need to show that $\cM^N(M_1)^{\mathbf{T}'_1}$ is proper. $\cM^N(M_1)^{\mathbf{T}'_1}$ is the moduli of stable representations of ADHM quivers where the gauge node is a $\mathbf{T}'_1$ representation, and the flavour node is a trivial $\mathbf{T}'_1$ representation, and $(X,Y,I,J)$ are equivariant maps. The stability condition $\bC\langle X,Y,\rangle\mathrm{Im}(I)=\bC^{M_1}$ implies that the $\mathbf{T}'_1$ weights of $\bC^{M_1}$ are positive. Since $J$ has weight $(1,1)$ and $\bC^N$ has weight zero, $J$ must be zero. This shows that $\cM^N(M_1)^{\mathbf{T}'_1}=\mathrm{Quot}^{M_1}(\mathcal O_{\bC^2}^N)^{\mathbf{T}'_1}$, and the latter is proper as we have shown in the proof of Corollary \ref{compactness 1}.
\end{proof}

\subsection{Equivariant K-theory index}\label{sec:5.5}
Following \cite{Nekrasov:2014nea}, the equivariant K-theory index of the moduli space $\cM^{1,1}(M_1,M_2)$ is defined by
\ie 
\chi(\cM^{1,1}(M_1,M_2),\mathcal O^{\mathrm{vir}})
\fe
where $\mathcal O^{\mathrm{vir}}$ is the virtual structure sheaf of $\cM^{1,1}(M_1,M_2)$. Since our moduli space $\cM^{1,1}(M_1,M_2)$ is l.c.i by Proposition \ref{prop: regular embedding}, the virtual structure sheaf is the actual structure sheaf, i.e. $\mathcal O^{\mathrm{vir}}=\mathcal O$. The equivariant K-theory partition function is defined by
\ie
Z_{1,1}=\sum_{M_1,M_2}\chi(\cM^{1,1}(M_1,M_2),\mathcal O^{\mathrm{vir}})x^{M_1}y^{M_2}=\sum_{M_1,M_2}\chi(\cM^{1,1}(M_1,M_2))x^{M_1}y^{M_2}
\fe
Using equivariant localization, the $\mathbf{T}_{\mathrm{edge}}$-equivariant K-theory index of $\cM^{1,1}(M_1,M_2)$ can be written in terms of $\mathbf{T}_{\mathrm{edge}}$-fixed points:
\ie 
\chi(\cM^{1,1}(M_1,M_2),\mathcal O)=\chi\left(\cM^{1,1}(M_1,M_2)^{\mathbf{T}_{\mathrm{edge}}},\left(S^{\bullet}(T^{\mathrm{vir}})^*\right)|_{\cM^{1,1}(M_1,M_2)^{\mathbf{T}_{\mathrm{edge}}}}\right)
\fe
where $T^{\mathrm{vir}}$ is the virtual tangent bundle \eqref{tangent of our moduli}. 

Let us denote the equivariant parameters of $\mathbf{T}'_1$ by $r_1,r_2$, and denote the equivariant parameters of $\mathbf{T}'_2$ by $s_1,s_2$, and that of $\mathbf{T}_3$ by $t$. Then the $\mathbf{T}_{\mathrm{edge}}$-equivariant K-theory class of the virtual tangent bundle \eqref{tangent of our moduli} can be written as 
\ie\label{virtual tangent character}
\cV_1+\cV_1^*r_1r_2-\cV_1\cV_1^*(1-r_1)(1-r_2)+\cV_2+\cV_2^*s_1s_2-\cV_2\cV_2^*(1-s_1)(1-s_2)+t\cV_2\cV_1^*
\fe
where $\cV_1$ and $\cV_2$ are universal bundles on $\cM^{1,1}(M_1,M_2)$ of rank $M_1$ and $M_2$ respectively, note that they are equivariant under the $\mathbf{T}_{\mathrm{edge}}$ action. At a fixed point labelled by $$(Y^{(1)},Y^{(2)}_0,Y^{(2)}_1,\cdots, Y^{(2)}_{M_1})$$ the fibers of $\cV_1$ and $\cV_2$ are representations of $\mathbf{T}_{\mathrm{edge}}$, and their $\mathbf{T}_{\mathrm{edge}}$ characters can be described as following. Let us introduce the notation
\ie 
w\left(Y^{(1)}\right)=\sum_{(i,j)\in Y^{(1)}}r_1^{-i}r_2^{-j},\; w\left(Y^{(2)}_{\alpha}\right)=\sum_{(i,j)\in Y^{(2)}_{\alpha}}s_1^{-i}s_2^{-j}, \alpha\in \{0,\cdots,M_1\}.
\fe
From the description of the $\mathbf{T}_{\mathrm{edge}}$ fixed points, we know that $Y^{(2)}_1,\cdots,Y^{(2)}_{M_1}$ are one-to-one correspond to boxes $(i,j)$ in the diagram $Y^{(1)}$, so we use the notation $Y^{(2)}_{(i,j)}$ to indicate their relationship to $Y^{(1)}$. Then the characters of $\cV_1$ and $\cV_2$ are
\ie 
\cV_1=w\left(Y^{(1)}\right),\; \cV_2=w\left(Y^{(2)}_0\right)+t^{-1}\sum_{(i,j)\in Y^{(1)}}r_1^{-i}r_2^{-j}w\left(Y^{(2)}_{(i,j)}\right)
\fe
For example, if $\Phi=\mathcal O$, then the $\mathbf{T}_{\mathrm{edge}}$-equivariant K-theory index is
\ie\label{Index for Phi=0}
\chi(\cM^{1,1}(M_1,M_2),\mathcal O)=\sum_{Y^{(1)},Y^{(2)}_0,Y^{(2)}_{(i,j)}}S^{\bullet}(T^{\mathrm{vir}})^*
\fe
where $T^{\mathrm{vir}}$ is given by \eqref{virtual tangent character}, and $S^{\bullet}(T^{\mathrm{vir}})$ is the plethystic exponential of the virtual character $T^{\mathrm{vir}}$, and the $*$ operator maps $(r_1,r_2,s_1,s_2,t)$ to $(r_1^{-1},r_2^{-1},s_1^{-1},s_2^{-1},t^{-1})$.
\\
\newline
{\bf{Example}}\\
\newline
Let us compute the example $M_1=M_2=1$ explicitly. In this case there is only one possible $Y^{(1)}$, which is a single box at $(0,0)$. There are two $\mathbf{T}_{\mathrm{edge}}$-fixed points:
\ie
Y^{(2)}_0&=\square, Y^{(2)}_1=\emptyset,\; \text{and}\\
Y^{(2)}_0&=\emptyset, Y^{(2)}_1=\square, 
\fe
and the virtual tangent characters at these fixed points are
\ie 
T^{\mathrm{vir}}&=r_1+r_2+s_1+s_2+t,\; \text{and}\\
T^{\mathrm{vir}}&=r_1+r_2+s_1+s_2+(t-1)s_1s_2+t^{-1},
\fe
respectively. Plug these into \eqref{Index for Phi=0} we get
\ie
\chi(\cM^{1,1}(1,1),\mathcal O)&=\left(\frac{1}{(1-r_1)(1-r_2)(1-s_1)(1-s_2)}\left(\frac{1}{1-t}+\frac{1-s_1s_2}{(1-ts_1s_2)(1-t^{-1})}\right)\right)^*\\
&=\frac{1}{(1-r_1^{-1})(1-r_2^{-1})(1-s_1^{-1})(1-s_2^{-1})}\left(\frac{1}{1-t^{-1}}+\frac{1-s_1^{-1}s_2^{-1}}{(1-t^{-1}s_1^{-1}s_2^{-1})(1-t)}\right)\\
&=\frac{r_1r_2s_1^2s_2^2t}{(1-r_1)(1-r_2)(1-s_1)(1-s_2)(ts_1s_2-1)}
\fe
\noindent{\bf{The limit $s_1s_2\to 1$}}\\
\newline
In the example that we computed above, we notice that when $s_1s_2\to 1$, the contribution from the fixed point $Y^{(2)}_0=\emptyset, Y^{(2)}_1=\square$ vanishes. This is not a coincidence:
\begin{proposition}\label{prop: the limit}
In the limit $s_1s_2\to 1$, we have 
\ie 
\chi(\cM^{1,1}(M_1,M_2),\mathcal O)_{s_1s_2\to 1}=\chi(\mathring{\cM}^{1,1}(M_1,M_2),\mathcal O)_{s_1s_2\to 1}.
\fe
In other words, the fixed points outside of $\mathring{\cM}^{1,1}(M_1,M_2)$ do not contribute in this limit.
\end{proposition}

\begin{proof}
Suppose that we have a fixed point labelled by $M_1+2$ Young diagrams $Y^{(1)},Y^{(2)}_0,Y^{(2)}_{(i,j)}$, where $(i,j)$ takes value in coordinate of box in $Y^{(1)}$. We need to show that $S^{\bullet}(T^{\mathrm{vir}})$ vanishes in the limit $s_1s_2\to 1$ when there is an $(i,j)\in Y^{(1)}$ such that $Y^{(2)}_{(i,j)}$ is non-empty. We claim that it is enough to show that the virtual character $T^{\mathrm{vir}}_{s_1s_2\to 0}$ has negative constant term. To prove the claim, we notice that constant term in $T^{\mathrm{vir}}_{s_1s_2\to 0}$ comes from $$\sum_{k}a_k(s_1s_2)^k-\sum_{l}b_l(s_1s_2)^l$$in $T^{\mathrm{vir}}$, where $a_k,b_l$ are positive integers. After taking the plethystic exponential, this term becomes
$$\frac{\prod_l(1-s_1^ls_2^l)^{b_l}}{\prod_k(1-s_1^ks_2^k)^{a_k}}$$and this becomes zero in the limit $s_1s_2\to 1$ if $\sum_l b_l>\sum_ka_k$, i.e. $T^{\mathrm{vir}}_{s_1s_2\to 0}$ has negative constant term.

Therefore, it remains to compute the constant term in $T^{\mathrm{vir}}_{s_1s_2\to 0}$. Let us introduce some notations. Define
\ie 
\cV_2^{(0)}=w\left(Y^{(0)}_0\right),\; \cV_2^{(1)}=\sum_{(i,j)\in Y^{(1)}}r_1^{-i}r_2^{-j}w\left(Y^{(2)}_{(i,j)}\right),
\fe
so $\cV_2=\cV_2^{(0)}+t^{-1}\cV_2^{(1)}$. For every $V\in K_{\bC^2}(\mathrm{pt})=\bC[u_1,u_1^{-1},u_2,u_2^{-1}]$, define
\ie 
T(V)=V+V^*-VV^*(1-u_1)(1-u_2)
\fe
For example, let $\bC^2$ be $\mathbf{T}'_1$ then $u_1=r_1,u_2=r_2$, and $T(\cV_1)=\cV_1+\cV_1^*r_1r_2-\cV_1\cV_1^*(1-r_1)(1-r_2)$. Using this notation, we expand \eqref{virtual tangent character} as following:
\ie\label{expand virtual tangent character}
T^{\mathrm{vir}}&=T(\cV_1)+T(\cV_2^{(0)})+t^{-1}\cV_2^{(1)}+ts_1s_2\cV_2^{(1)}-\cV_2^{(1)}\cV_2^{(1)*}(1-s_1)(1-s_2)\\
&-t\cV_2^{(0)}\cV_2^{(1)*}(1-s_1)(1-s_2)-t^{-1}\cV_2^{(0)*}\cV_2^{(1)}(1-s_1)(1-s_2)+t\cV_2^{(0)}\cV_1^*+\cV_2^{(1)}\cV_1^*.
\fe
Recall that if $V$ is the weight space of a Young diagram $Y$, then \cite[Lemma 6]{Carlsson:2008}
\ie\label{simple tangent character}
T(V)=\sum_{\square\in Y}u_1^{-l(\square)}u_2^{a(\square)+1}+\sum_{\square\in Y}u_1^{l(\square)+1}u_2^{-a(\square)},
\fe
where $a(\square)$ and $l(\square)$ are arm-length and leg-length respectively:
\ie 
a(\square)=\#\{j'>j|(i,j')\in Y\},\; l(\square)=\#\{i'>i|(i',j)\in Y\}.
\fe 
In particular, we see that constant terms in $T(\cV_1)_{s_1s_2\to 0}$ and $T(\cV_2^{(0)})_{s_1s_2\to 0}$ are zero. Next, we simply drop all the other terms in \eqref{expand virtual tangent character} which involve $t$ or $t^{-1}$ because they can not be constant in the limit $s_1s_2\to 0$. Therefore, we only need to compute the constant term in 
\ie\label{reduced virtual tangent character}
\cV_2^{(1)}\cV_1^*-\cV_2^{(1)}\cV_2^{(1)*}(1-s_1)(1-s_2)
\fe
in the limit $s_1s_2\to 0$. Let us take $\mathbf{T}'_1$-invariant first, this amounts to take those terms in \eqref{reduced virtual tangent character} which do not involve $r_1$ or $r_2$, and the result is
\ie
\sum_{(i,j)\in Y^{(1)}}w\left(Y^{(2)}_{(i,j)}\right)-\sum_{(i,j)\in Y^{(1)}}w\left(Y^{(2)}_{(i,j)}\right)w\left(Y^{(2)}_{(i,j)}\right)^*(1-s_1)(1-s_2).
\fe
Notice that this can be rewritten as 
\ie 
\sum_{(i,j)\in Y^{(1)}}T\left(w\left(Y^{(2)}_{(i,j)}\right)\right)-\sum_{(i,j)\in Y^{(1)}}w\left(Y^{(2)}_{(i,j)}\right)^*s_1s_2,
\fe
and each individual $T\left(w\left(Y^{(2)}_{(i,j)}\right)\right)$ has zero constant term in the limit $s_1s_2\to 0$, as can be seen from \eqref{simple tangent character}. Thus the constant term of $T^{\mathrm{vir}}_{s_1s_2\to 1}$ equals to the constant term of 
\ie 
\sum_{(i,j)\in Y^{(1)}}w\left(Y^{(2)}_{(i,j)}\right)^*_{s_1s_2\to 1},
\fe
which is minus the number of non-empty Young diagrams in $Y^{(2)}_{(i,j)}$. In particular, it is negative if one of $Y^{(2)}_{(i,j)}$ is non-empty. This concludes the proof of proposition.
\end{proof}
\noindent{\bf{The index of the open subset $\mathring{\cM}^{1,1}(M_1,M_2)$}}\\
\newline
We have seen in \eqref{Open locus is a Hom bundle} that $\mathring{\cM}^{1,1}(M_1,M_2)$ is the vector bundle $\mathrm{Hom}(\cV_1,\cV_2)$, so its ring of functions is 
\ie 
\bigoplus_{k\ge 0}t^kS^k(\cV_2\otimes \cV_1^*)
\fe 
as a $\mathbf{T}_{\mathrm{edge}}$-equivariant sheaf on $\cM(M_1,1)\times \cM(M_2,1)$. It is easy to see that
\ie\label{decomposition of symmetric product}
S^k(\cV_2\otimes \cV_1^*)=\bigoplus_{|\underline{\lambda}|=k} S^{\underline{\lambda}}(\cV_2)\otimes S^{\underline{\lambda}}(\cV_1^*),
\fe 
where $\underline\lambda$ is a partition of $k$ (given by a Young diagram) and $S^{\underline{\lambda}}(\cV_2)$ is the irreducible representation of $\mathrm{GL}(\cV_2)$ defined by
\ie 
\mathrm{Hom}_{S_k}(R_{\underline\lambda},\cV_2^{\otimes k}).
\fe 
Here $R_{\underline\lambda}$ is the irreducible representation of permutation group $S_k$ with the Young diagram $\underline\lambda$. For example,
\ie 
S^{(k)}(\cV_2)=S^k(\cV_2),\; S^{(1,1,\cdots, 1)}(\cV_2)=\wedge^k(\cV_2).
\fe 
Using the decomposition \eqref{decomposition of symmetric product}, we can write the $\mathbf{T}_{\mathrm{edge}}$-equivariant K-theory index of $\mathring{\cM}^{1,1}(M_1,M_2)$ as
\ie 
\chi(\mathring{\cM}^{1,1}(M_1,M_2),\mathcal O)=\sum_{\underline\lambda}t^{|\underline\lambda|}\chi(\cM(M_1,1),S^{\underline{\lambda}}(\cV_1^*))\chi(\cM(M_2,1),S^{\underline{\lambda}}(\cV_2)),
\fe 
where the sum is for all Young diagrams $\underline\lambda$. Note that
\ie 
\chi(\cM(M_1,1),S^{\underline{\lambda}}(\cV_1^*))=r_1^{-2M_1}r_2^{-2M_1}\chi(\cM(M_1,1),S^{\underline{\lambda}}(\cV_1))^*.
\fe 
Therefore, we have 
\ie\label{factorization of index}
\chi(\mathring{\cM}^{1,1}(M_1,M_2),\mathcal O)=r_1^{-2M_1}r_2^{-2M_1}\sum_{\underline\lambda}t^{|\underline\lambda|}\chi(\cM(M_1,1),S^{\underline{\lambda}}(\cV_1))^*\chi(\cM(M_2,1),S^{\underline{\lambda}}(\cV_2)).
\fe 
Recall that the equivariant K-theory partition function of $\cM^{1,1}$ is defined as
\ie 
Z_{1,1}=\sum_{M_1,M_2}\chi(\cM^{1,1}(M_1,M_2),\mathcal O)x^{M_1}y^{M_2}.
\fe 
Combine \eqref{factorization of index} with Propistion \ref{prop: the limit}, we then have
\begin{proposition}\label{prop: factorization}
In the limit $r_1r_2\to 1,s_1s_2\to 1$, the equivariant K-theory partition function of $\cM^{1,1}$ factorizes as
\ie 
Z_{1,1}=\sum_{\underline\lambda}t^{|\underline\lambda|}\mathcal Z^{\underline{\lambda}}_1(r_1,r_2,x)_{r_1r_2\to 1}^*\mathcal Z^{\underline{\lambda}}_1(s_1,s_2,y)_{s_1s_2\to 1}
\fe 
where $\mathcal Z^{\underline{\lambda}}_1(r_1,r_2,x)$ is the equivariant K-theory partition function for ADHM quiver:
\ie\label{halffactor}
\mathcal Z^{\underline{\lambda}}_1(r_1,r_2,x)=\sum _{M}\chi(\cM(M,1),S^{\underline{\lambda}}(\cV))x^M.
\fe 
\end{proposition}

\begin{remark}
Consider the ADHM quiver with framing rank one and gauge rank $M_1$, and an extra framing node with rank $M_2$, see Figure \ref{WittenIndexQuiver}, then it is easy to see that the corresponding quiver variety is the universal bundle $\cV$ on the ADHM quiver without extra framing. Note that there is an action of $\mathbf{T}'_1\times \mathbf{T}_3$ on the quiver variety and the equivariant K-theory index for this variety is 
\ie 
\sum_{\underline{\lambda}}\chi(\cM(M,1),S^{\underline{\lambda}}(\cV))t^{|\underline\lambda|}.
\fe
This is exactly the coefficient in the summation \eqref{halffactor}.
\end{remark}

For general $N,K$, we know that $\cM^{N,K}(M_1,M_2)^{\mathbf{T}_{\mathrm{edge}}}$ is compact, so $\mathbf{T}_{\mathrm{edge}}$-equivariant K-theory index of $\cM^{N,K}(M_1,M_2)$ is well-defined. To compute this index, we turn on the maximal torus of the flavour symmtry, i.e. consider the diagonal action of $\mathbf{T}_{\mathrm{framing}}=\bC^{\times N}\times \bC^{\times N}$ on the framing vector space, and then take the $\mathbf{T}_{\mathrm{edge}}\times \mathbf{T}_{\mathrm{framing}}$ fixed points of $\cM^{N,K}(M_1,M_2)$ and this will be disjoint union of points (scheme-theoretically), one could follow the same argument in the computation of $N=K=1$ case and compute the $\mathbf{T}_{\mathrm{edge}}\times \mathbf{T}_{\mathrm{framing}}$-equivariant K-theory index of $\cM^{N,K}(M_1,M_2)$ and then take $\mathbf{T}_{\mathrm{framing}}$ equivariant parameters to $1$ and this would be the $\mathbf{T}_{\mathrm{edge}}$-equivariant K-theory index of $\cM^{N,K}(M_1,M_2)$.

\subsection{Connection to the 5d Chern-Simons algebra of operators}\label{app:c.6}

Notice that in the $t\to 0, r_1r_2\to 1,s_1s_2\to 1$ limit, the equivariant K-theory partition function of $\cM^{1,1}$ equals 
\ie 
\sum_{M,N}\chi(\cM(M,1),\mathcal O)\chi(\cM(N,1),\mathcal O)x^My^N,
\fe 
i.e. it is a product of equivariant K-theory partition functions of ADHM quivers:
\ie 
\lim_{t\to 0}Z_{1,1}=Z_{\mathrm{ADHM}}(r_1,r_2,x)_{r_1r_2\to 1}Z_{\mathrm{ADHM}}(s_1,s_2,y)_{s_1s_2\to 1}.
\fe 
Since the equivariant K-theory of ADHM quiver varieties admits a structure of Verma modules of the affine quantum group $U_q(\hat{\mathfrak{gl}}_1)$, we see that $t=0$ part of $Z_{1,1}$ is a character of a Verma module for $U_q(\hat{\mathfrak{gl}}_1)\otimes U_q(\hat{\mathfrak{gl}}_1)$. After passing from K-theory to cohomology, part of $Z_{1,1}$ is a character of a Verma module for the affine Yangian of $\mathfrak{gl}_1\oplus \mathfrak{gl}_1$, which is exactly the algebra of gauge invariant local observables of the 5d holomorphic Chern-Simons theory with gauge group $U(1)\times U(1)$.

\section{Connectedness of $\cM^{N,K}(M_1,M_2)$}
In this appendix, we show that
\begin{proposition}
$\cM^{N,K}(M_1,M_2)$ is connected.
\end{proposition}
First of all, if $\xi=0$, then the moduli space is the affine quotient of the solutions to ADHM equations \eqref{ADHM eqn} by $\mathrm{GL}_N\times \mathrm{GL}_K$. Consider the action of $\bC^{\times}$ by
\ie\label{loop contraction}
(X,Y,I,J,\tilde X,\tilde Y,\tilde I,\tilde J,S)\mapsto (rX,rY,rI,rJ,r\tilde X,r\tilde Y,r\tilde I,r\tilde J,rS), 
\fe 
then it contracts the $\cM^{N,K}_{\xi=0}(M_1,M_2)$ to a single point $(0,0,\cdots,0)$, thus $\cM^{N,K}_{\xi=0}(M_1,M_2)$ is connected. 

Next, we show that $\cM^{N,K}(M_1,M_2)$ is connected when $\xi>0$, and the $\xi<0$ case is similar. Consider the $\bC^{\times}$ action \eqref{loop contraction}, its fixed points in $\cM^{N,K}(M_1,M_2)$ correspondes to quiver representations such that $\bC^{M_1}$ and $\bC^{M_2}$ are $\bZ$-graded and all arrows in the quiver increase the grade by one. From the stability condition, $\bC\langle \tilde X,\tilde Y\rangle (\mathrm{Im}(\tilde I)+\mathrm{Im}(S))=\bC^{M_2}$ we deduce that $\bC^{M_2}$ must be positively graded. Since the framing vector space has grade zero, we conclude that $\tilde J=0$, i.e. $\cM^{N,K}(M_1,M_2)^{\bC^{\times}}\subset \cM^{N,K}_{\tilde J=0}(M_1,M_2)$.

Note that the natural projection $q:\cM^{N,K}(M_1,M_2)\to \cM^{N,K}_{\xi=0}(M_1,M_2)$ defined by semisimplification is projective and equivariant under the $\bC^{\times}$ action. We claim that $\forall x\in \cM^{N,K}(M_1,M_2)$, its limit under the $\bC^{\times}$ action $\lim_{r\to 0}r(x)$ exists. This follows from the fact that $\lim_{r\to 0}r(q(x))$ exists and the valuative criterion for a proper morphism. Since $\forall s\in \bC^{\times }$, we have 
\ie 
s\left(\underset{r\to 0}{\lim}r(x)\right)=\underset{r\to 0}{\lim}sr(x)=\underset{r'\to 0}{\lim}r'(x),
\fe 
we see that the limit $\lim_{r\to 0}r(x)\in \cM^{N,K}(M_1,M_2)^{\bC^{\times}}$, thus every point in $\cM^{N,K}(M_1,M_2)$ is connected to some point in $\cM^{N,K}(M_1,M_2)^{\bC^{\times}}$. As there is an inclusion $\cM^{N,K}(M_1,M_2)^{\bC^{\times}}\subset \cM^{N,K}_{\tilde J=0}(M_1,M_2)$, we only need to show that $\cM^{N,K}_{\tilde J=0}(M_1,M_2)$ is connected. In fact, we prove a stronger result

\begin{proposition}\label{Irreducible J=0 variety}
Assume that $\xi>0$, the $\cM^{N,K}_{\tilde J=0}(M_1,M_2)$ is irreducible of dimension $2NM_1+M_2(M_1+K+1)$.
\end{proposition}

\begin{proof}
The same argument in Lemma \ref{Fibration over Hilb} shows that there is locally trivial fibration $$\pi:\cM^{N,K}_{\tilde J=0}(M_1,M_2)\to \cM^{N}(M_1)$$with fibers isomorphic to $\mathrm{Quot}^{M_2}(\mathcal O_{\bC^2}^{M_1+K})$. It is known that the Nakajima quiver variety $\cM^{N}(M_1)$ is irreducible\footnote{The Kirwan map $K_{\mathrm{GL_{M_1}}}(\mathrm{pt})\to K(\cM^{N}(M_1))$ is surjective by \cite{Nevins:2016}, so every K-theory class on $\cM^{N}(M_1)$ must have constant rank, then it follows that $\cM^{N}(M_1)$ is connected thus it is irreducible.} of dimension $2NM_1$, so we only need to show that $\mathrm{Quot}^{M_2}(\mathcal O_{\bC^2}^{M_1+K})$ is irreducible of dimension $M_2(M_1+K+1)$, and this is proven in the next lemma.
\end{proof}

\begin{lemma}
Let $\mathcal S$ be a smooth connected algebraic surface and $\mathcal V$ is a vector bundle of rank $N$ on $\mathcal S$, then the Quot scheme $\mathrm{Quot}^{M}(\mathcal V)$ is irreducible of dimension $(N+1)M$.
\end{lemma}

\begin{proof}
Consider the Hilbert-Chow map $h:\mathrm{Quot}^{M}(\mathcal V)\to \mathrm{Sym}^{M}\mathcal S$. Its fiber over a point
\ie 
(\underbrace{x_1,\cdots,x_1}_{L_1},\underbrace{x_2,\cdots,x_2}_{L_2},\cdots,\underbrace{x_n,\cdots,x_n}_{L_n})\in \mathrm{Sym}^{M}\mathcal S,\; M=\sum_{i=1}^n L_i
\fe 
is isomorphic to 
\ie 
\prod_{i=1}^{n} \mathrm{Quot}^{L_i}(\bC[\![x,y]\!]^{N})
\fe
where $\mathrm{Quot}^{L_i}(\bC[\![x,y]\!]^{N})$ is the \textit{formal Quot scheme} defined by the moduli space of quotient module of $\bC[\![x,y]\!]^{N}$ of $\bC$-dimension $L_i$ on the ring $\bC[\![x,y]\!]$. We claim that $\mathrm{Quot}^{L_i}(\bC[\![x,y]\!]^{N})$ is irreducible. The irreducibility of $\mathrm{Quot}^{M}(\mathcal V)$ follows from this claim and the irreducibility of $\mathrm{Sym}^{M}\mathcal S$.

To prove this claim, we notice that $\mathrm{Quot}^{L_i}(\bC[\![x,y]\!]^{N})$ has the form $U/\mathrm{GL}_{L_i}$ where $U$ is an open subvariety of the following variety
\ie\label{formal quot scheme before quotient}
\{X,Y\in\mathfrak{gl}_{L_i}|[X,Y]=0,\text{ X and Y are nilpotent}\}\times \mathrm{Hom}(\bC^N,\bC^{L_i}).
\fe 
The first factor in \eqref{formal quot scheme before quotient} is irreducible \cite{Baranovsky:2001}, and the second factor is obviously irreducible, so $U$ is an irreducible variety, thus $\mathrm{Quot}^{L_i}(\bC[\![x,y]\!]^{N})$ is irreducible.

Finally, the dimension of $\mathrm{Quot}^{M}(\mathcal V)$ can be computed using the open locus $h^{-1}(V)$, where $V\subset \mathrm{Sym}^{M}\mathcal S$ is the open subset of points $(x_1,\cdots,x_M)$ where $x_i$ are distinct. Restrict to $h^{-1}(V)$, $h$ is a locally trivial fibration with fibers isomorphic to product of $M$ copies of $\mathbb{P}^{N-1}$, so 
\ie 
\dim \mathrm{Quot}^{M}(\mathcal V)=\dim h^{-1}(V)=(N-1)M+2M=(N+1)M
\fe 
\end{proof}

\bibliographystyle{ytphys}
\bibliography{G2.bib}

\providecommand{\href}[2]{#2}\begingroup\raggedright\begin{thebibliography}{100}

\bibitem{Halverson:2014tya}
J.~Halverson and D.~R. Morrison, ``{The landscape of M-theory compactifications
  on seven-manifolds with G$_{2}$ holonomy},''
  \href{http://dx.doi.org/10.1007/JHEP04(2015)047}{{\em JHEP} {\bfseries 04}
  (2015) 047},
\href{http://arxiv.org/abs/1412.4123}{{\ttfamily arXiv:1412.4123 [hep-th]}}.

\bibitem{Halverson:2015vta}
J.~Halverson and D.~R. Morrison, ``{On gauge enhancement and singular limits in
  G$_{2}$ compactifications of M-theory},''
  \href{http://dx.doi.org/10.1007/JHEP04(2016)100}{{\em JHEP} {\bfseries 04}
  (2016) 100},
\href{http://arxiv.org/abs/1507.05965}{{\ttfamily arXiv:1507.05965 [hep-th]}}.

\bibitem{Braun:2016igl}
A.~P. Braun, ``{Tops as building blocks for G$_{2}$ manifolds},''
  \href{http://dx.doi.org/10.1007/JHEP10(2017)083}{{\em JHEP} {\bfseries 10}
  (2017) 083},
\href{http://arxiv.org/abs/1602.03521}{{\ttfamily arXiv:1602.03521 [hep-th]}}.

\bibitem{daCGuio:2017ifs}
T.~C. da~C.~Guio, H.~Jockers, A.~Klemm, and H.-Y. Yeh, ``{Effective Action from
  M-Theory on Twisted Connected Sum G$_{2}$-Manifolds},''
  \href{http://dx.doi.org/10.1007/s00220-017-3045-0}{{\em Commun. Math. Phys.}
  {\bfseries 359} no.~2, (2018) 535--601},
  \href{http://arxiv.org/abs/1702.05435}{{\ttfamily arXiv:1702.05435
  [hep-th]}}.

\bibitem{Braun:2017ryx}
A.~P. Braun and M.~Del~Zotto, ``{Mirror Symmetry for $G_2$-Manifolds: Twisted
  Connected Sums and Dual Tops},''
  \href{http://dx.doi.org/10.1007/JHEP05(2017)080}{{\em JHEP} {\bfseries 05}
  (2017) 080},
\href{http://arxiv.org/abs/1701.05202}{{\ttfamily arXiv:1701.05202 [hep-th]}}.

\bibitem{Braun:2017uku}
A.~P. Braun and S.~Sch\"afer-Nameki, ``{Compact, Singular $G_2$-Holonomy
  Manifolds and M/Heterotic/F-Theory Duality},''
  \href{http://dx.doi.org/10.1007/JHEP04(2018)126}{{\em JHEP} {\bfseries 04}
  (2018) 126}, \href{http://arxiv.org/abs/1708.07215}{{\ttfamily
  arXiv:1708.07215 [hep-th]}}.

\bibitem{Braun:2017csz}
A.~P. Braun and M.~Del~Zotto, ``{Towards Generalized Mirror Symmetry for
  Twisted Connected Sum $G_2$ Manifolds},''
  \href{http://dx.doi.org/10.1007/JHEP03(2018)082}{{\em JHEP} {\bfseries 03}
  (2018) 082}, \href{http://arxiv.org/abs/1712.06571}{{\ttfamily
  arXiv:1712.06571 [hep-th]}}.

\bibitem{Braun:2018fdp}
A.~P. Braun, M.~Del~Zotto, J.~Halverson, M.~Larfors, D.~R. Morrison, and
  S.~Sch\"afer-Nameki, ``{Infinitely many M2-instanton corrections to M-theory
  on G$_{2}$-manifolds},''
  \href{http://dx.doi.org/10.1007/JHEP09(2018)077}{{\em JHEP} {\bfseries 09}
  (2018) 077}, \href{http://arxiv.org/abs/1803.02343}{{\ttfamily
  arXiv:1803.02343 [hep-th]}}.

\bibitem{Acharya:2018nbo}
B.~S. Acharya, A.~P. Braun, E.~E. Svanes, and R.~Valandro, ``{Counting
  associatives in compact $G_{2}$ orbifolds},''
  \href{http://dx.doi.org/10.1007/JHEP03(2019)138}{{\em JHEP} {\bfseries 03}
  (2019) 138}, \href{http://arxiv.org/abs/1812.04008}{{\ttfamily
  arXiv:1812.04008 [hep-th]}}.

\bibitem{Braun:2018vhk}
A.~P. Braun, S.~Cizel, M.~H\"ubner, and S.~Sch\"afer-Nameki, ``{Higgs bundles
  for M-theory on $G_{2}$-manifolds},''
  \href{http://dx.doi.org/10.1007/JHEP03(2019)199}{{\em JHEP} {\bfseries 03}
  (2019) 199}, \href{http://arxiv.org/abs/1812.06072}{{\ttfamily
  arXiv:1812.06072 [hep-th]}}.

\bibitem{Braun:2019lnn}
A.~P. Braun, S.~Majumder, and A.~Otto, ``{On Mirror Maps for Manifolds of
  Exceptional Holonomy},''
  \href{http://dx.doi.org/10.1007/JHEP10(2019)204}{{\em JHEP} {\bfseries 10}
  (2019) 204}, \href{http://arxiv.org/abs/1905.01474}{{\ttfamily
  arXiv:1905.01474 [hep-th]}}.

\bibitem{Braun:2019wnj}
A.~P. Braun, ``{M-Theory and Orientifolds},''
  \href{http://dx.doi.org/10.1007/JHEP09(2020)065}{{\em JHEP} {\bfseries 09}
  (2020) 065}, \href{http://arxiv.org/abs/1912.06072}{{\ttfamily
  arXiv:1912.06072 [hep-th]}}.

\bibitem{Barbosa:2019bgh}
R.~Barbosa, M.~Cveti\v{c}, J.~J. Heckman, C.~Lawrie, E.~Torres, and
  G.~Zoccarato, ``{T-branes and $G_2$ backgrounds},''
  \href{http://dx.doi.org/10.1103/PhysRevD.101.026015}{{\em Phys. Rev. D}
  {\bfseries 101} no.~2, (2020) 026015},
  \href{http://arxiv.org/abs/1906.02212}{{\ttfamily arXiv:1906.02212
  [hep-th]}}.

\bibitem{Acharya:2020vmg}
B.~S. Acharya, L.~Foscolo, M.~Najjar, and E.~E. Svanes, ``{New
  G$_{2}$-conifolds in M-theory and their field theory interpretation},''
  \href{http://dx.doi.org/10.1007/JHEP05(2021)250}{{\em JHEP} {\bfseries 05}
  (2021) 250}, \href{http://arxiv.org/abs/2011.06998}{{\ttfamily
  arXiv:2011.06998 [hep-th]}}.

\bibitem{Acharya:2021rvh}
B.~S. Acharya, A.~Kinsella, and D.~R. Morrison, ``{Non-Perturbative Heterotic
  Duals of M-Theory on $G_{2}$ Orbifolds},''
  \href{http://arxiv.org/abs/2106.03886}{{\ttfamily arXiv:2106.03886
  [hep-th]}}.

\bibitem{Ashmore:2021pdm}
A.~Ashmore, A.~Coimbra, C.~Strickland-Constable, E.~E.~Svanes and D.~Tennyson,
``Topological G$_{2}$ and Spin(7) strings at 1-loop from double complexes,''
JHEP \textbf{02}, 089 (2022)
doi:10.1007/JHEP02(2022)089
[arXiv:2108.09310 [hep-th]].

\bibitem{Cvetic:2021maf}
M.~Cveti\v{c}, J.~J.~Heckman, E.~Torres and G.~Zoccarato,
``Reflections on the matter of 3D N=1 vacua and local Spin(7) compactifications,''
Phys. Rev. D \textbf{105}, no.2, 026008 (2022)
doi:10.1103/PhysRevD.105.026008
[arXiv:2107.00025 [hep-th]].

\bibitem{Hubner:2020yde}
M.~Hubner, ``{Local G$_{2}$-manifolds, Higgs bundles and a colored quantum
  mechanics},'' \href{http://dx.doi.org/10.1007/JHEP05(2021)002}{{\em JHEP}
  {\bfseries 05} (2021) 002}, \href{http://arxiv.org/abs/2009.07136}{{\ttfamily
  arXiv:2009.07136 [hep-th]}}.
  
\bibitem{Cvetic:2020piw}
M.~Cveti\v{c}, J.~J.~Heckman, T.~B.~Rochais, E.~Torres and G.~Zoccarato,
``Geometric unification of Higgs bundle vacua,''
Phys. Rev. D \textbf{102}, no.10, 106012 (2020)
doi:10.1103/PhysRevD.102.106012
[arXiv:2003.13682 [hep-th]].

\bibitem{Earp:2011dh}
H.~N.~S.~Earp,
``Instantons on $G_2$?manifolds,''
PhD. Thesis.

\bibitem{earpthesis}
H.~N.~S.~Earp,
``$G_2$-instantons on Kovalev manifolds,''
[arXiv:1101.0880 [math.DG]].

\bibitem{MR2024648}
A.~Kovalev, ``Twisted connected sums and special {R}iemannian holonomy,''
  \href{http://dx.doi.org/10.1515/crll.2003.097}{{\em J. Reine Angew. Math.}
  {\bfseries 565} (2003) 125--160}.

\bibitem{Kovalev2010}
A.~Kovalev and J.~Nordstr{\"o}m, ``Asymptotically cylindrical 7-manifolds of
  holonomy g2 with applications to compact irreducible g2-manifolds,''
  \href{http://dx.doi.org/10.1007/s10455-010-9210-8}{{\em Annals of Global
  Analysis and Geometry} {\bfseries 38} no.~3, (2010) 221--257},
  \href{http://arxiv.org/abs/0907.0497}{{\ttfamily arXiv:0907.0497 [math.DG]}}.

\bibitem{MR3109862}
A.~Corti, M.~Haskins, J.~Nordstr{\"o}m, and T.~Pacini, ``Asymptotically
  cylindrical {C}alabi-{Y}au 3-folds from weak {F}ano 3-folds,''
  \href{http://dx.doi.org/10.2140/gt.2013.17.1955}{{\em Geom. Topol.}
  {\bfseries 17} no.~4, (2013) 1955--2059}.

\bibitem{Corti:2012kd}
A.~Corti, M.~Haskins, J.~Nordstr{\"o}m, and T.~Pacini,
  ``{$\mathrm{G}_{2}$-manifolds and associative submanifolds via semi-Fano
  $3$-folds},'' \href{http://dx.doi.org/10.1215/00127094-3120743}{{\em Duke
  Math. J.} {\bfseries 164} no.~10, (2015) 1971--2092},
\href{http://arxiv.org/abs/1207.4470}{{\ttfamily arXiv:1207.4470 [math.DG]}}.

\bibitem{Foscolo:2017vzf}
L.~Foscolo, M.~Haskins, and J.~Nordstr\"om, ``{Complete non-compact
  G2-manifolds from asymptotically conical Calabi-Yau 3-folds},''
  \href{http://arxiv.org/abs/1709.04904}{{\ttfamily arXiv:1709.04904
  [math.DG]}}.

\bibitem{Foscolo:2018mfs}
L.~Foscolo, M.~Haskins, and J.~Nordstr\"om, ``{Infinitely many new families of
  complete cohomogeneity one G\_2-manifolds: G\_2 analogues of the Taub-NUT and
  Eguchi-Hanson spaces},'' \href{http://arxiv.org/abs/1805.02612}{{\ttfamily
  arXiv:1805.02612 [math.DG]}}.

\bibitem{Nekrasov:2014nea}
N.~Nekrasov and A.~Okounkov, ``{Membranes and Sheaves},''
  \href{http://arxiv.org/abs/1404.2323}{{\ttfamily arXiv:1404.2323 [math.AG]}}.

\bibitem{Costello:2016mgj}
K.~Costello and S.~Li, ``{Twisted supergravity and its quantization},''
  \href{http://arxiv.org/abs/1606.00365}{{\ttfamily arXiv:1606.00365
  [hep-th]}}.

\bibitem{Costello:2016nkh}
K.~Costello, ``{M-theory in the Omega-background and 5-dimensional
  non-commutative gauge theory},''
  \href{http://arxiv.org/abs/1610.04144}{{\ttfamily arXiv:1610.04144
  [hep-th]}}.

\bibitem{Eager:2021ufo}
R.~Eager and F.~Hahner, ``{Maximally twisted eleven-dimensional
  supergravity},'' \href{http://arxiv.org/abs/2106.15640}{{\ttfamily
  arXiv:2106.15640 [hep-th]}}.

\bibitem{Saberi:2021weg}
I.~Saberi and B.~R. Williams, ``{Twisting pure spinor superfields, with
  applications to supergravity},''
  \href{http://arxiv.org/abs/2106.15639}{{\ttfamily arXiv:2106.15639
  [math-ph]}}.

\bibitem{Iqbal:2003ds}
A.~Iqbal, N.~Nekrasov, A.~Okounkov, and C.~Vafa, ``{Quantum foam and
  topological strings},''
  \href{http://dx.doi.org/10.1088/1126-6708/2008/04/011}{{\em JHEP} {\bfseries
  04} (2008) 011}, \href{http://arxiv.org/abs/hep-th/0312022}{{\ttfamily
  arXiv:hep-th/0312022}}.

\bibitem{DelZotto:2021gzy}
M.~Del~Zotto, N.~Nekrasov, N.~Piazzalunga, and M.~Zabzine, ``{Playing with the
  index of M-theory},'' \href{http://arxiv.org/abs/2103.10271}{{\ttfamily
  arXiv:2103.10271 [hep-th]}}.

\bibitem{Dijkgraaf:2006um}
R.~Dijkgraaf, C.~Vafa, and E.~Verlinde, ``{M-theory and a topological string
  duality},'' \href{http://arxiv.org/abs/hep-th/0602087}{{\ttfamily
  arXiv:hep-th/0602087}}.

\bibitem{Dijkgraaf:2004te}
R.~Dijkgraaf, S.~Gukov, A.~Neitzke, and C.~Vafa, ``{Topological M-theory as
  unification of form theories of gravity},''
  \href{http://dx.doi.org/10.4310/ATMP.2005.v9.n4.a5}{{\em Adv. Theor. Math.
  Phys.} {\bfseries 9} no.~4, (2005) 603--665},
  \href{http://arxiv.org/abs/hep-th/0411073}{{\ttfamily arXiv:hep-th/0411073}}.

\bibitem{deBoer:2005pt}
J.~de~Boer, A.~Naqvi, and A.~Shomer, ``{The Topological G(2) string},''
  \href{http://dx.doi.org/10.4310/ATMP.2008.v12.n2.a2}{{\em Adv. Theor. Math.
  Phys.} {\bfseries 12} no.~2, (2008) 243--318},
\href{http://arxiv.org/abs/hep-th/0506211}{{\ttfamily arXiv:hep-th/0506211}}.

\bibitem{Bonelli:2005ti}
G.~Bonelli and M.~Zabzine, ``{From current algebras for p-branes to topological
  M-theory},'' \href{http://dx.doi.org/10.1088/1126-6708/2005/09/015}{{\em
  JHEP} {\bfseries 09} (2005) 015},
  \href{http://arxiv.org/abs/hep-th/0507051}{{\ttfamily arXiv:hep-th/0507051}}.

\bibitem{Bonelli:2005rw}
G.~Bonelli, A.~Tanzini, and M.~Zabzine, ``{On topological M-theory},''
  \href{http://dx.doi.org/10.4310/ATMP.2006.v10.n2.a4}{{\em Adv. Theor. Math.
  Phys.} {\bfseries 10} no.~2, (2006) 239--260},
  \href{http://arxiv.org/abs/hep-th/0509175}{{\ttfamily arXiv:hep-th/0509175}}.

\bibitem{Bonelli:2006ph}
G.~Bonelli, A.~Tanzini, and M.~Zabzine, ``{Computing Amplitudes in topological
  M-theory},'' \href{http://dx.doi.org/10.1088/1126-6708/2007/03/023}{{\em
  JHEP} {\bfseries 03} (2007) 023},
  \href{http://arxiv.org/abs/hep-th/0611327}{{\ttfamily arXiv:hep-th/0611327}}.

\bibitem{Donaldson:1996kp}
S.~K. Donaldson and R.~P. Thomas, ``{Gauge theory in higher dimensions},'' in
  {\em {Conference on Geometric Issues in Foundations of Science in honor of
  Sir Roger Penrose's 65th Birthday}}, pp.~31--47.
\newblock 6, 1996.

\bibitem{Costello:2017fbo}
K.~Costello, ``{Holography and Koszul duality: the example of the $M2$
  brane},'' \href{http://arxiv.org/abs/1705.02500}{{\ttfamily arXiv:1705.02500
  [hep-th]}}.

\bibitem{Gaiotto:2019wcc}
D.~Gaiotto and J.~Oh, ``{Aspects of $\Omega$-deformed M-theory},''
  \href{http://arxiv.org/abs/1907.06495}{{\ttfamily arXiv:1907.06495
  [hep-th]}}.

\bibitem{CORR}
M.~Del~Zotto, ``{Geometric Engineering and Correspondences}; to appear,''.

\bibitem{Johansen:1994aw}
A.~Johansen, ``{Twisting of $N=1$ SUSY gauge theories and heterotic topological
  theories},'' \href{http://dx.doi.org/10.1142/S0217751X9500200X}{{\em Int. J.
  Mod. Phys. A} {\bfseries 10} (1995) 4325--4358},
  \href{http://arxiv.org/abs/hep-th/9403017}{{\ttfamily arXiv:hep-th/9403017}}.

\bibitem{Witten:1994ev}
E.~Witten, ``{Supersymmetric Yang-Mills theory on a four manifold},''
  \href{http://dx.doi.org/10.1063/1.530745}{{\em J. Math. Phys.} {\bfseries 35}
  (1994) 5101--5135}, \href{http://arxiv.org/abs/hep-th/9403195}{{\ttfamily
  arXiv:hep-th/9403195}}.

\bibitem{Johansen:1994ud}
A.~Johansen, ``{Infinite conformal algebras in supersymmetric theories on four
  manifolds},'' \href{http://dx.doi.org/10.1016/0550-3213(94)00408-7}{{\em
  Nucl. Phys. B} {\bfseries 436} (1995) 291--341},
  \href{http://arxiv.org/abs/hep-th/9407109}{{\ttfamily arXiv:hep-th/9407109}}.

\bibitem{Johansen:2003hw}
A.~Johansen, ``{Holomorphic currents and duality in N = 1 supersymmetric
  theories},'' \href{http://dx.doi.org/10.1088/1126-6708/2003/12/032}{{\em
  JHEP} {\bfseries 12} (2003) 032},
  \href{http://arxiv.org/abs/hep-th/0309125}{{\ttfamily arXiv:hep-th/0309125}}.

\bibitem{Closset:2013vra}
C.~Closset, T.~T. Dumitrescu, G.~Festuccia, and Z.~Komargodski, ``{The Geometry
  of Supersymmetric Partition Functions},''
  \href{http://dx.doi.org/10.1007/JHEP01(2014)124}{{\em JHEP} {\bfseries 01}
  (2014) 124}, \href{http://arxiv.org/abs/1309.5876}{{\ttfamily arXiv:1309.5876
  [hep-th]}}.

\bibitem{Closset:2014uda}
C.~Closset, T.~T. Dumitrescu, G.~Festuccia, and Z.~Komargodski, ``{From Rigid
  Supersymmetry to Twisted Holomorphic Theories},''
  \href{http://dx.doi.org/10.1103/PhysRevD.90.085006}{{\em Phys. Rev. D}
  {\bfseries 90} no.~8, (2014) 085006},
  \href{http://arxiv.org/abs/1407.2598}{{\ttfamily arXiv:1407.2598 [hep-th]}}.

\bibitem{Witten:2009xu}
E.~Witten, ``{Branes, Instantons, And Taub-NUT Spaces},''
  \href{http://dx.doi.org/10.1088/1126-6708/2009/06/067}{{\em JHEP} {\bfseries
  06} (2009) 067}, \href{http://arxiv.org/abs/0902.0948}{{\ttfamily
  arXiv:0902.0948 [hep-th]}}.

\bibitem{GarciaEtxebarria:2019caf}
I.~n. Garc\'\i{}a~Etxebarria, B.~Heidenreich, and D.~Regalado, ``{IIB flux
  non-commutativity and the global structure of field theories},''
  \href{http://dx.doi.org/10.1007/JHEP10(2019)169}{{\em JHEP} {\bfseries 10}
  (2019) 169}, \href{http://arxiv.org/abs/1908.08027}{{\ttfamily
  arXiv:1908.08027 [hep-th]}}.

\bibitem{Albertini:2020mdx}
F.~Albertini, M.~Del~Zotto, I.~n. Garc\'\i{}a~Etxebarria, and S.~S. Hosseini,
  ``{Higher Form Symmetries and M-theory},''
  \href{http://dx.doi.org/10.1007/JHEP12(2020)203}{{\em JHEP} {\bfseries 12}
  (2020) 203}, \href{http://arxiv.org/abs/2005.12831}{{\ttfamily
  arXiv:2005.12831 [hep-th]}}.

\bibitem{Morrison:2020ool}
D.~R. Morrison, S.~Schafer-Nameki, and B.~Willett, ``{Higher-Form Symmetries in
  5d},'' \href{http://dx.doi.org/10.1007/JHEP09(2020)024}{{\em JHEP} {\bfseries
  09} (2020) 024}, \href{http://arxiv.org/abs/2005.12296}{{\ttfamily
  arXiv:2005.12296 [hep-th]}}.

\bibitem{BryantSalamon}
R.~L. Bryant and S.~M. Salamon, ``{On the construction of some complete metrics
  with exceptional holonomy},''
  \href{https://doi.org/10.1215/S0012-7094-89-05839-0}{{\em Duke Mathematical
  Journal} {\bfseries 58} no.~3, (1989) 829 -- 850}.

\bibitem{Atiyah:2001qf}
M.~Atiyah and E.~Witten, ``{M theory dynamics on a manifold of G(2)
  holonomy},'' {\em Adv. Theor. Math. Phys.} {\bfseries 6} (2003) 1--106,
\href{http://arxiv.org/abs/hep-th/0107177}{{\ttfamily arXiv:hep-th/0107177}}.

\bibitem{Acharya:2001gy}
B.~S. Acharya and E.~Witten, ``{Chiral fermions from manifolds of G(2)
  holonomy},''
\href{http://arxiv.org/abs/hep-th/0109152}{{\ttfamily arXiv:hep-th/0109152}}.

\bibitem{Oh:2021bwi}
J.~Oh and Y.~Zhou, ``{A domain wall in twisted M-theory},''
  \href{http://arxiv.org/abs/2105.09537}{{\ttfamily arXiv:2105.09537
  [hep-th]}}.

\bibitem{Joyce:2016fij}
D.~Joyce, ``{Conjectures on counting associative 3-folds in $G_2$-manifolds},''
  \href{http://arxiv.org/abs/1610.09836}{{\ttfamily arXiv:1610.09836
  [math.DG]}}.

\bibitem{Dedushenko:2017tdw}
M.~Dedushenko, S.~Gukov, and P.~Putrov,
  \href{http://dx.doi.org/10.1093/oso/9780198802013.003.0011}{``{Vertex
  algebras and 4-manifold invariants},''} in {\em {Nigel Hitchin's 70th
  Birthday Conference}}, vol.~1, pp.~249--318.
\newblock 5, 2017.
\newblock \href{http://arxiv.org/abs/1705.01645}{{\ttfamily arXiv:1705.01645
  [hep-th]}}.

\bibitem{Cecotti:2019hyk}
S.~Cecotti, C.~Gerig, and C.~Vafa, ``{G$_{2}$ holonomy, Taubes\textquoteright{}
  construction of Seiberg-Witten invariants and superconducting vortices},''
  \href{http://dx.doi.org/10.1007/JHEP04(2020)038}{{\em JHEP} {\bfseries 04}
  (2020) 038}, \href{http://arxiv.org/abs/1909.10453}{{\ttfamily
  arXiv:1909.10453 [hep-th]}}.

\bibitem{Minahan:2015jta}
J.~A. Minahan and M.~Zabzine, ``{Gauge theories with 16 supersymmetries on
  spheres},'' \href{http://dx.doi.org/10.1007/JHEP03(2015)155}{{\em JHEP}
  {\bfseries 03} (2015) 155}, \href{http://arxiv.org/abs/1502.07154}{{\ttfamily
  arXiv:1502.07154 [hep-th]}}.

\bibitem{Prins:2018hjc}
D.~Prins, ``{Supersymmeric Gauge Theory on Curved 7-Branes},''
  \href{http://dx.doi.org/10.1002/prop.201900009}{{\em Fortsch. Phys.}
  {\bfseries 67} no.~5, (2019) 1900009},
  \href{http://arxiv.org/abs/1812.05349}{{\ttfamily arXiv:1812.05349
  [hep-th]}}.

\bibitem{Iakovidis:2020znp}
N.~Iakovidis, J.~Qiu, A.~Roc\'en, and M.~Zabzine, ``{7D supersymmetric
  Yang-Mills on hypertoric 3-Sasakian manifolds},''
  \href{http://dx.doi.org/10.1007/JHEP06(2020)026}{{\em JHEP} {\bfseries 06}
  (2020) 026}, \href{http://arxiv.org/abs/2003.12461}{{\ttfamily
  arXiv:2003.12461 [hep-th]}}.

\bibitem{Festuccia:2011ws}
G.~Festuccia and N.~Seiberg, ``{Rigid Supersymmetric Theories in Curved
  Superspace},'' \href{http://dx.doi.org/10.1007/JHEP06(2011)114}{{\em JHEP}
  {\bfseries 06} (2011) 114}, \href{http://arxiv.org/abs/1105.0689}{{\ttfamily
  arXiv:1105.0689 [hep-th]}}.

\bibitem{Acharya:1997gp}
B.~S. Acharya, M.~O'Loughlin, and B.~J. Spence, ``{Higher dimensional analogs
  of Donaldson-Witten theory},''
  \href{http://dx.doi.org/10.1016/S0550-3213(97)00515-4}{{\em Nucl. Phys. B}
  {\bfseries 503} (1997) 657--674},
  \href{http://arxiv.org/abs/hep-th/9705138}{{\ttfamily arXiv:hep-th/9705138}}.

\bibitem{Baulieu:1997em}
L.~Baulieu, H.~Kanno, and I.~M. Singer,
  \href{http://dx.doi.org/10.1142/9789814447287_0011}{``{Cohomological
  Yang-Mills theory in eight-dimensions},''} in {\em {APCTP Winter School on
  Dualities of Gauge and String Theories}}, pp.~365--373.
\newblock 4, 1997.
\newblock \href{http://arxiv.org/abs/hep-th/9705127}{{\ttfamily
  arXiv:hep-th/9705127}}.

\bibitem{Baulieu:1997jx}
L.~Baulieu, H.~Kanno, and I.~M. Singer, ``{Special quantum field theories in
  eight-dimensions and other dimensions},''
  \href{http://dx.doi.org/10.1007/s002200050353}{{\em Commun. Math. Phys.}
  {\bfseries 194} (1998) 149--175},
  \href{http://arxiv.org/abs/hep-th/9704167}{{\ttfamily arXiv:hep-th/9704167}}.

\bibitem{Fubini:1985jm}
S.~Fubini and H.~Nicolai, ``{The Octonionic Instanton},''
  \href{http://dx.doi.org/10.1016/0370-2693(85)91589-8}{{\em Phys. Lett. B}
  {\bfseries 155} (1985) 369--372}.

\bibitem{Popov:1992sy}
A.~D. Popov, ``{Spherically symmetric solutions of Yang-Mills equations in D =
  7 for arbitrary gauge group},''
  \href{http://dx.doi.org/10.1209/0295-5075/17/1/005}{{\em Europhys. Lett.}
  {\bfseries 17} (1992) 23--26}.

\bibitem{Ivanova:1992nj}
T.~A. Ivanova and A.~D. Popov, ``{Selfdual Yang-Mills fields in d = 7, 8,
  octonions and Ward equations},''
  \href{http://dx.doi.org/10.1007/BF00402672}{{\em Lett. Math. Phys.}
  {\bfseries 24} (1992) 85--92}.

\bibitem{Harland:2010ojo}
D.~Harland, T.~A. Ivanova, O.~Lechtenfeld, and A.~D. Popov, ``{Yang-Mills flows
  on nearly Kahler manifolds and G(2)-instantons},''
  \href{http://dx.doi.org/10.1007/s00220-010-1115-7}{{\em Commun. Math. Phys.}
  {\bfseries 300} (2010) 185--204},
  \href{http://arxiv.org/abs/0909.2730}{{\ttfamily arXiv:0909.2730 [hep-th]}}.

\bibitem{Green:1987mn}
M.~B. Green, J.~H. Schwarz, and E.~Witten, {\em {SUPERSTRING THEORY. VOL. 2:
  LOOP AMPLITUDES, ANOMALIES AND PHENOMENOLOGY}}.
\newblock ., 7, 1988.

\bibitem{Townsend:1983kk}
P.~K. Townsend and P.~van Nieuwenhuizen, ``{GAUGED SEVEN-DIMENSIONAL
  SUPERGRAVITY},'' \href{http://dx.doi.org/10.1016/0370-2693(83)91230-3}{{\em
  Phys. Lett. B} {\bfseries 125} (1983) 41--46}.

\bibitem{Tian:2000fu}
G.~Tian, ``{Gauge theory and calibrated geometry. 1.},''
  \href{http://dx.doi.org/10.2307/121116}{{\em Annals Math.} {\bfseries 151}
  (2000) 193--268}, \href{http://arxiv.org/abs/math/0010015}{{\ttfamily
  arXiv:math/0010015}}.

\bibitem{joyce2000compact}
D.~Joyce, {\em Compact Manifolds with Special Holonomy}.
\newblock Oxford mathematical monographs. Oxford University Press, 2000.

\bibitem{Acharya:2004qe}
B.~S. Acharya and S.~Gukov, ``{M theory and singularities of exceptional
  holonomy manifolds},''
  \href{http://dx.doi.org/10.1016/j.physrep.2003.10.017}{{\em Phys. Rept.}
  {\bfseries 392} (2004) 121--189},
\href{http://arxiv.org/abs/hep-th/0409191}{{\ttfamily arXiv:hep-th/0409191}}.

\bibitem{BryantMetrics}
R.~Bryant, \href{http://dx.doi.org/10.1007/BFb0084595}{{\em {Metrics with
  holonomy G2 or spin (7)}}}.
\newblock {Hirzebruch F., Schwermer J., Suter S. (eds) Arbeitstagung Bonn 1984.
  Lecture Notes in Mathematics, vol 1111}. Springer, Berlin, Heidelberg., 1985.

\bibitem{Doanthesis}
A.~Doan, ``{Monopoles and Fueter sections on three-manifolds},'' {\em Ph.D
  Thesis} (2019) .

\bibitem{Walpuskithesis}
T.~Walpuski, ``{Gauge theory on $G_2-$manifolds},'' {\em Ph.D Thesis} (2013) .

\bibitem{driscoll2021deformations}
J.~Driscoll, ``Deformations of asymptotically conical $g_2$-instantons,'' 2021.

\bibitem{Lossev:1997bz}
A.~Lossev, N.~Nekrasov, and S.~L. Shatashvili, ``{Testing Seiberg-Witten
  solution},'' {\em NATO Sci. Ser. C} {\bfseries 520} (1999) 359--372,
  \href{http://arxiv.org/abs/hep-th/9801061}{{\ttfamily arXiv:hep-th/9801061}}.

\bibitem{Nekrasov:2002qd}
N.~A. Nekrasov, ``{Seiberg-Witten prepotential from instanton counting},''
  \href{http://dx.doi.org/10.4310/ATMP.2003.v7.n5.a4}{{\em Adv. Theor. Math.
  Phys.} {\bfseries 7} no.~5, (2003) 831--864},
  \href{http://arxiv.org/abs/hep-th/0206161}{{\ttfamily arXiv:hep-th/0206161}}.

\bibitem{Cherkis:2009hpw}
S.~A. Cherkis, ``{Moduli Spaces of Instantons on the Taub-NUT Space},''
  \href{http://dx.doi.org/10.1007/s00220-009-0863-8}{{\em Commun. Math. Phys.}
  {\bfseries 290} (2009) 719--736},
  \href{http://arxiv.org/abs/0805.1245}{{\ttfamily arXiv:0805.1245 [hep-th]}}.

\bibitem{Cherkis:2011ee}
S.~A. Cherkis, C.~O'Hara, and C.~Saemann, ``{Super Yang-Mills Theory with
  Impurity Walls and Instanton Moduli Spaces},''
  \href{http://dx.doi.org/10.1103/PhysRevD.83.126009}{{\em Phys. Rev. D}
  {\bfseries 83} (2011) 126009},
  \href{http://arxiv.org/abs/1103.0042}{{\ttfamily arXiv:1103.0042 [hep-th]}}.

\bibitem{Kachru:2001je}
S.~Kachru and J.~McGreevy, ``{M theory on manifolds of G(2) holonomy and type
  IIA orientifolds},''
  \href{http://dx.doi.org/10.1088/1126-6708/2001/06/027}{{\em JHEP} {\bfseries
  06} (2001) 027}, \href{http://arxiv.org/abs/hep-th/0103223}{{\ttfamily
  arXiv:hep-th/0103223}}.

\bibitem{Ooguri:1997ih}
H.~Ooguri and C.~Vafa, ``{Geometry of N=1 dualities in four-dimensions},''
  \href{http://dx.doi.org/10.1016/S0550-3213(97)00304-0}{{\em Nucl. Phys. B}
  {\bfseries 500} (1997) 62--74},
  \href{http://arxiv.org/abs/hep-th/9702180}{{\ttfamily arXiv:hep-th/9702180}}.

\bibitem{Lawrence:1998ja}
A.~E. Lawrence, N.~Nekrasov, and C.~Vafa, ``{On conformal field theories in
  four-dimensions},''
  \href{http://dx.doi.org/10.1016/S0550-3213(98)00495-7}{{\em Nucl. Phys. B}
  {\bfseries 533} (1998) 199--209},
  \href{http://arxiv.org/abs/hep-th/9803015}{{\ttfamily arXiv:hep-th/9803015}}.

\bibitem{Klebanov:1998hh}
I.~R. Klebanov and E.~Witten, ``{Superconformal field theory on three-branes at
  a Calabi-Yau singularity},''
  \href{http://dx.doi.org/10.1016/S0550-3213(98)00654-3}{{\em Nucl. Phys. B}
  {\bfseries 536} (1998) 199--218},
  \href{http://arxiv.org/abs/hep-th/9807080}{{\ttfamily arXiv:hep-th/9807080}}.

\bibitem{Beasley:1999uz}
C.~Beasley, B.~R. Greene, C.~I. Lazaroiu, and M.~R. Plesser, ``{D3-branes on
  partial resolutions of Abelian quotient singularities of Calabi-Yau
  threefolds},'' \href{http://dx.doi.org/10.1016/S0550-3213(99)00646-X}{{\em
  Nucl. Phys. B} {\bfseries 566} (2000) 599--640},
  \href{http://arxiv.org/abs/hep-th/9907186}{{\ttfamily arXiv:hep-th/9907186}}.

\bibitem{Hanany:2001py}
A.~Hanany and A.~Iqbal, ``{Quiver theories from D6 branes via mirror
  symmetry},'' \href{http://dx.doi.org/10.1088/1126-6708/2002/04/009}{{\em
  JHEP} {\bfseries 04} (2002) 009},
  \href{http://arxiv.org/abs/hep-th/0108137}{{\ttfamily arXiv:hep-th/0108137}}.

\bibitem{Feng:2005gw}
B.~Feng, Y.-H. He, K.~D. Kennaway, and C.~Vafa, ``{Dimer models from mirror
  symmetry and quivering amoebae},''
  \href{http://dx.doi.org/10.4310/ATMP.2008.v12.n3.a2}{{\em Adv. Theor. Math.
  Phys.} {\bfseries 12} no.~3, (2008) 489--545},
  \href{http://arxiv.org/abs/hep-th/0511287}{{\ttfamily arXiv:hep-th/0511287}}.

\bibitem{Strominger:1996it}
A.~Strominger, S.-T. Yau, and E.~Zaslow, ``{Mirror symmetry is T duality},''
  \href{http://dx.doi.org/10.1016/0550-3213(96)00434-8}{{\em Nucl. Phys.}
  {\bfseries B479} (1996) 243--259},
\href{http://arxiv.org/abs/hep-th/9606040}{{\ttfamily arXiv:hep-th/9606040}}.

\bibitem{Douglas:1996sw}
M.~R. Douglas and G.~W. Moore, ``{D-branes, quivers, and ALE instantons},''
  \href{http://arxiv.org/abs/hep-th/9603167}{{\ttfamily arXiv:hep-th/9603167}}.

\bibitem{Katz:1996xe}
S.~H. Katz and C.~Vafa, ``{Matter from geometry},''
  \href{http://dx.doi.org/10.1016/S0550-3213(97)00280-0}{{\em Nucl. Phys. B}
  {\bfseries 497} (1997) 146--154},
  \href{http://arxiv.org/abs/hep-th/9606086}{{\ttfamily arXiv:hep-th/9606086}}.

\bibitem{Hori:2000kt}
K.~Hori and C.~Vafa, ``{Mirror symmetry},''
\href{http://arxiv.org/abs/hep-th/0002222}{{\ttfamily arXiv:hep-th/0002222}}.

\bibitem{Harvey:1999as}
J.~A. Harvey and G.~W. Moore, ``{Superpotentials and membrane instantons},''
\href{http://arxiv.org/abs/hep-th/9907026}{{\ttfamily arXiv:hep-th/9907026}}.

\bibitem{Pantev:2009de}
T.~Pantev and M.~Wijnholt, ``{Hitchin's Equations and M-Theory
  Phenomenology},''
  \href{http://dx.doi.org/10.1016/j.geomphys.2011.02.014}{{\em J. Geom. Phys.}
  {\bfseries 61} (2011) 1223--1247},
\href{http://arxiv.org/abs/0905.1968}{{\ttfamily arXiv:0905.1968 [hep-th]}}.

\bibitem{Intriligator:1994jr}
K.~A. Intriligator, R.~G. Leigh, and N.~Seiberg, ``{Exact superpotentials in
  four-dimensions},'' \href{http://dx.doi.org/10.1103/PhysRevD.50.1092}{{\em
  Phys. Rev. D} {\bfseries 50} (1994) 1092--1104},
  \href{http://arxiv.org/abs/hep-th/9403198}{{\ttfamily arXiv:hep-th/9403198}}.

\bibitem{Beasley:2004ys}
C.~Beasley and E.~Witten, ``{New instanton effects in supersymmetric QCD},''
  \href{http://dx.doi.org/10.1088/1126-6708/2005/01/056}{{\em JHEP} {\bfseries
  01} (2005) 056}, \href{http://arxiv.org/abs/hep-th/0409149}{{\ttfamily
  arXiv:hep-th/0409149}}.

\bibitem{Apruzzi:2018oge}
F.~Apruzzi, J.~J. Heckman, D.~R. Morrison, and L.~Tizzano, ``{4D Gauge Theories
  with Conformal Matter},''
  \href{http://dx.doi.org/10.1007/JHEP09(2018)088}{{\em JHEP} {\bfseries 09}
  (2018) 088}, \href{http://arxiv.org/abs/1803.00582}{{\ttfamily
  arXiv:1803.00582 [hep-th]}}.

\bibitem{Aganagic:2003db}
M.~Aganagic, A.~Klemm, M.~Marino, and C.~Vafa, ``{The Topological vertex},''
  \href{http://dx.doi.org/10.1007/s00220-004-1162-z}{{\em Commun. Math. Phys.}
  {\bfseries 254} (2005) 425--478},
  \href{http://arxiv.org/abs/hep-th/0305132}{{\ttfamily arXiv:hep-th/0305132}}.

\bibitem{Iqbal:2007ii}
A.~Iqbal, C.~Kozcaz, and C.~Vafa, ``{The Refined topological vertex},''
  \href{http://dx.doi.org/10.1088/1126-6708/2009/10/069}{{\em JHEP} {\bfseries
  10} (2009) 069}, \href{http://arxiv.org/abs/hep-th/0701156}{{\ttfamily
  arXiv:hep-th/0701156}}.

\bibitem{Witten:1982df}
E.~Witten, ``{Constraints on Supersymmetry Breaking},''
  \href{http://dx.doi.org/10.1016/0550-3213(82)90071-2}{{\em Nucl. Phys. B}
  {\bfseries 202} (1982) 253}.

\bibitem{Witten:1994tz}
E.~Witten, ``{Sigma models and the ADHM construction of instantons},''
  \href{http://dx.doi.org/10.1016/0393-0440(94)00047-8}{{\em J. Geom. Phys.}
  {\bfseries 15} (1995) 215--226},
  \href{http://arxiv.org/abs/hep-th/9410052}{{\ttfamily arXiv:hep-th/9410052}}.

\bibitem{Douglas:1995bn}
M.~R. Douglas, ``{Branes within branes},'' {\em NATO Sci. Ser. C} {\bfseries
  520} (1999) 267--275, \href{http://arxiv.org/abs/hep-th/9512077}{{\ttfamily
  arXiv:hep-th/9512077}}.

\bibitem{Douglas:1996uz}
M.~R. Douglas, ``{Gauge fields and D-branes},''
  \href{http://dx.doi.org/10.1016/S0393-0440(97)00024-7}{{\em J. Geom. Phys.}
  {\bfseries 28} (1998) 255--262},
  \href{http://arxiv.org/abs/hep-th/9604198}{{\ttfamily arXiv:hep-th/9604198}}.

\bibitem{Hwang:2014uwa}
C.~Hwang, J.~Kim, S.~Kim, and J.~Park, ``{General instanton counting and 5d
  SCFT},'' \href{http://dx.doi.org/10.1007/JHEP07(2015)063}{{\em JHEP}
  {\bfseries 07} (2015) 063}, \href{http://arxiv.org/abs/1406.6793}{{\ttfamily
  arXiv:1406.6793 [hep-th]}}. [Addendum: JHEP 04, 094 (2016)].

\bibitem{Hori:2014tda}
K.~Hori, H.~Kim, and P.~Yi, ``{Witten Index and Wall Crossing},''
  \href{http://dx.doi.org/10.1007/JHEP01(2015)124}{{\em JHEP} {\bfseries 01}
  (2015) 124}, \href{http://arxiv.org/abs/1407.2567}{{\ttfamily arXiv:1407.2567
  [hep-th]}}.

\bibitem{Kim:2016qqs}
H.-C. Kim, ``{Line defects and 5d instanton partition functions},''
  \href{http://dx.doi.org/10.1007/JHEP03(2016)199}{{\em JHEP} {\bfseries 03}
  (2016) 199}, \href{http://arxiv.org/abs/1601.06841}{{\ttfamily
  arXiv:1601.06841 [hep-th]}}.

\bibitem{Jeffrey:aa}
L.~Jeffrey and F.~Kirwan, ``{Localization for nonabelian group actions},''
  \href{http://arxiv.org/abs/9307001}{{\ttfamily arXiv:9307001 [alg-geom]}}.

\bibitem{Benini:2013xpa}
F.~Benini, R.~Eager, K.~Hori, and Y.~Tachikawa, ``{Elliptic Genera of 2d
  ${\mathcal{N}}$ = 2 Gauge Theories},''
  \href{http://dx.doi.org/10.1007/s00220-014-2210-y}{{\em Commun. Math. Phys.}
  {\bfseries 333} no.~3, (2015) 1241--1286},
  \href{http://arxiv.org/abs/1308.4896}{{\ttfamily arXiv:1308.4896 [hep-th]}}.

\bibitem{Haouzi:2019jzk}
N.~Haouzi and C.~Koz\c{c}az, ``{Supersymmetric Wilson Loops, Instantons, and
  Deformed ${\cal W}$-Algebras},''
  \href{http://arxiv.org/abs/1907.03838}{{\ttfamily arXiv:1907.03838
  [hep-th]}}.

\bibitem{Hanany:1996ie}
A.~Hanany and E.~Witten, ``{Type IIB superstrings, BPS monopoles, and
  three-dimensional gauge dynamics},''
  \href{http://dx.doi.org/10.1016/S0550-3213(97)00157-0}{{\em Nucl. Phys. B}
  {\bfseries 492} (1997) 152--190},
  \href{http://arxiv.org/abs/hep-th/9611230}{{\ttfamily arXiv:hep-th/9611230}}.

\bibitem{Costello:2018txb}
K.~Costello and J.~Yagi, ``{Unification of integrability in supersymmetric
  gauge theories},'' \href{http://arxiv.org/abs/1810.01970}{{\ttfamily
  arXiv:1810.01970 [hep-th]}}.

\bibitem{Bashkirov:2010kz}
D.~Bashkirov and A.~Kapustin, ``{Supersymmetry enhancement by monopole
  operators},'' \href{http://dx.doi.org/10.1007/JHEP05(2011)015}{{\em JHEP}
  {\bfseries 05} (2011) 015}, \href{http://arxiv.org/abs/1007.4861}{{\ttfamily
  arXiv:1007.4861 [hep-th]}}.

\bibitem{Yagi:2014toa}
J.~Yagi, ``{$\Omega$-deformation and quantization},''
  \href{http://dx.doi.org/10.1007/JHEP08(2014)112}{{\em JHEP} {\bfseries 08}
  (2014) 112}, \href{http://arxiv.org/abs/1405.6714}{{\ttfamily arXiv:1405.6714
  [hep-th]}}.

\bibitem{Bullimore:2016nji}
M.~Bullimore, T.~Dimofte, D.~Gaiotto, and J.~Hilburn, ``{Boundaries, Mirror
  Symmetry, and Symplectic Duality in 3d $\mathcal{N}=4$ Gauge Theory},''
  \href{http://dx.doi.org/10.1007/JHEP10(2016)108}{{\em JHEP} {\bfseries 10}
  (2016) 108}, \href{http://arxiv.org/abs/1603.08382}{{\ttfamily
  arXiv:1603.08382 [hep-th]}}.

\bibitem{Bullimore:2016hdc}
M.~Bullimore, T.~Dimofte, D.~Gaiotto, J.~Hilburn, and H.-C. Kim, ``{Vortices
  and Vermas},'' \href{http://dx.doi.org/10.4310/ATMP.2018.v22.n4.a1}{{\em Adv.
  Theor. Math. Phys.} {\bfseries 22} (2018) 803--917},
  \href{http://arxiv.org/abs/1609.04406}{{\ttfamily arXiv:1609.04406
  [hep-th]}}.

\bibitem{Mezei:2017kmw}
M.~Mezei, S.~S. Pufu, and Y.~Wang, ``{A 2d/1d Holographic Duality},''
  \href{http://arxiv.org/abs/1703.08749}{{\ttfamily arXiv:1703.08749
  [hep-th]}}.

\bibitem{Dedushenko:2016jxl}
M.~Dedushenko, S.~S. Pufu, and R.~Yacoby, ``{A one-dimensional theory for Higgs
  branch operators},'' \href{http://dx.doi.org/10.1007/JHEP03(2018)138}{{\em
  JHEP} {\bfseries 03} (2018) 138},
  \href{http://arxiv.org/abs/1610.00740}{{\ttfamily arXiv:1610.00740
  [hep-th]}}.

\bibitem{Gaiotto:2020vqj}
D.~Gaiotto and J.~Abajian, ``{Twisted M2 brane holography and sphere
  correlation functions},'' \href{http://arxiv.org/abs/2004.13810}{{\ttfamily
  arXiv:2004.13810 [hep-th]}}.

\bibitem{Oh:2020hph}
J.~Oh and Y.~Zhou, ``{Feynman diagrams and $\Omega$-deformed M-theory},''
  \href{http://dx.doi.org/10.21468/SciPostPhys.10.2.029}{{\em SciPost Phys.}
  {\bfseries 10} no.~2, (2021) 029},
  \href{http://arxiv.org/abs/2002.07343}{{\ttfamily arXiv:2002.07343
  [hep-th]}}.

\bibitem{Hatsuda:2021oxa}
Y.~Hatsuda and T.~Okazaki, ``{Fermi-gas correlators of ADHM theory and triality
  symmetry},'' \href{http://arxiv.org/abs/2107.01924}{{\ttfamily
  arXiv:2107.01924 [hep-th]}}.

\bibitem{Berkooz:1996km}
M.~Berkooz, M.~R. Douglas, and R.~G. Leigh, ``{Branes intersecting at
  angles},'' \href{http://dx.doi.org/10.1016/S0550-3213(96)00452-X}{{\em Nucl.
  Phys. B} {\bfseries 480} (1996) 265--278},
  \href{http://arxiv.org/abs/hep-th/9606139}{{\ttfamily arXiv:hep-th/9606139}}.

\bibitem{Witten:2000mf}
E.~Witten, ``{BPS Bound states of D0 - D6 and D0 - D8 systems in a B field},''
  \href{http://dx.doi.org/10.1088/1126-6708/2002/04/012}{{\em JHEP} {\bfseries
  04} (2002) 012}, \href{http://arxiv.org/abs/hep-th/0012054}{{\ttfamily
  arXiv:hep-th/0012054}}.

\bibitem{Seiberg:1999vs}
N.~Seiberg and E.~Witten, ``{String theory and noncommutative geometry},''
  \href{http://dx.doi.org/10.1088/1126-6708/1999/09/032}{{\em JHEP} {\bfseries
  09} (1999) 032}, \href{http://arxiv.org/abs/hep-th/9908142}{{\ttfamily
  arXiv:hep-th/9908142}}.

\bibitem{Gates:1983nr}
S.~J. Gates, M.~T. Grisaru, M.~Rocek, and W.~Siegel, {\em {Superspace Or One
  Thousand and One Lessons in Supersymmetry}}, vol.~58 of {\em Frontiers in
  Physics}.
\newblock ., 1983.
\newblock \href{http://arxiv.org/abs/hep-th/0108200}{{\ttfamily
  arXiv:hep-th/0108200}}.

\bibitem{Nekrasov:2016gud}
N.~Nekrasov and N.~S. Prabhakar, ``{Spiked Instantons from Intersecting
  D-branes},'' \href{http://dx.doi.org/10.1016/j.nuclphysb.2016.11.014}{{\em
  Nucl. Phys. B} {\bfseries 914} (2017) 257--300},
  \href{http://arxiv.org/abs/1611.03478}{{\ttfamily arXiv:1611.03478
  [hep-th]}}.

\bibitem{Abouelsaood:1986gd}
A.~Abouelsaood, C.~G. Callan, Jr., C.~R. Nappi, and S.~A. Yost, ``{Open Strings
  in Background Gauge Fields},''
  \href{http://dx.doi.org/10.1016/0550-3213(87)90164-7}{{\em Nucl. Phys. B}
  {\bfseries 280} (1987) 599--624}.

\bibitem{Balasubramanian:1996uc}
V.~Balasubramanian and R.~G. Leigh, ``{D-branes, moduli and supersymmetry},''
  \href{http://dx.doi.org/10.1103/PhysRevD.55.6415}{{\em Phys. Rev. D}
  {\bfseries 55} (1997) 6415--6422},
  \href{http://arxiv.org/abs/hep-th/9611165}{{\ttfamily arXiv:hep-th/9611165}}.

\bibitem{Pomoni:2021hkn}
E.~Pomoni, W.~Yan, and X.~Zhang, ``{Tetrahedron instantons},''
  \href{http://arxiv.org/abs/2106.11611}{{\ttfamily arXiv:2106.11611
  [hep-th]}}.

\bibitem{Nakajima:1994}
H.~Nakajima, ``{Instantons on ALE Spaces, Quiver Varieties, and Kac-Moody
  Algebras},'' \href{http://dx.doi.org/10.1215/S0012-7094-94-07613-8}{{\em Duke
  Mathematical Journal} {\bfseries 76} (1994) 365?416}.

\bibitem{Crawley-Boevey:2001}
W.~Crawley-Boevey, ``{Geometry of the Moment Map for Representations of
  Quivers},'' \href{http://dx.doi.org/10.1023/A:1017558904030}{{\em Compositio
  Mathematica} {\bfseries 126} (2001) 257?93}.

\bibitem{Maulik:2012}
D.~Maulik and A.~Okounkov, ``Quantum groups and quantum cohomology,'' 2018.

\bibitem{Fantechi:2005}
Fantechi, Gottsche, Illusie, Kleiman, Nitsure, and Vistoli, {\em {Fundamental
  Algebraic Geometry: Grothendieck?s FGA Explained}}.
\newblock American Mathematical Society., 2005.

\bibitem{Carlsson:2008}
E.~Carlsson and A.~Okounkov, ``Exts and vertex operators,''
  \href{http://dx.doi.org/10.1215/00127094-1593380}{{\em Duke Mathematical
  Journal} {\bfseries 161} no.~9, (Jun, 2012) .}

\bibitem{Nevins:2016}
K.~McGerty and T.~Nevins, ``Kirwan surjectivity for quiver varieties,''
  \href{http://dx.doi.org/10.1007/s00222-017-0765-x}{{\em Inventiones
  mathematicae} {\bfseries 212} no.~1, (Nov, 2017) 161?187}.

\bibitem{Baranovsky:2001}
V.~Baranovsky, ``{The Variety of Pairs of Commuting Nilpotent Matrices is
  Irreducible},'' \href{http://dx.doi.org/10.1007/BF01236059}{{\em
  Transformation Groups} {\bfseries 6} (2001) 3--8}.

\end{thebibliography}\endgroup

\end{document}